\documentclass[envcountsect, envcountsame, runningheads, final]{llncs}
\pdfoutput=1

\usepackage{hyperref}
\usepackage{amsmath, amssymb}
\usepackage{etex}
\usepackage[utf8]{inputenc}
\usepackage{bbm,xspace,wrapfig}

\usepackage{tikz}
\usetikzlibrary{backgrounds, calc, fit, cd}

\usepackage{textcomp}
\usepackage{microtype}
\usepackage{times}
\usepackage{relsize}
\usepackage{enumerate}
\usepackage{ifdraft}
\usepackage{mathtools}

\hypersetup{
    colorlinks=true,
    linkcolor=black,
    citecolor=black,
    filecolor=black,
    urlcolor=black,
}

\usepackage{ifthen}
\newboolean{trimmedmargins} \setboolean{trimmedmargins}{false}
\usepackage{substr}
\IfSubStringInString{\detokenize{trimmed}}{\jobname}{
\setboolean{trimmedmargins}{true}
}{}

\iftrimmedmargins
\usepackage[marginparwidth=2cm,nohead,nofoot,
            headsep = 1cm,
            footskip = 1cm,
            text={12.2cm,19.3cm},
            paperwidth=16.2cm,
            paperheight=23.3cm,
            marginparsep=1mm,
            marginparwidth=1.8cm,
            footskip=1cm,
            centering
            ]{geometry}
\else
\fi

\usepackage{newunicodechar}
\newunicodechar{×}{\ensuremath{\times}}
\newunicodechar{≤}{\ensuremath{\le}}
\newunicodechar{⊆}{\ensuremath{\subseteq}}

\newcommand{\Nat}{\mathbb{N}}
\newcommand{\litlang}{L_0}
 
\newcommand{\barlang}{L_\alpha}
\newcommand{\datalang}{D}

\newcommand{\Expbar}{\ensuremath{\operatorname{\sf RBExp}}}
\newcommand{\longmid}{\hspace{0.9ex}\smash{\rule[-1.0ex]{0.41pt}{3.2ex}}\hspace{0.9ex}}
\newcommand{\midmid}{\hspace{0.2ex}{\rule[-0.1ex]{0.6pt}{1.65ex}}\hspace{0.2ex}}
\newcommand{\scriptmidmid}{\hspace{0.2ex}{\rule[-0.1ex]{0.6pt}{1.1ex}}\hspace{0.2ex}}
\newcommand{\subscriptmidmid}{\hspace{0.2ex}{\rule[-0.1ex]{0.6pt}{0.9ex}}\hspace{0.2ex}}

\newcommand{\op}[1]{\ensuremath{\operatorname{\sf #1}}}
\newcommand{\degree}{\op{deg}}
\newcommand{\unbar}{\op{ub}}
\newcommand{\dataord}{\sqsubseteq}
\newcommand{\datasup}{\sqsupseteq}
\newcommand{\ExpSpace}{\mbox{\textsc{ExpSpace}}\xspace}
\newcommand{\NExpSpace}{\mbox{\textsc{NExpSpace}}\xspace}
\newcommand{\PSpace}{\mbox{\textsc{PSpace}}\xspace}
\newcommand{\eq}[2]{{\langle {#1}\rangle {#2}}}

\newcommand{\freesuc}{\mathsf{fsuc}} 
\newcommand{\boundsuc}{\mathsf{bsuc}} 
 
\newcommand{\partialto}{\rightharpoonup} 
\newcommand{\newletter}[1]{{\midmid}#1}
\newcommand{\scriptnew}[1]{{\scriptmidmid}#1}
\newcommand{\subscriptnew}[1]{{\subscriptmidmid}#1}
\newcommand{\FN}{\mathsf{FN}}

\newcommand{\Names}{\mathbb{A}}
\newcommand{\A}{\ensuremath{\mathbb{A}}\xspace}
\newcommand{\alphaeq}{\equiv_\alpha}
\newcommand{\powufs}{\pow_{\mathsf{ufs}}}
\newcommand{\powfs}{\pow_{\mathsf{fs}}}
\newcommand{\powf}{\pow_{\omega}}

\newcommand{\BoolCombis}{\ensuremath{{\color{black!60}\mathcal{B}}}\xspace}
\newcommand{\X}{\ensuremath{\mathord{\tt X}}}
\newcommand{\branchleft}{\ensuremath{\operatorname{\mathrlap{\ \vert}\!\!<}}}
\newcommand{\branchright}{\ensuremath{\operatorname{\mathrlap{\,\vert}\!\!>}}}
\newcommand{\nfresh}[1]{%
\ensuremath{\underline{%
{\hspace{2mm}}{\#^{#1}}{\hspace{2mm}}
}}}

\makeatletter
\newcommand{\Setfs}{\@ifnextchar_\SetfsWithAtoms{\ensuremath{\Set_{\mathsf{fs}}}}}
\def\SetfsWithAtoms_#1{\ensuremath{\Set_{\mathsf{fs},#1}}}
\makeatother

\newcommand{\myparagraph}[1]{\medskip\par\noindent\textbf{\textbf{#1}}\hspace{6pt}}

\newcommand{\fresh}{\mathbin{\#}}
\newcommand{\supp}{\mathsf{supp}}
\newcommand{\fix}{\mathop{\mathsf{fix}}\xspace}
\newcommand{\Fix}{\mathop{\mathsf{Fix}}\xspace}
\newcommand{\trans}[1]{\xrightarrow{#1}}

\newcommand{\Nom}{\ensuremath{\mathsf{Nom}}\xspace}

\tikzstyle{shiftarr}=[
        rounded corners,%
        to path={--([#1]\tikztostart.center)
                     -- ([#1]\tikztotarget.center) \tikztonodes
                     -- (\tikztotarget)},
]
\newcommand{\descto}[3][]{
    \arrow[draw=none,
           to path={
            (\tikztostart.center)
            -- (\tikztotarget.center) \tikztonodes
           },
          ]{#2}[description,#1]{#3}
}

\usepackage[notref,notcite]{showkeys}

\usepackage{seqsplit}
\iftrimmedmargins

\else
\fi

\usepackage[author=anonymous,nomargin,marginclue,footnote,status=final]{fixme}
\FXRegisterAuthor{ls}{als}{LS}
\FXRegisterAuthor{sm}{asm}{SM}
\FXRegisterAuthor{dk}{adk}{DK}
\FXRegisterAuthor{tw}{atw}{TW}

\newlength{\myboxwidth}
	{\end{center}\end{figure}}

\newcounter{blubber}

\newenvironment{myenumerate}
{\begin{enumerate}
\setlength{\itemsep}{0pt}
    \setlength{\leftmargin}{0pt}
    \setlength{\itemindent}{0pt}
}
{\end{enumerate}}

\newenvironment{rmenumerate}{\begin{enumerate}}{\end{enumerate}}

\newenvironment{myrmenumerate}
{\begin{rmenumerate}
\setlength{\itemsep}{0pt}
    \setlength{\leftmargin}{0pt}
    \setlength{\itemindent}{0pt}
}
{\end{rmenumerate}}

\usepackage{enumitem}

\def\Id{\mathit{Id}}

\renewcommand{\theta}{\vartheta}

\newcommand{\by}[1]{\text{(#1)}}

\newcommand{\id}{\mathsf{id}}

\newcommand{\BC}{\mathbf{C}}

\newcommand{\pow}{\mathcal{P}}

\newcommand{\Set}{\ensuremath{\mathsf{Set}}\xspace}

\def\epito{\twoheadrightarrow}

\newcommand{\takeout}[1]{\empty}

\newcommand\size[1]{\sharp #1}

\makeatletter
\newcommand*{\@old@slash}{}\let\@old@slash\slash
\def\slash{\relax\ifmmode\delimiter"502F30E\mathopen{}\else\@old@slash\fi}
\makeatother

\spnewtheorem{thm}[theorem]{Theorem}{\bfseries}{\itshape}
\spnewtheorem{cor}[theorem]{Corollary}{\bfseries}{\itshape}
\spnewtheorem{lem}[theorem]{Lemma}{\bfseries}{\itshape}
\spnewtheorem{lemdefn}[theorem]{Lemma and Definition}{\bfseries}{\itshape}
\spnewtheorem{propn}[theorem]{Proposition}{\bfseries}{\itshape}
\spnewtheorem{obs}[theorem]{Observation}{\bfseries}{\upshape}
\spnewtheorem{fact}[theorem]{Fact}{\bfseries}{\upshape}
\spnewtheorem{thmdefn}[theorem]{Theorem and Definition}{\bfseries}{\itshape}
\spnewtheorem{propdefn}[theorem]{Proposition and Definition}{\bfseries}{\itshape}

\spnewtheorem{defn}[theorem]{Definition}{\bfseries}{\upshape}
\spnewtheorem{construction}[theorem]{Construction}{\bfseries}{\upshape}
\spnewtheorem{rem}[theorem]{Remark}{\bfseries}{\upshape}
\spnewtheorem{expl}[theorem]{Example}{\bfseries}{\upshape}
\spnewtheorem{assn}[theorem]{Assumption}{\bfseries}{\upshape}
\spnewtheorem{conv}[theorem]{Convention}{\bfseries}{\upshape}
\spnewtheorem{notn}[theorem]{Notation}{\bfseries}{\upshape}
\spnewtheorem{open}[theorem]{Problem}{\bfseries}{\upshape}

\newcommand{\qedhere}{\qed}

\title{Nominal Automata with Name Binding}
\author{
Lutz Schr{\"o}der\inst{1} \and
Dexter Kozen\inst{2} \and
Stefan Milius\inst{1} \and
Thorsten Wi\ss{}mann\inst{1}
}
\institute{Friedrich-Alexander-Universität Erlangen-N\"urnberg
\and Cornell University}
\authorrunning{L.~Schr{\"o}der, D.~Kozen, S.~Milius, and T.~Wi\ss{}mann}

\begin{document}
\maketitle
\begin{abstract}
  Nominal sets are a convenient setting for languages over infinite
  alphabets, i.e.~data languages.  We introduce an automaton
  model over nominal sets, \emph{regular nondeterministic nominal
    automata (RNNA)}, which have a natural coalgebraic definition using
  abstraction sets to capture transitions that read a fresh letter
  from the input word. We prove a Kleene theorem for RNNAs w.r.t.\ a
  simple expression language that extends \emph{nominal Kleene algebra
    (NKA)} with unscoped name binding, thus remedying the known
  failure of the expected Kleene theorem for NKA itself.  We analyse
  RNNAs under two notions of freshness: \emph{global} and
  \emph{local}. Under global freshness, RNNAs turn out to be
  equivalent to session automata, and as such have a decidable
  inclusion problem. Under \emph{local} freshness, RNNAs retain a
  decidable inclusion problem, and translate into register automata.
  We thus obtain decidability of inclusion for a reasonably expressive
  class of nondeterministic register automata, with no bound on the
  number of registers.
\end{abstract}

\section{Introduction}

\noindent \emph{Data languages} are languages over infinite alphabets,
regarded as modeling the communication of values from infinite data
types such as nonces~\cite{KurtzEA07}, channel
names~\cite{Hennessy02}, process identifiers~\cite{BolligEA14},
URL's~\cite{BieleckiEA02}, or data values in XML documents
(see~\cite{NevenEA04} for a summary). There is a plethora of automata
models for data languages~\cite{Bojanczyk10,GrumbergEA10,Segoufin06},
which can be classified along several axes. One line of division is
between models that use explicit registers and have a finite-state
description (generating infinite configuration spaces) on the one
hand, and more abstract models phrased as automata over nominal
sets~\cite{Pitts13} on the other hand. The latter have infinitely many
states but are typically required to be \emph{orbit-finite}, i.e.\ to
have only finitely many states up to renaming implicitly stored
letters. There are correspondences between the two styles;
e.g.~Boja\'nczyk, Klin, and Lasota's \emph{nondeterministic
  orbit-finite automata (NOFA)}~\cite{BojanczykEA14} are equivalent to
Kaminski and Francez' \emph{register automata
  (RAs)}~\cite{KaminskiFrancez94} (originally called finite memory
automata), more precisely to RAs with nondeterministic
reassignment~\cite{KaminskiZeitlin10}. A second distinction concerns
notions of freshness: \emph{global freshness} requires that the next
letter to be consumed has not been seen before, while \emph{local
  freshness} postulates only that the next letter is distinct from the
(boundedly many) letters currently stored in the registers.

Although local freshness looks computationally more natural,
nondeterministic automata models (typically more expressive than
deterministic ones~\cite{KozenEA15}) featuring local freshness tend to
have undecidable inclusion problems. This includes RAs (unless
restricted to two registers~\cite{KaminskiFrancez94}) and
NOFAs~\cite{NevenEA04,BojanczykEA14} as well as \emph{variable
  automata}~\cite{GrumbergEA10}. \emph{Finite-state unification-based
  automata (FSUBAs)}~\cite{KaminskiTan06} have a decidable inclusion
problem but do not support freshness. 
Contrastingly, \emph{session automata}, which give up local
freshness in favor of global freshness, have a decidable inclusion
problem~\cite{BolligEA14}.

Another formalism for global freshness is \emph{nominal Kleene algebra
  (NKA)}~\cite{GabbayCiancia11}. It has been shown that a slight
variant of the original NKA semantics satisfies one half of a Kleene
theorem~\cite{KozenEA15}, which states that NKA expressions can be
converted into a species of nondeterministic nominal automata with
explicit \emph{name binding} transitions (the exact definition of
these automata being left implicit in op.~cit.); the converse
direction of the Kleene theorem fails even for deterministic nominal
automata.

Here, we introduce \emph{regular bar expressions (RBEs)}, which differ
from NKA in making name binding dynamically scoped. RBEs are just
regular expressions over an extended alphabet that includes bound
letters, and hence are equivalent to the corresponding
nondeterministic finite automata, which we call \emph{bar NFAs}. We
equip RBEs with two semantics capturing global and local freshness,
respectively, with the latter characterized as a quotient of the
former: For global freshness, we insist on bound names being
instantiated with names not seen before, while in local freshness
semantics, we accept also names that have been read previously but
will not be used again; this is exactly the usual behaviour of
\emph{$\alpha$-equivalence}, and indeed is formally defined using this
notion. Under global freshness, bar NFAs are essentially equivalent to
session automata.

We prove bar NFAs to be expressively equivalent to a nondeterministic
nominal automaton model with name binding, \emph{regular
  nondeterministic nominal automata (\mbox{RNNAs})}. The states of an
RNNA form an orbit-finite nominal set; RNNAs are distinguished from
NOFAs by having both free and bound transitions and being finitely
branching up to $\alpha$-equivalence of free transitions. This is
equivalent to a concise and natural definition of RNNAs as coalgebras
for a functor on nominal sets (however, this coalgebraic view is not
needed to understand our results). From the equivalence of bar NFAs
and RNNAs we obtain (i) a full Kleene theorem relating RNNAs and RBEs;
(ii) a translation of NKA into RBEs, hence, for closed expressions,
into session automata; and (iii) decidability in parametrized \PSpace\
of inclusion for RBEs, implying the known \ExpSpace decidability
result for NKA~\cite{KozenEA15}.

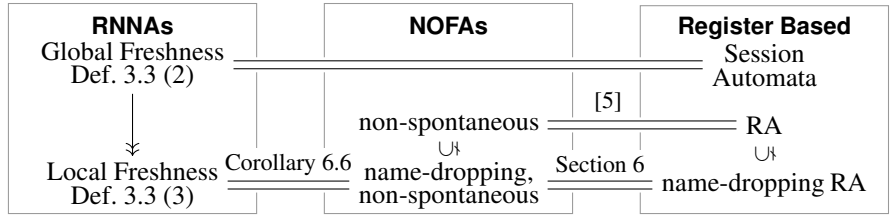
\begin{figure}[t]%
\centering
\setlength{\baselineskip}{2pt}
\begin{tikzpicture}[
    x=42mm,
    y=8mm,
    node distance = 2cm and 1cm,
    class/.style={
        shape=rectangle,
        inner sep=2pt,
        align=center,
        font=\fontsize{10pt}{9pt}\selectfont,
    },
    world/.style={
        shape=rectangle,
        draw=black!40!white,
        inner ysep=0.3em,
        inner xsep=0.0em,
        yshift = 0.7em,
        minimum width=33mm,
    },
    worldlabel/.style={
        yshift=-3mm,
        font=\bf\sffamily,
    },
    ]
    \begin{scope}[every node/.style={ class }]
    \node at (-1,2)      (Global) {Global Freshness\\Def.~\ref{def:barString}~\eqref{eq:globalFreshness}};
    \begin{scope}
    \node at (-1,0)      (Local)  {Local Freshness\\Def.~\ref{def:barString}~\eqref{eq:localFreshness}};
    \node at (1,1) (RA) {RA};
    \node at ($ (RA) - (0,1) $) (NRA) {name-dropping RA};
    \node at ($ (RA) + (0,1) $)      (SessionA) {Session\\Automata};
    \node at (0,0)                       (NameDropping) {name-dropping, \\ non-spontaneous};
    \node at ($ (NameDropping) + (0,1) $) (NonSpontaneous) {non-spontaneous};
    \end{scope}
    \end{scope}

    \begin{scope}[every edge/.style={draw=none},every node/.style={sloped}]
    \path
        (NRA)         edge node {$\subsetneq$} (RA)
        (NameDropping)  edge node {$\subsetneq$} (NonSpontaneous)
        ;
    \path[every edge/.style={
          -,draw=white,
          line width=2mm,
          shorten <=0.5mm,
          shorten >=0.5mm,
          postaction={transform canvas={yshift=1.5pt},draw=black,thin},
          postaction={transform canvas={yshift=-1.5pt},draw=black,thin},
          },
          every node/.style={
            above,
            fill=white,
            inner xsep=2pt,
            inner ysep=2pt,
            font=\relsize{0.1},
          },
          ]
        (NameDropping) edge node[xshift=0mm]{Section~\ref{sec:NameDroppingRA}} (NRA)
        (Local) edge node[yshift=-1pt]{Corollary~\ref{cor:rnna-nofa}} (NameDropping)
        (NonSpontaneous) edge node[xshift=-4.5mm]{\cite{BojanczykEA14}} (RA)
        (Global) edge (SessionA)
    ;
    \path[-,draw=white,line width=2mm] (Global) -- (Local) ;
    \path[->>,draw=black] (Global) -- (Local) ;
    \end{scope}

    \node[fit=(Global) (Local)] (NodesRNNA) {};
    \node[fit=(NonSpontaneous) (NameDropping)] (NodesNOFAs) {};
    \node[fit=(SessionA) (RA) (NRA)] (NodesRegisters) {};

    \begin{scope}[on background layer,every node/.style={ world }]
        \node[fit=(NodesRNNA)
                    (NodesRNNA.north|-NodesNOFAs.north) (NodesRNNA.north|-NodesRegisters.north)
                    (NodesRNNA.south|-NodesNOFAs.south) (NodesRNNA.south|-NodesRegisters.south)
              ] (RNNAs) {};
        \node[fit=(NodesNOFAs)
                    (NodesNOFAs.north|-NodesRNNA.north) (NodesNOFAs.north|-NodesRegisters.north)
                    (NodesNOFAs.south|-NodesRNNA.south) (NodesNOFAs.south|-NodesRegisters.south)
                    ] (NOFAs) {};
        \node[fit= (NodesRegisters)
                   (NodesRegisters.north|-NodesRNNA) (NodesRegisters.north|-NodesNOFAs.north)
                   (NodesRegisters.south|-NodesRNNA) (NodesRegisters.south|-NodesNOFAs.south)
                ] (Registers) {};
        \begin{scope}[every node/.style={
                    worldlabel,
                }]
            \node at (RNNAs.north) {RNNAs};
            \node at (NOFAs.north) {NOFAs};
            \node at (Registers.north) {Register Based};
        \end{scope}
    \end{scope}

\end{tikzpicture}%
\vspace{-1mm}
\caption {Expressivity of selected data language formalisms
  (restricted to empty initial register assignment). FSUBAs are
  properly contained in name-dropping RA.}%
  \vspace{-5mm}
\label{fig:models}%
\end{figure}
Under local freshness, RNNAs correspond to a natural subclass of RAs
(equivalently, NOFAs) defined by excluding nondeterministic
reassignment and by enforcing a policy of \emph{name dropping}, which
can be phrased as ``at any time, the automaton may
nondeterministically lose letters from registers'' -- thus freeing the
register but possibly getting stuck when lost names are expected to be
seen later.  
This policy is compatible with verification problems that relate to
scoping, such as `files that have been opened need to be closed before
termination' or `currently uncommitted transactions must be
either committed or explicitly aborted'.
Unsurprisingly, RNNAs with local freshness semantics are strictly more
expressive than FSUBAs; the relationships of the various models are
summarised in Figure~\ref{fig:models}. We show that RNNAs nevertheless
retain a decidable inclusion problem under local freshness, again in
parametrized \PSpace, using an algorithm that we obtain by varying the
one for global freshness. This is in spite of the fact that RNNAs a)
do not impose any bound on the number of registers, and b) allow
unrestricted nondeterminism and hence express languages whose
complement cannot be accepted by any RA, such as `some letter occurs
twice'.

\myparagraph{Further Related Work} A Kleene theorem for
\emph{deterministic} nominal automata and expressions with recursion
appears straightforward~\cite{KozenEA15}. Kurz et al.~\cite{KurzEA12}
introduce regular expressions for languages over words with scoped
binding, which differ technically from those used in the semantics of
NKA and regular bar expressions in that they are taken only modulo
$\alpha$-equivalence, not the other equations of NKA concerning scope
extension of binders. They satisfy a Kleene theorem for automata that
incorporate a bound on the nesting depth of binding, rejecting words
that exceed this depth.

Data languages are often represented as products of a classical finite
alphabet and an infinite alphabet; for simplicity, we use just the set
of names as the alphabet. Our unscoped name binders are, under local
semantics, similar to the binders in \emph{regular expressions with
  memory}, which are equivalent to unrestricted register
automata~\cite{LibkinEA15}.

Automata models for data languages, even models beyond register
automata such as fresh-register automata~\cite{Tzevelekos11} and
history-register automata~\cite{GrigoreTzevelekos16}, often have
decidable \emph{emptyness} problems, and their (less expressive)
deterministic restrictions then have decidable inclusion
problems. Decidability of inclusion can be recovered for
nondeterministic or even alternating register-based models by
drastically restricting the number of registers, to at most two in the
nondeterministic case~\cite{KaminskiFrancez94} and at most one in the
alternating case~\cite{DemriLazic09}. The complexity of the inclusion
problem for alternating one-register automata is non-primitive
recursive. \emph{Unambiguous register automata} have a decidable
inclusion problem and are closed under complement as recently shown by
Colcombet et al.~\cite{Colcombet15,ColcombetEA15}. RNNAs and
unambiguous RAs are incomparable: Closure under complement implies
that the language $L=$`some letter occurs twice' cannot be accepted by
an unambiguous RA, as its complement cannot be accepted by any
RA~\cite{BojanczykBook}. However, $L$ can be accepted by an RNNA (even
by an FSUBA). Failure of the reverse inclusion is due to name
dropping. 

Data walking automata~\cite{ManuelEA16} have strong navigational
capabilities but no registers, and are incomparable with unrestricted
RAs; we do not know how they relate to name-dropping RAs. Their
inclusion problem is decidable even under nondeterminism but at least
as hard as Petri net reachability, in particular not known to be
elementary.

\section{Preliminaries}\label{sec:prelim}

We summarise the basics of \emph{nominal sets}; \cite{Pitts13} offers
a comprehensive introduction.  \myparagraph{Group actions} Recall that
an \emph{action} of a group $G$ on a set $X$ is a map
$G \times X \to X$, denoted by juxtaposition or infix $\cdot$, such
that $\pi(\rho x) = (\pi\rho)x$ and $1x = x$ for $\pi, \rho \in G$,
$x \in X$. A \emph{$G$-set} is a set $X$ equipped with an action of
$G$.  The \emph{orbit} of $x \in X$ is the set
$\{\pi x \mid \pi \in G\}$. A function $f:X\to Y$ between $G$-sets
$X,Y$ is \emph{equivariant} if $f(\pi x) = \pi (fx)$ for all
$\pi \in G,x \in X$. Given a $G$-set $X$, $G$ acts on subsets
$A\subseteq X$ by $\pi A = \{\pi x \mid x \in A\}$. For $A\subseteq X$
and $x\in X$, we put
\[
\fix x = \{ \pi \in G \mid \pi x = x\} 
\qquad
\text{and}
\qquad
\textstyle\Fix A = \bigcap_{x \in A} \fix x.
\]
Note that elements of $\fix A$ and $\Fix A$ fix $A$ setwise and pointwise,
respectively.

\myparagraph{Nominal sets} Fix a countably infinite set $\Names$ of
\emph{names}, and write $G$ for the group of finite permutations on
$\Names$. Putting $\pi a = \pi(a)$ makes $\Names$ into a
$G$-set. Given a $G$-set $X$ and $x\in X$, a set $A \subseteq \Names$
\emph{supports} $x$ if $\Fix A \subseteq \fix x$, and $x$ \emph{has
  finite support} if some finite~$A$ supports $x$. In this case, there
is a least set $\supp(x)$ supporting~$x$. We say that $a\in\Names$ is
\emph{fresh for $x$}, and write $a\fresh x$, if $a\notin\supp(x)$. A
\emph{nominal set} is a $G$-set all whose elements have finite
support.  For every equivariant function $f$ between nominal sets, we
have $\supp(fx) \subseteq \supp(x)$. The function $\supp$ is
equivariant, i.e.\ $\supp(\pi x)=\pi(\supp(x))$ for $\pi\in G$.  Hence
$\size{\supp(x_1)} = \size{\supp(x_2)}$ whenever $x_1$, $x_2$ are in
the same orbit of a nominal set (we use $\size{}$ for cardinality). A
subset $S \subseteq X$ is \emph{finitely supported (fs)} if $S$ has
finite support with respect to the above-mentioned action of $G$ on
subsets; \emph{equivariant} if $\pi x \in S$ for all $\pi \in G$ and
$x \in S$ (which implies $\supp(S)=\emptyset$); and \emph{uniformly
  finitely supported (ufs)} if $\bigcup_{x \in S} \supp(x)$ is
finite~\cite{TurnerWinskel09}.  We denote by $\powfs(X)$ and
$\powufs(X)$ the sets of fs and ufs subsets of a nominal set~$X$,
respectively. Any ufs set is fs but not conversely; e.g.\ the set
$\Names$ is fs but not ufs. Moreover, any finite subset of $X$ is ufs
but not conversely; e.g.\ the set of words $a^n$ for fixed
$a\in\Names$ is ufs but not finite. A nominal set $X$ is
\emph{orbit-finite} if the action of~$G$ on it has only finitely many
orbits.
\begin{lem}[\cite{gabbay2011}, Theorem~2.29]\label{lem:supp}
  If $S$ is ufs, then $\supp(S) = \bigcup_{x \in S} \supp(x)$.
\end{lem}
\begin{lem}\label{lem:of}
  Every ufs subset of an orbit-finite set $X$ is finite.
\end{lem}
For a nominal set $X$ we have the \emph{abstraction
  set}~\cite{GabbayPitts99}
\begin{equation*}
[\Names]X=(\Names\times X)/{\sim}
\end{equation*}
where $\sim$ abstracts the notion of $\alpha$-equivalence as known
from calculi with name binding, such as the $\lambda$-calculus:
$(a,x)\sim(b,y)$ iff $(c\,a)\cdot x=(c\,b)\cdot y$ for any fresh
$c$. This captures the situation where $x$ and $y$ differ only in the
concrete name given to a bound entity that is called $a$ in $x$ and
$b$ in~$y$, respectively. We write $\langle a\rangle x$ for the
$\sim$-equivalence class of~$(a,x)$. E.g.\
$\langle a\rangle\{a,d\}=\langle b\rangle \{b,d\}$ in
$[\Names]\powf(\Names)$ provided that $d\notin\{a,b\}$.

\takeout{
\myparagraph{Coalgebra} An \emph{$F$-coalgebra} $(C,\gamma)$ for an
endofunctor $F:\BC\to\BC$ on a category $\BC$ consists of a
$\BC$-object $C$ of \emph{states} and a morphism $\gamma:C\to FC$;
here, we are interested in the case $\BC=\Nom$. A \emph{coalgebra
  morphism} $f:(C,\gamma)\to (D,\delta)$ is a morphism $f:C\to D$ such
that $Ff\gamma=\delta f$. An $F$-coalgebra $(C,\gamma)$ is
\emph{final} if for each $F$-coalgebra $(D,\delta)$, there exists a
unique coalgebra morphism $(D,\delta)\to(C,\gamma)$. 
E.g.\ coalgebras for the functor $FX=\Names \times X$ on $\Nom$
consist of equivariant maps $X\to\Names$ (output) and $X\to X$ (next
state), thus produce a stream of names at each state; by equivariance,
this stream has finite support, i.e.\ contains only finitely many
distinct names. Consequently, the final $F$-coalgebra in this case is
the set of finitely supported streams over~$\Names$.
}

\section{Strings and Languages with Name Binding}
\label{sec:barstrings}

\noindent As indicated in the introduction, we will take a simplified
view of \emph{data languages} as languages over an infinite alphabet;
we will use the set $\Names$ of names, introduced in
Section~\ref{sec:prelim}, as this alphabet, so that a \emph{data
  language} is just a subset $A\subseteq\Names^*$. Much like nominal
Kleene algebra (NKA)~\cite{GabbayCiancia11}, our formalism will
generate data words from more abstract strings that still include a
form of \emph{name binding}. Unlike in NKA, our binders will have
unlimited scope to the right, a difference that is in fact immaterial
at the level of strings but will be crucial at the level of regular
expressions. We write a bound occurrence of $a\in\Names$ as
$\newletter a$, and define an extended alphabet $\bar\Names$ by
\begin{equation*}
  \bar\Names=\Names\cup\{\,\newletter{a}\mid a\in\Names\}.
\end{equation*}%
\vspace{-2em}
\begin{defn}
  A \emph{bar string} is a word over $\bar\Names$, i.e.\ an element of
  $\bar\Names^*$. The set $\bar\Names^*$ is made into a nominal set by
  the letter-wise action of $G$. The \emph{free names} occurring in a
  bar string $w$ are those names $a$ that occur in $w$ to the left of
  any occurrence of $\newletter{a}$. A bar string is \emph{clean} if
  its bound letters $\newletter a$ are mutually distinct and distinct
  from all its free names. We write $\FN(w)$ for the set of free names
  of $w$, and say that $w$ is \emph{closed} if $\FN(w)=\emptyset$;
  otherwise, $w$ is \emph{open}. We define \emph{$\alpha$-equivalence}
  $\alphaeq$ on bar strings as the equivalence (\emph{not}:
  congruence) generated by $w \newletter{a}v\alphaeq w\newletter{b}u$
  if $\eq{a}v=\eq{b}u$ in $[\Names]\bar\Names^*$
  (Section~\ref{sec:prelim}). We write $[w]_\alpha$ for the
  $\alpha$-equivalence class of $w$. For a bar string $w$, we denote
  by $\unbar(w)\in\Names^*$ (for \emph{unbind}) the word arising from
  $w$ by replacing all bound names $\newletter a$ with the
  corresponding free name $a$.
\end{defn}
The set $\FN(w)$ is invariant under $\alpha$-equivalence, so we have a
well-defined notion of \emph{free names} of bar strings modulo
$\alphaeq$. Every bar string is $\alpha$-equivalent to a clean one.

\begin{expl}
  We have
  \( [ab\newletter{c}ab]_\alpha \neq [ab\newletter{a}ab]_\alpha =
  [ab\newletter{c}cb]_\alpha \neq [ap\newletter{c}cp]_\alpha \) where
  $\FN(ab\newletter{c}ab) = \FN(ab\newletter{a}ab) = \{a,b\}$. The bar
  string $ab\newletter{a}ab$ is not clean, and an $\alpha$-equivalent
  clean one is $ab\newletter{c}cb$.
\end{expl}
\begin{defn}
  \label{def:barString}
  A \emph{literal language} is a set of bar strings, and a \emph{bar
    language} is an fs set of bar strings modulo $\alpha$-equivalence,
  i.e.~an fs subset of
  \begin{equation}\label{eq:barM}
    \bar M:=\bar\Names^*/\mathord{\alphaeq}.
  \end{equation}
  A literal or bar language is \emph{closed} if all bar strings it
  contains are closed.
\end{defn}
\noindent Bar languages capture \emph{global freshness}; in fact, the
operator $N$ defined by
\begin{equation}  \label{eq:globalFreshness}
N(L)=\{\unbar(w)\mid w\text{ clean}, [w]_\alpha\in
L\}\subseteq\Names^*
\end{equation}
is injective on closed bar languages. Additionally, we define the
\emph{local freshness semantics} $\datalang(L)$ of a bar language $L$
by
\begin{equation}
  \datalang(L)=\{\unbar(w)\mid [w]_\alpha\in L\}\subseteq\Names^*.
  \label{eq:localFreshness}
\end{equation}
That is, $\datalang(L)$ is obtained by taking \emph{all}
representatives of $\alpha$-equivalence classes in $L$ and then
removing bars, while $N$ takes only clean
representatives. Intuitively, $D$ enforces local freshness by blocking
$\alpha$-renamings of bound names into names that have free
occurrences later in the bar string. The operator $\datalang$ fails to
be injective; e.g.\ (omitting notation for $\alpha$-equivalence
classes)
$D(\{\newletter a\newletter b,\newletter aa\})=\Names^2=D(\{\newletter
a\newletter b\})$.
This is what we mean by our slogan that \emph{local freshness is a
  quotient of global freshness}.
\begin{rem}
  Again omitting $\alpha$-equivalence classes, we have
  $D(\{\newletter a\newletter b\})=\Names^2$ because
  $\newletter a\newletter b\alphaeq\newletter a\newletter a$. On the
  other hand,
  $D(\{\newletter a\newletter b a\})=\{cdc\in\Names^3\mid c\neq d\}$
  because
  $\newletter a\newletter b a\not\alphaeq \newletter a\newletter a a$.
  We see here that since our local freshness semantics is based on
  $\alpha$-equivalence, we can only insist on a letter $d$ being
  distinct from a previously seen letter $c$ if $c$ will be seen again
  later. This resembles the process of \emph{register allocation} in a
  compiler, where program variables are mapped to CPU registers
  (see \cite[Sec.~9.7]{dragonbook1986} for
  details): 
  Each time the register allocation algorithm needs a register for a
  variable name ($\newletter{v}$), any register may be (re)used whose
  current content is not going to be accessed later.
\end{rem}

\begin{rem}
  In \emph{dynamic sequences}~\cite{GabbayEA15}, there are two
  dynamically scoped constructs $\langle a$ and $a \rangle$ for
  dynamic allocation and deallocation, respectively, of a name $a$; in
  this notation, our $\newletter{a}$ corresponds to $\langle a a$.  
\end{rem}

\section{Regular Bar Expressions}\label{sec:regexp}

Probably the most obvious formalism for bar languages are regular
expressions, equivalently finite automata, over the extended alphabet
$\bar\Names$. Explicitly:

\begin{defn}
  A \emph{nondeterministic finite bar automaton}, or \emph{bar NFA}
  for short, over $\Names$ is an NFA $A$ over $\bar\Names$. We call
  transitions of type $q\trans{a}q$ in $A$ \emph{free transitions} and
  transitions of type $q\trans{\scriptnew{a}}q$ \emph{bound
    transitions}. The \emph{literal language} $\litlang(A)$ of $A$ is
  the language accepted by $A$ as an NFA over $\bar\Names$. The
  \emph{bar language} $\barlang(A)\subseteq\bar M$
  (see~\eqref{eq:barM}) accepted by $A$ is defined as
  \begin{equation*}
    \barlang(A)=\litlang(A)/\mathord{\alphaeq}.
  \end{equation*}
  Generally, we denote by $\litlang(q)$ the $\bar\Names$-language
  accepted by the state $q$ in $A$ and by $\barlang(q)$ the quotient
  of $\litlang(q)$ by $\alpha$-equivalence.  The \emph{degree}
  $\degree(A)$ of $A$ is the number of names $a\in\Names$ that occur
  in transitions $q\trans{a}q'$ or $q\trans{\scriptnew a}q'$ in $A$.
  
  Similarly, a \emph{regular bar expression} is a regular expression
  $r$ over $\bar\Names$; the \emph{literal language}
  $\litlang(r)\subseteq\bar\Names^*$ \emph{defined by $r$} is the
  language expressed by $r$ as a regular expression, and the \emph{bar
    language defined by $r$} is
  $\barlang(r)=\litlang(r)/{\alphaeq}$. The \emph{degree} $\degree(r)$
  \emph{of $r$} is the number of names $a$ occurring as either
  $\newletter a$ or~$a$ in~$r$.
\end{defn}
\begin{expl}
  We have
  $L_\alpha(ac+\newletter{c}d) = \{ac\}\cup [\newletter{c}d]_\alpha$.
  Under local freshness semantics, this bar language contains for
  example $ad$, $bd$, and $cd$ but not
  $dd$. $D(L_\alpha\big((a+\newletter{a})^*\big))$ is the same
  language as $D(L_\alpha(\newletter{a}^*))$, even though
  $(a+\newletter{a})^*$ and $\newletter{a}^*$ define different bar
  languages.
\end{expl}
\begin{rem}\label{rem:session}
  Up to the fact that we omit the finite component of the alphabet
  often considered in data languages, a \emph{session
    automaton}~\cite{BolligEA14} is essentially a bar NFA (where free
  names $a$ are denoted as $a^\uparrow$, and bound names
  $\newletter a$ as $a^\circledast$). It defines an $\Names$-language
  and interprets bound transitions for $\newletter a$ as binding $a$
  to some globally fresh name. In the light of the equivalence of
  global freshness semantics and bar language semantics in the closed
  case, session automata are thus essentially the same as bar NFAs;
  the only difference concerns the treatment of open bar strings:
  While session automata explicitly reject bar strings that fail to be
  closed (\emph{well-formed}~\cite{BolligEA14}), a bar NFA will
  happily accept open bar strings. Part of the motivation for this
  permissiveness is that we now do not need to insist on regular bar
  expressions to be closed; in particular, regular bar expressions are
  closed under subexpressions.
\end{rem}
\begin{expl}
  In terms of $\Names$-languages, bar NFAs under global freshness
  semantics, like session automata, can express the language ``all
  letters are distinct'' (as $\newletter a^*$) but not the universal
  language~$\Names^*$~\cite{BolligEA14}.
\end{expl}

\begin{expl}\label{expl:regexp}
  The bar language
  $L=\{\epsilon,\newletter ba,\newletter ba \newletter ab,$
  $\newletter ba\newletter ab\newletter ba, \newletter ba\newletter
  ab\newletter ba\newletter ab\dots\}$
  (omitting equivalence classes) is defined by the regular bar
  expression $(\newletter ba \newletter ab)^*(1+\newletter ba)$ and
  accepted by the bar NFA $A$ with four states $s,t,u,v$, where $s$ is
  initial and $s$ and $u$ are final, and transitions
  $s\trans{\scriptnew b} t\trans{a} u\trans{\scriptnew a}v\trans{b}s$.
  Under global freshness, the closed bar language $\newletter a L$
  defines the language of odd-length words over $\Names$ with
  identical letters in positions $0$ and $2$ (if any), and with every
  letter in an odd position being globally fresh and repeated three
  positions later. Under local freshness, $\newletter a L$ defines the
  $\Names$-language consisting of all odd-length words over $\Names$
  that contain the same letters in positions $0$ and $2$ (if any) and
  repeat every letter in an odd position three positions later (if
  any) \emph{but no earlier}; that is, the bound names are indeed
  interpreted as being locally fresh. The reason for this is that,
  e.g., in the bar string $\newletter a\newletter ba \newletter ab$,
  $\alpha$-renaming of the bound name $\newletter b$ into
  $\newletter a$ is blocked by the occurrence of $a$
  after~$\newletter b$; similarly, the second occurrence of
  $\newletter a$ cannot be renamed into~$\newletter
  b$. \lsnote{@Thorsten: this is where the new examples can go}
\end{expl}
\begin{expl}
    The choice of fresh letters may restrict the branching later:
    The language $D(L_\alpha(\newletter{a} (c+dd))) = \{ac,dc,add,cdd\mid a\in
    \A\setminus\{c,d\}\}$ contains neither $bbb$ nor $cc$.
\end{expl}
\noindent We will see in the sequel that bar NFAs and regular bar
expressions are expressively equivalent to several other models,
specifically 
\begin{itemize}
\item under both semantics, to a nominal automaton model with name
  binding that we call \emph{regular nondeterministic nominal
    automata};
\item under local freshness, to a class of nondeterministic orbit
  finite automata~\cite{BojanczykEA14}; and consequently to a class of
  register automata.
\end{itemize}

\myparagraph{Nominal Kleene algebra} We recall that expressions $r,s$
of \emph{nominal Kleene algebra (NKA)}~\cite{GabbayCiancia11}, briefly
\emph{NKA expressions}, are defined by the grammar
\begin{equation*}
  r,s::= 0 \mid 1 \mid a \mid r+s\mid rs \mid r^* \mid \nu a.\,r\qquad 
  (a\in\Names).
\end{equation*}
Kozen et al.~\cite{KozenEA15b,KozenEA15} give a semantics of NKA in
terms of \emph{$\nu$-languages}. These are fs languages over words
with binding, so called \emph{$\nu$-strings}, which are either $1$ or
$\nu$-regular expressions formed using only names $a\in\Names$,
sequential composition, and name binding $\nu$, taken modulo the
equational laws of NKA~\cite{GabbayCiancia11}, including
$\alpha$-equivalence and laws for scope extension of
binding. \label{p:iso} In this semantics, a binder $\nu a$ is just
interpreted as itself, and all other clauses are standard. It is easy
to see that the nominal set of $\nu$-strings modulo the NKA laws is
isomorphic to the universal bar language $\bar M$; one converts bar
strings into $\nu$-strings by replacing any occurrence of
$\newletter{a}$ with $\nu a.a$, with the scope of the binder extending
to the end of the string. On closed expressions, $\nu$-language
semantics is equivalent to the semantics originally defined by Gabbay
and Ciancia~\cite{GabbayCiancia11,KozenEA15b}, which is given by the
operator $N$ defined in~\eqref{eq:globalFreshness} (now applied also
to languages containing open bar strings). Summing up, we can see NKA
as another formalism for bar languages. We will see in the next
section that regular bar expressions are strictly more expressive than
NKA; the crucial difference is that the name binding construct $\nu a$
of NKA has a static scope, while bound names $\newletter a$ in regular
bar expressions have dynamic scope.
\begin{rem}
  On open expressions, the semantics of~\cite{GabbayCiancia11}
  and~\cite{KozenEA15,KozenEA15b} differ as $N$ may interpret bound
  names with free names appearing elsewhere in the expression; e.g.\
  the NKA expressions $a+\nu a.\,a$ and $\nu a.\,a$ have distinct bar
  language semantics $\{a,\newletter a\}$ and $\{\newletter a\}$,
  respectively, which are both mapped to $\Names$ under $N$. For
  purposes of expressivity comparisons, we will generally restrict to
  closed expressions as well as ``closed'' automata and languages in
  the sequel. For automata, this typically amounts to the initial
  register assignment being empty, and for languages to being
  equivariant subsets of $\bar\Names^*$. 
\end{rem}
\section{Regular Nondeterministic Nominal Automata}\label{sec:rnna}

%


\noindent We proceed to develop a nominal automaton model that
essentially introduces a notion of configuration space into the
picture, and will turn out to be equivalent to bar NFAs. The
deterministic restriction of our model has been considered in the
context of NKA~\cite{KozenEA15}.

\begin{defn}
  A \emph{regular nondeterministic nominal automaton (RNNA)} is a
  tuple $A=(Q,\to,s,F)$ consisting of
  \begin{itemize}
  \item an orbit-finite set $Q$ of states, with an \emph{initial} state $s\in Q$;
  \item an equivariant subset $\to$ of $Q\times\bar\Names\times Q$,
    the \emph{transition relation}, where we write
    $q\trans{\alpha} q'$ for $(q,\alpha,q')\in\mathord{\to}$;
    transitions of type $q\trans{a}q'$ are called \emph{free}, and
    those of type $q\trans{\scriptnew a}q'$ \emph{bound};
  \item an equivariant subset $F\subseteq Q$ of \emph{final} states
  \end{itemize} 
  such that the following conditions are satisfied:
  \begin{itemize}
  \item The relation $\to$ is \emph{$\alpha$-invariant}, i.e.\ closed
    under $\alpha$-equivalence of transitions, where transitions
    $q \trans{\scriptnew a} q'$ and $p \trans{\scriptnew b} p'$ are
    \emph{$\alpha$-equivalent} if $q = p$ and
    $\langle a \rangle q' = \langle b \rangle p'$.
  \item The relation $\to$ is \emph{finitely branching up to
      $\alpha$-equivalence}, i.e.\ for each state $q$ the sets
    $\{(a,q')\mid q\trans{a}q'\}$ and
    $\{\langle a\rangle q'\mid q\trans{\scriptnew a}q'\}$ are finite
    (equivalently ufs, by Lemma~\ref{lem:of}).
  \end{itemize}
  The \emph{degree} $\degree(A)=\max\{\sharp\supp(q)\mid q\in Q\}$ of
  $A$ is the maximum size of supports of states in $A$.
\end{defn} 
\begin{rem}
  For readers familiar with universal coalgebra~\cite{Rutten00}, we
  note that RNNAs have a much more compact definition in coalgebraic
  terms, and in fact we regard the coalgebraic definition as evidence
  that RNNAs are a natural class of automata; however, no familiarity
  with coalgebras is required to understand the results of this
  paper. Coalgebraically, an RNNA is simply an orbit-finite coalgebra
  $\gamma:Q\to FQ$ for the functor $F$ on $\Nom$ given by
  \begin{equation*}
    FX = 2 \times\powufs(\Names\times X)\times\powufs([\Names]X),
  \end{equation*}
  together with an initial state $s \in Q$. The functor $F$ is a
  nondeterministic variant of the functor
  $KX = 2\times X^\Names\times[\Names]X$ whose coalgebras are
  \emph{deterministic nominal automata}~\cite{KozenEA15}. Indeed Kozen
  et al.~\cite{KozenEA15} show that the $\nu$-languages, equivalently
  the bar languages, form the final $K$-coalgebra.
\end{rem}
%
%
We proceed to define the language semantics of RNNAs.
\begin{defn}
  An RNNA $A$, with data as above, \emph{(literally) accepts} a bar
  string $w\in\bar\Names^*$ if $s\trans{w}q$ for some $q\in F$, where
  we extend the transition notation $\trans{w}$ to bar strings in the
  usual way. The \emph{literal language accepted by $A$} is the set
  $\litlang(A)$ of bar strings accepted by $A$, and the \emph{bar
    language accepted by $A$} is the quotient
  $\barlang(A) = \litlang(A)/\mathord{\equiv_\alpha}$.
\end{defn}


\noindent A key property of RNNAs is that supports of states evolve in
the expected way along transitions (cf.~\cite[Lemma 4.6]{KozenEA15}
for the deterministic case):
\begin{lem}\label{lem:trans}
  Let $A$ be an RNNA. Then the following hold.
  \begin{enumerate}
  \item\label{item:free-trans} If $q\trans{a}q'$ in $A$ then
    $\supp(q')\cup\{a\}\subseteq\supp(q)$.
  \item\label{item:bound-trans} If $q\trans{\scriptnew{a}}q'$ in $A$ then
    $\supp(q')\subseteq\supp(q)\cup\{a\}$.
  \end{enumerate}
\end{lem}
In fact, the properties in the lemma are clearly also sufficient for
ufs branching. From Lemma~\ref{lem:trans}, an easy induction shows
that for any state $q$ in an RNNA and any $w$ literally accepted by
$A$ from $q$, we have $\FN(w) = \supp([w]_\alpha) \subseteq
\supp(q)$. Hence:
\begin{cor}\label{cor:lang-ufs}
  Let $A$ be an RNNA. Then $\barlang(A)$ is ufs; specifically, if $s$ is the
  initial state of $A$ and $w\in \barlang(A)$, then
  $\supp(w)\subseteq\supp(s)$.
\end{cor}

\noindent We have an evident notion of $\alpha$-equivalence of paths
in RNNAs, defined analogously as for bar strings. 
Of course,
$\alpha$-equivalent paths always start in the same state. The set of
paths of an RNNA $A$ is closed under $\alpha$-equivalence.
However, this does not in
general imply that $\litlang(A)$ is closed under $\alpha$-equivalence;
e.g.\ for $A$ being
\begin{equation}\label{eq:non-alpha-closure}
s()\trans{\scriptnew a}t(a)\trans{\scriptnew b}u(a,b)
\end{equation}
(with $a,b$ ranging over distinct names in $\Names$), where $s()$ is
initial and the states $u(-,-)$ are final, we have
$\newletter a\newletter b\in\litlang(A)$ but the $\alpha$-equivalent
$\newletter a\newletter a$ is not in $\litlang(A)$. Crucially,
closure of $\litlang(A)$ under $\alpha$-equivalence is
nevertheless without loss of generality, as we show next. 
\begin{defn}\label{def:name-drop} An RNNA $A$ is \emph{name-dropping}
  if for every state $q$ in $A$ and every subset $N\subseteq\supp(q)$
  there exists a state $q|_N$ in $A$ that \emph{restricts $q$ to $N$};
  that is, $\supp(q|_N)= N$, $q|_N$ is final if $q$ is final, and
  $q|_N$ has at least the same incoming transitions as $q$ (i.e.\
  whenever $p\trans{\alpha}q$ then $p\trans{\alpha}q|_N$), and as many
  of the outgoing transitions of $q$ as possible;
  i.e.~$q|_N\trans{a}q'$ whenever $q\trans{a}q'$ and
  $\supp(q')\cup\{a\}\subseteq N$, and $q|_N\trans{\scriptnew a}q'$
  whenever $q\trans{\scriptnew a}q'$ and
  $\supp(q')\subseteq N\cup\{a\}$.
\end{defn} 
\noindent The counterexample shown in~\eqref{eq:non-alpha-closure}
fails to be name-dropping, as no state restricts $q=u(a,b)$ to
$N = \{b\}$. The following lemma shows that closure under
$\alpha$-equivalence is restored under name-dropping:
\begin{lem}\label{lem:name-drop-alpha}  Let $A$ be a name-dropping RNNA. Then $\litlang(A)$ is closed under
  $\alpha$-equivalence, i.e.\
  $\litlang(A)=\{w\mid[w]_\alpha\in\barlang(A)\}$.
\end{lem}
Finally, we can close a given RNNA under name dropping, preserving the
bar language:
\begin{lem}\label{lem:wlog-name-drop}
  Given an RNNA of degree $k$ with $n$ orbits, there exists a bar
  language-equivalent name-dropping RNNA of degree $k$ with at most
  $n2^k$ orbits.
\end{lem}
\begin{proof}[Sketch]
  From an RNNA $A$, construct an equivalent name-dropping RNNA with
  states of the form
  \begin{equation*}
    q|_N := \Fix(N) q
  \end{equation*}
  where $q$ is a state in $A$, $N\subseteq\supp(q)$, and $\Fix(N) q$
  denotes the orbit of $q$ under $\Fix(N)$. The final states are the
  $q|_N$ with $q$ final in $A$, and the initial state is
  $s|_{\supp(s)}$, where $s$ is the initial state of $A$. As
  transitions, we take
  \begin{itemize}
  \item $q|_N\trans{a}q'|_{N'}$ whenever
    $q\trans{a}q'$, $N'\subseteq N$, and $a\in N$, and
  \item $q|_N\trans{\scriptnew a} q'|_{N'}$ whenever
    $q\trans{\scriptnew b}q''$,
    $N''\subseteq\supp(q'')\cap(N\cup\{b\})$, and
    $\eq{a}{ (q'|_{N'})}=\eq{b}{(q''|_{N''})}$.
    \qed
  \end{itemize}
\end{proof}

\begin{expl} 
  Closing the RNNA from~\eqref{eq:non-alpha-closure} under name
  dropping as per Lemma~\ref{lem:wlog-name-drop} yields additional
  states that we may denote $u(\bot,b)$ (among others), with
  transitions $t(a)\trans{\scriptnew {b}}u(\bot,b)$; now,
  $\eq{b}{u(\bot,b)}=\eq{a}{u(\bot,a)}$, so $\newletter a\newletter a$
  is (literally) accepted.
\end{expl}

\myparagraph{Equivalence to bar NFAs} We proceed to show that RNNAs
are expressively equivalent to bar NFAs by providing mutual
translations. In consequence, we obtain a Kleene theorem connecting
RNNAs and regular bar expressions.
\begin{construction}\label{constr:RNNA}
  We construct an RNNA $\bar A$ from a given bar NFA $A$ with set $Q$
  of states, already incorporating closure under name dropping as per
  Lemma~\ref{lem:wlog-name-drop}. For $q\in Q$, put
\begin{math}
  N_q=\supp(\barlang(q)).
\end{math}
The set $\bar Q$ of states of $\bar A$ consists of pairs
\begin{equation*}
  (q,\pi F_N)\qquad(q\in Q\text{, }N\subseteq N_q)
\end{equation*}
where $F_N$ abbreviates $\Fix(N)$ and $\pi F_N$ denotes a left
coset. Left cosets for~$F_N$ can be identified with injective
renamings $N \to\Names$; intuitively, $(q,\pi F_N)$ restricts $q$ to
$N$ and renames~$N$ according to $\pi$. (That is, we construct a
configuration space, as in other translations into
NOFAs~\cite{CianciaTuosto09,BojanczykEA14}; here, we create virtual
registers according to $\supp(\barlang(q))$.) We let $G$ act on states
by $\pi_1\cdot(q,\pi_2F_N)=(q,\pi_1\pi_2F_N)$. The initial state of
$\bar A$ is $(s,F_{N_s})$, where $s$ is the initial state of $A$; a
state $(q,\pi F_N)$ is final in $\bar A$ iff $q$ is final in $A$. Free
transitions in $\bar A$ are given by
\begin{equation*}
  (q,\pi F_N)\trans{\pi(a)}(q',\pi F_{N'})\;\;\text{whenever}\;\; q\trans{a}q'
\text{ and }N'\cup\{a\}\subseteq N
\end{equation*}
and bound transitions by
\begin{equation*}
  (q,\pi F_N)\trans{\scriptnew{a}}(q',\pi'F_{N'})\;\:\text{whenever}\;\:
   q\trans{\scriptnew{b}}q',
  N'\subseteq N\cup\{b\},\eq{a}{\pi' F_{N'}} = \eq{\pi(b)}{\pi F_{N'}}.
\end{equation*}
\end{construction}
\begin{thm}\label{thm:freernna}
  $\bar A$ is a name-dropping RNNA with at most $|Q|2^{\degree(A)}$
  orbits, $\degree(\bar A)=\degree(A)$, and
  $\barlang(\bar A)=\barlang(A)$.
\end{thm}
\begin{expl}\label{expl:bar-A}
  The above construction converts the bar NFA $A$ of
  Example~\ref{expl:regexp}, i.e.\ the expression
  $(\newletter ba \newletter ab)^*(1+\newletter ba)$, into an RNNA
  that is similar to the one appearing in the counterexample to one
  direction of the Kleene theorem for NKA~\cite{KozenEA15} (cf.\
  Remark~\ref{rem:nka-non-kleene}): By the above description of left
  cosets for $F_N$, we annotate every state $q$ with a list of
  $\size{\supp(\barlang(q))}$ entries that are either (pairwise
  distinct) names or $\bot$, indicating that the corresponding name
  from $\supp(\barlang(q))$ has been dropped. We can draw those orbits
  of the resulting RNNA that have the form $(q,\pi N_q)$, i.e.~do not
  drop any names, as
  \begin{equation*}
  \begin{tikzcd}[row sep=-3.0mm, column sep=6mm,baseline=-1mm]
    &t(c,b)
        \arrow{dr}[above]{c}
    \\
    s(c) \arrow{ur}[above]{\scriptnew b}
    & & u(b) \arrow{dl}[below]{\scriptnew c}
    \\
    & v(b,c) \arrow{ul}[below]{b}
  \end{tikzcd}
  \quad
  \begin{array}{l}
  \text{for $b\neq c$, with $s(c)$, $u(b)$ final for all
  $b,c\in\Names$,}\\
    \text{and $s(c)$ initial.}
  \end{array}
  \end{equation*}
  Additional states then arise from name dropping; e.g.\ for $t$ we
  have additional states $t(\bot,b)$, $t(c,\bot)$, and $t(\bot,\bot)$,
  all with a $\newletter b$-transition from $s(c)$. The states
  $t(\bot,\bot)$ and $t(\bot,b)$ have no outgoing transitions, while
  $t(c,\bot)$ has a $c$-transition to $u(\bot)$.

\end{expl}
We next present the reverse construction, i.e.\ given an RNNA $A$ we
extract a bar NFA~$A_0$ (a subautomaton of $A$) such that
$\barlang(A_0)=\barlang(A)$.

Put $k=\degree(A)$. We fix a set $\Names_0\subseteq\Names$ of size
$\size{\Names_0}=k$ such that $\supp(s)\subseteq\Names_0$ for the initial
state $s$ of $A$, and a name $*\in\Names-\Names_0$. The states of
$A_0$ are those states $q$ in $A$ such that
\(
  \supp(q)\subseteq\Names_0.
\)
As this implies that the set $Q_0$ of states in $A_0$ is ufs, $Q_0$ is
finite by Lemma~\ref{lem:of}.  For $q,q'\in Q_0$,
the free transitions $q\trans{a}q'$ in $A_0$ are the same as in $A$
(hence $a\in\Names_0$ by Lemma~\ref{lem:trans}.1).  The bound
transitions $q\trans{\scriptnew{a}}q'$ in $A_0$ are those bound
transitions $q\trans{\scriptnew{a}}q'$ in $A$ such that
$a\in\Names_0\cup\{*\}$. A state is final in $A_0$ iff it is final in
$A$. The initial state of $A_0$ is $s\in Q_0$.
\begin{thm}\label{thm:nom-bar}
  The number of states in the bar NFA $A_0$ is linear in the number of
  orbits of $A$ and exponential in $\degree(A)$. Moreover,
  $\degree(A_0)\le\degree(A)+1$, and $\barlang(A_0) =\barlang(A)$.
\end{thm}

\noindent Combining this with Theorem~\ref{thm:freernna}, we obtain
the announced equivalence result:
\begin{cor}\label{cor:rnna-barnfa}
  RNNAs are expressively equivalent to bar NFAs, hence to regular bar
  expressions.
\end{cor}
This amounts to a \emph{Kleene theorem for RNNAs}. 
The decision procedure for inclusion (Section~\ref{sec:inclusion})
will use the equivalence of bar NFAs and RNNAs, essentially running a
bar NFA in synchrony with an RNNA.



\begin{rem}\label{rem:nka-non-kleene}
  It has been shown in that an NKA expression $r$ can be translated
  into a nondeterministic nominal automaton whose states are the
  so-called \emph{spines} of $r$, which amounts to one direction of a
  Kleene theorem~\cite{KozenEA15}. One can show that the spines in
  fact form an RNNA, so that NKA embeds into regular bar
  expressions. The automata-to-NKA direction of the Kleene theorem
  fails even for deterministic nominal automata, i.e.\ regular bar
  expressions are strictly more expressive than NKA. Indeed, the
  regular bar expression
  $(\newletter ba \newletter ab)^*(1+\newletter ba)$ of
  Example~\ref{expl:regexp} defines a language that cannot be defined
  in NKA because it requires unbounded nesting of name
  binding~\cite{KozenEA15}.
\end{rem}

\section{Name-Dropping Register Automata}
\label{sec:name-dropping}

We next relate RNNAs to two equivalent models of local freshness,
nondeterministic orbit-finite automata~\cite{BojanczykEA14} and
register automata (RAs)~\cite{KaminskiFrancez94}. RNNAs necessarily
only capture subclasses of these models, since RAs have an undecidable
inclusion problem~\cite{KaminskiFrancez94}; the distinguishing
condition is a version of name-dropping.
\begin{defn}~\cite{BojanczykEA14} A nondeterministic orbit-finite
  automaton (NOFA) $A$ consists of an orbit finite set $Q$ of states,
  two equivariant subsets $I, E \subseteq Q$ of \emph{initial} and
  \emph{final} states, respectively, and an equivariant transition
  relation $\mathord{\to}\subseteq Q\times\Names\times Q$, where we
  write $q\trans{a}p$ for $(q,a,p)\in\mathord{\to}$. 
  The $\Names$-language $L(A)=\{w\mid A\text{ accepts }w\}$
  \emph{accepted} by $A$ is defined in the standard way: extend the
  transition relation to words $w\in\Names^*$ as usual, 
  and then say that $A$ \emph{accepts} $w$ if there exist an initial
  state $q$ and a final state $p$ such that $q\trans{w}p$.  A
  \emph{DOFA} is a NOFA with a deterministic transition relation.
\end{defn}
\begin{rem}
  A more succinct equivalent presentation of NOFAs is as orbit-finite
  coalgebras $\gamma: Q \to GQ$ for the functor
  \begin{equation*}
    GX=2\times\powfs (\Names\times X)\hspace{5em} (2=\{\top,\bot\}) 
  \end{equation*}
  on the category $\Nom$ of nominal sets and equivariant maps,
  together with an equivariant subset of initial states. 
\end{rem}
\noindent More precisely speaking, NOFAs are equivalent to \emph{RAs
  with nondeterministic
  reassignment}~\cite{BojanczykEA14,KaminskiZeitlin10}. RAs are
roughly described as having a finite set of registers in which names
from the current word can be stored if they are \emph{locally fresh},
i.e.\ not currently stored in any register; transitions are labeled
with register indices $k$, meaning that the transition accepts the
next letter if it equals the content of register $k$. In the
equivalence with NOFAs, the names currently stored in the registers
correspond to the support of states.

To enable a comparison of RNNAs with NOFAs over $\Names$
(Section~\ref{sec:rnna}), we restrict our attention in the following
discussion to \mbox{RNNAs} that are \emph{closed}, i.e.\ whose initial
state has empty support, and therefore accept equivariant
$\Names$-languages. We can convert a closed RNNA $A$ into a NOFA
$D(A)$ accepting $D(\barlang(A))$ by simply replacing every transition
$q\trans{\scriptnew{a}}q'$ with a transition $q\trans{a}q'$. We show
that the image of this translation is a natural class of NOFAs:

\begin{defn}\label{def:nofa-namedrop}
  A NOFA $A$ is \emph{non-spontaneous} if $\supp(s)=\emptyset$ for
  initial states $s$, and \vspace{-1mm}
  \begin{equation*}
    \supp(q')\subseteq\supp(q) \cup \{a\}\quad\text{whenever}\quad q\trans{a}q'.
  \end{equation*}
  (In words, $A$ is non-spontaneous if transitions $q\trans{a}q'$ in
  $A$ create no new names other than $a$ in $q'$.) Moreover, $A$ is
  \emph{$\alpha$-invariant} if $q\trans{a}q''$ whenever
  $q\trans{b}q'$, $b\fresh q$, and $\eq{a}{q''}=\eq{b}{q'}$ (this
  condition is automatic if $a\fresh q$).  Finally, $A$ is
  \emph{name-dropping} if for each state $q$ and each set
  $N\subseteq\supp(q)$ of names, there exists a state $q|_N$ that
  \emph{restricts $q$ to $N$}, i.e.\ $\supp(q|_N)= N$, $q|_N$ is final
  if $q$ is final, and
  \begin{itemize}
  \item $q|_N$ has at least the same incoming transitions as $q$;
  \item whenever $q\trans{a}q'$, $a\in\supp(q)$, and
    $\supp(q')\cup\{a\}\subseteq N$, then $q|_N\trans{a}q'$;
  \item whenever $q\trans{a}q'$, $a\fresh q$, and
    $\supp(q')\subseteq N\cup\{a\}$, then $q|_N\trans{a}q'$.
  \end{itemize}
\end{defn} 
\begin{propn}\label{prop:nonspont}
  A NOFA is of the form $D(B)$ for some (name-dropping) RNNA $B$ iff
  it is (name-dropping and) non-spontaneous and $\alpha$-invariant.
\end{propn}
\begin{propn}\label{prop:alphainv}
  For every non-spontaneous and name-dropping NOFA, there is an
  equivalent non-spontaneous, name-dropping, and $\alpha$-invariant
  NOFA.
\end{propn}
In combination with Lemma~\ref{lem:name-drop-alpha}, these facts imply
\begin{cor}\label{cor:rnna-nofa}
  Under local freshness semantics, RNNAs are expressively equivalent
  to non-spontaneous name-dropping NOFAs.
\end{cor}
\begin{cor}\label{cor:intersections}
  The class of languages accepted by RNNAs under local freshness
  semantics is closed under finite intersections.
\end{cor}
\begin{proof}[Sketch]
  Non-spontaneous name-dropping NOFAs are closed under the standard
  product construction. \qed
\end{proof}

\begin{rem}\label{rem:local-expressivity}
  Every DOFA is non-spontaneous. Moreover, RAs are morally
  non-spontaneous according to their original definition, i.e.\ they
  can read names from the current word into the registers but cannot
  guess names nondeterministically~\cite{KaminskiFrancez94,NevenEA04};
  the variant of register automata that is equivalent to
  NOFAs~\cite{BojanczykEA14} in fact allows such
  \emph{nondeterministic reassignment}~\cite{KaminskiZeitlin10}. This
  makes unrestricted NOFAs strictly more expressive than
  non-spontaneous ones~\cite{KaminskiFrancez94,Wysocki13}. 
  Name-dropping restricts expressivity further, as witnessed by the
  language $\{ab\mid a\neq b\}$ mentioned above. In return, it buys
  decidability of inclusion (Section~\ref{sec:inclusion}), while for
  non-spontaneous NOFAs even universality is
  undecidable~\cite{BojanczykEA14,NevenEA04}. DOFAs are incomparable
  to RNNAs under local freshness semantics---the language ``the last
  letter has been seen before'' is defined by the regular bar
  expression $(\newletter b)^*\newletter a(\newletter b)^*a$ but not
  accepted by any DOFA.
\end{rem}

\myparagraph{Name-Dropping Register Automata and FSUBAs}
\label{sec:NameDroppingRA}
In consequence of Corollary~\ref{cor:rnna-nofa} and the equivalence
between RAs and nonspontaneous NOFAs, we have that RNNAs are
expressively equivalent to \emph{name-dropping} RAs, which we just
define as those RAs that map to name-dropping NOFAs under the
translation given in~\cite{BojanczykEA14}. We spend a moment on
identifying a more concretely defined class of \emph{forgetful} RAs
that are easily seen to be name-dropping. We expect that forgetful RAs
are as expressive as name-dropping RAs but are currently more
interested in giving a compact description of a class of name-dropping
RAs to clarify expressiveness.

We use the very general definition of RAs given
in~\cite{BojanczykEA14}: An RA with $n$ registers consists of a set
$C$ of \emph{locations} and for each pair $(c,c')$ of locations a
\emph{transition constraint} $\phi$. \emph{Register assignments}
$w\in R:=(\Names\cup\{\bot\})^n$ determine the, possibly undefined,
contents of the $n$ registers, and \emph{configurations} are elements
of $C\times R$. Transition constraints 
are equivariant subsets $\phi\subseteq R\times\Names\times R$, and
$(w,a,v)\in\phi$ means that from configuration $(c,w)$ the RA can
nondeterministically go to $(c',v)$ under input $a$. Transition
constraints have a syntactic representation in terms of Boolean
combinations of certain equations. The NOFA generated by an RA just
consists of its configurations.

For $w\in R$ and $N\subseteq\Names$ we define $w|_N\in R$ by
$(w|_N)_i=w_i$ if $w_i\in N$, and $(w|_N)_i=\bot$ otherwise. An RA is
\emph{forgetful} if it generates a non-spontaneous NOFA and for every
configuration $(c,w)$ and every $N$, $(c,w|_N)$ restricts $(c,w)$ to
$N$ in the sense of Definition~\ref{def:nofa-namedrop}; this property
is equivalent to evident conditions on the individual transition
constraints. In particular, it is satisfied if all transition
constraints of the RA are conjunctions of the evident non-spontaneity
restriction (letters in the poststate come from the input or the
prestate) with a positive Boolean combination of the following:
\begin{itemize}
\item $\mathsf{cmp}_i=\{(w,a,v)\mid w_i=a\}$ (block unless register
  $i$ contains the input)
\item $\mathsf{store}_i=\{(w,a,v)\mid v_i\in\{\bot,a\}\}$ (store the
  input in register $i$ or forget)
\item $\mathsf{fresh}_i=\{(w,a,v)\mid a\neq w_i\}$
  (block if register $i$ contains the input)
\item $\mathsf{keep}_{ji}=\{(w,a,v)\mid v_i\in\{\bot,w_j\}\}$ (copy
  register $j$ to register $i$, or forget)
\end{itemize}
FSUBAs~\cite{KaminskiTan06} can be translated into name-dropping
RAs. Unlike FSUBAs, forgetful RAs do allow for freshness
constraints. E.g.\ the language $\{aba\mid a\neq b\}$ is accepted by
the forgetful RA
$c_0\xrightarrow{\mathsf{store}_1}c_1\xrightarrow{\mathsf{fresh}_1\land\mathsf{keep}_{11}}c_2\xrightarrow{\mathsf{cmp}_1}c_3$,
with $c_3$ final. Note how $\mathsf{store}$ and $\mathsf{keep}$ will
possibly lose the content of register~$1$ but runs where this happens
will not get past $\mathsf{cmp}_1$.


\section{Deciding Inclusion under Global and Local Freshness }
\label{sec:inclusion}

We next show that under both global and local freshness, the inclusion
problem for bar NFAs (equivalently regular bar expressions) is in
\ExpSpace. For global freshness, this essentially just reproves the
known decidability of inclusion for session automata~\cite{BolligEA14}
(Remark~\ref{rem:session}; the complexity bound is not stated
in~\cite{BolligEA14} but can be extracted), while the result for local
freshness appears to be new. Our algorithm differs
from~\cite{BolligEA14} in that it exploits name dropping; we describe
it explicitly, as we will modify it for local freshness.

\begin{thm}\label{thm:barnfa-expspace}
  The inclusion problem for bar NFAs is in \ExpSpace; more precisely,
  the inclusion $\barlang(A_1)\subseteq \barlang(A_2)$ can be checked
  using space polynomial in the size of $A_1$ and $A_2$ and
  exponential in $\degree(A_2)\log(\degree(A_1)+\degree(A_2)+1)$.
\end{thm}
The theorem can be rephrased as saying that bar language inclusion
of NFA is in parametrized polynomial space
(para-\PSpace)~\cite{StockhusenTantau13}, the parameter being the degree. 
\begin{proof}[Sketch]
  Let $A_1$, $A_2$ be bar NFAs with initial states $s_1,s_2$. We
  exhibit an \NExpSpace procedure to check that $L_\alpha(A_1)$ is
  \emph{not} a subset of $L_\alpha(A_2)$, which implies the claimed
  bound by Savitch's theorem. It maintains a state $q$ of $A_1$ and a
  set $\Xi$ of states in the name-dropping RNNA $\bar A_2$ generated
  by $A_2$ as described in Construction~\ref{constr:RNNA}, with $q$
  initialized to $s_1$ and $\Xi$ to $\{(s_2,\id F_{N_{s_2}})\}$. It
  then iterates the following:
  \begin{enumerate}
  \item \label{step:left}Guess a transition $q\trans{\alpha} q'$ in
    $A_1$ and
  update $q$ to $q'$.
\item Compute the set $\Xi'$ of all states of
  $\bar A_2$ reachable from states in $\Xi$ via $\alpha$-transitions
  (literally, i.e.\ not up to $\alpha$-equivalence) and update
  $\Xi$ to $\Xi'$.
  \end{enumerate}

  \noindent The algorithm terminates successfully and reports that
  $L_\alpha(A_1)\not\subseteq L_\alpha(A_2)$ if it reaches a final
  state $q$ of $A_1$ while $\Xi$ contains only non-final states.

  Correctness of the algorithm follows from Theorem~\ref{thm:freernna}
  and Lemma~\ref{lem:name-drop-alpha}.
  For space usage, first recall that cosets $\pi F_N$ can be
  represented as injective renamings $N \to \Names$. Note that $\Xi$
  will only ever contain states $(q,\pi F_N)$ such that the image
  $\pi N$ of the corresponding injective renaming is contained in the
  set $P$ of names occurring literally in either $A_1$ or $A_2$. In
  fact, at the beginning, $\id N_{s_2}$ consists only of names
  literally occurring in $A_2$, and the only names that are added are
  those occurring in transitions guessed in Step~\ref{step:left},
  i.e.\ occurring literally in $A_1$.  So states $(q,\pi F_N)$ in
  $\Xi$ can be coded using partial functions $N_q\partialto P$. Since
  $\sharp P\le\degree(A_1)+\degree(A_2)$, there are
  at most
  $k\cdot (\degree(A_1)+\degree(A_2)+1)^{\degree(A_2)}=k\cdot
  2^{\degree(A_2)\log(\degree(A_1)+\degree(A_2)+1)}$ such states,
  where $k$ is the number of states of $A_2$. \qed
\end{proof}


\begin{rem}\label{rem:nka-complexity}
  The translation from NKA expressions to bar NFAs
  (Remark~\ref{rem:nka-non-kleene}) increases expression size
  exponentially but the degree only
  linearly. Theorem~\ref{thm:barnfa-expspace} thus implies the known
  \ExpSpace upper bound on inclusion for NKA
  expressions~\cite{KozenEA15}.  
\end{rem}
\noindent We now adapt the inclusion algorithm to local freshness
semantics.
  We denote by $\sqsubseteq$ the preorder (in fact: order) on
  $\bar\Names^*$ generated by $wav\sqsubseteq w\newletter a v$.

\begin{lem}\label{lem:datalang}
  Let $L_1,L_2$ be bar languages accepted by RNNA. Then
  $D(L_1)\subseteq D(L_2)$ iff for each $[w]_\alpha\in L_1$ there
  exists $w'\datasup w$ such that $[w']_\alpha\in L_2$.
\end{lem}

\begin{cor}\label{cor:weak-cont-expspace}
  Inclusion $D(\barlang(A_1))\subseteq D(\barlang(A_2))$ of bar
  NFAs (or regular bar expressions) under local freshness semantics
  is in para-\PSpace, with parameter
  $\degree(A_2)\log(\degree(A_1)+\degree(A_2)+1)$.
\end{cor}
\begin{proof}
  By Lemma~\ref{lem:datalang}, we can use a modification of the above
  algorithm where $\Xi'$ additionally contains states of $\bar A_2$
  reachable from states in $\Xi$ via $\newletter{a}$-transitions in
  case $\alpha$ is a free name~$a$. \qed
\end{proof}
%
%


\section{Conclusions}
We have studied the global and local freshness semantics of
\emph{regular nondeterministic nominal automata}, which feature
explicit name-binding transitions. We have shown that RNNAs are
equivalent to session automata~\cite{BolligEA14} under global
freshness and to \emph{non-spontaneous} and \emph{name-dropping}
nondeterministic orbit-finite automata (NOFAs)~\cite{BojanczykEA14}
under local freshness. Under both semantics, RNNAs are comparatively
well-behaved computationally, and in particular admit inclusion
checking in parameterized polynomial space. While this reproves known
results on session automata under global freshness, decidability of
inclusion under local freshness appears to be new. Via the equivalence
between NOFAs and register automata (RAs), we in fact obtain a
decidable class of RAs that allows unrestricted non-determinism and
any number of
registers. 



\myparagraph{Acknowledgements} We thank Charles Paperman for useful
discussions, and the anonymous reviewers of an earlier version of the
paper for insightful comments that led us to discover the crucial
notion of name dropping. Erwin R.\ Catesbeiana has commented on the
empty bar language.

\bibliographystyle{splncs03}
\bibliography{kleene}

\begin{thebibliography}{10}
\providecommand{\url}[1]{\texttt{#1}}
\providecommand{\urlprefix}{URL }

\bibitem{dragonbook1986}
Aho, A.V., Sethi, R., Ullman, J.D.: Compilers: Principles, Techniques, and
  Tools. Addison-Wesley Longman Publishing Co., Inc., Boston, MA, USA (1986)

\bibitem{BieleckiEA02}
Bielecki, M., Hidders, J., Paredaens, J., Tyszkiewicz, J., den Bussche, J.V.:
  Navigating with a browser. In: Automata, Languages and Programming, {ICALP}
  2002. LNCS, vol. 2380, pp. 764--775. Springer (2002)

\bibitem{Bojanczyk10}
Boja{\'{n}}czyk, M.: Automata for data words and data trees. In: Rewriting
  Techniques and Applications, {RTA} 2010. LIPIcs, vol.~6, pp. 1--4. Schloss
  Dagstuhl - Leibniz-Zentrum fuer Informatik (2010)

\bibitem{BojanczykBook}
Boja\'nczyk, M.: Computation in sets with atoms (November 2013), draft book;
  available at
  \url{http://atoms.mimuw.edu.pl/wp-content/uploads/2014/03/main.pdf}

\bibitem{BojanczykEA14}
Bojanczyk, M., Klin, B., Lasota, S.: Automata theory in nominal sets. Log.\
  Methods Comput.\ Sci.  10 (2014)

\bibitem{BolligEA14}
Bollig, B., Habermehl, P., Leucker, M., Monmege, B.: A robust class of data
  languages and an application to learning. Log.\ Meth.\ Comput.\ Sci.  10
  (2014)

\bibitem{CianciaTuosto09}
Ciancia, V., Tuosto, E.: A novel class of automata for languages on infinite
  alphabets. Tech. rep., University of Leicester (2009), cS-09-003

\bibitem{ColcombetEA15}
Colcombet, T., Puppis, G., Skrypczak, M.: Unambiguous register automata,
  preprint

\bibitem{Colcombet15}
Colcombet, T.: Unambiguity in automata theory. In: Descriptional Complexity of
  Formal Systems, {DCFS} 2015. LNCS, vol. 9118, pp. 3--18. Springer (2015)

\bibitem{DemriLazic09}
Demri, S., Lazic, R.: {LTL} with the freeze quantifier and register automata.
  {ACM} Trans.\ Comput.\ Log.  10 (2009)

\bibitem{GabbayPitts99}
Gabbay, M., Pitts, A.: A new approach to abstract syntax involving binders. In:
  Logic in Computer Science, LICS 1999. pp. 214--224. {IEEE} Computer Society
  (1999)

\bibitem{gabbay2011}
Gabbay, M.J.: Foundations of nominal techniques: logic and semantics of
  variables in abstract syntax. Bull.~Symbolic Logic  17(2),  161--229 (2011)

\bibitem{GabbayCiancia11}
Gabbay, M.J., Ciancia, V.: Freshness and name-restriction in sets of traces
  with names. In: Foundations of Software Science and Computational Structures,
  {FOSSACS} 2011. LNCS, vol. 6604, pp. 365--380. Springer (2011)

\bibitem{GabbayEA15}
Gabbay, M.J., Ghica, D.R., Petrisan, D.: Leaving the nest: Nominal techniques
  for variables with interleaving scopes. In: Computer Science Logic, {CSL}
  2015. LIPIcs, vol.~41, pp. 374--389. Schloss Dagstuhl - Leibniz-Zentrum fuer
  Informatik (2015)

\bibitem{GrigoreTzevelekos16}
Grigore, R., Tzevelekos, N.: History-register automata. Log.\ Meth.\ Comput.\
  Sci.  12(1) (2016)

\bibitem{GrumbergEA10}
Grumberg, O., Kupferman, O., Sheinvald, S.: Variable automata over infinite
  alphabets. In: Language and Automata Theory and Applications, {LATA} 2010.
  LNCS, vol. 6031, pp. 561--572. Springer (2010)

\bibitem{Hennessy02}
Hennessy, M.: A fully abstract denotational semantics for the pi-calculus.
  Theor.\ Comput.\ Sci.  278,  53--89 (2002)

\bibitem{KaminskiFrancez94}
Kaminski, M., Francez, N.: Finite-memory automata. Theor.\ Comput.\ Sci.  134,
  329--363 (1994)

\bibitem{KaminskiTan06}
Kaminski, M., Tan, T.: Regular expressions for languages over infinite
  alphabets. Fund.\ Inform.  69,  301--318 (2006)

\bibitem{KaminskiZeitlin10}
Kaminski, M., Zeitlin, D.: Finite-memory automata with non-deterministic
  reassignment. Int.\ J.\ Found.\ Comput.\ Sci.  21,  741--760 (2010)

\bibitem{KozenEA15}
Kozen, D., Mamouras, K., Petrisan, D., Silva, A.: Nominal {K}leene coalgebra.
  In: Automata, Languages, and Programming, {ICALP} 2015. LNCS, vol. 9135, pp.
  286--298. Springer (2015)

\bibitem{KozenEA15b}
Kozen, D., Mamouras, K., Silva, A.: Completeness and incompleteness in nominal
  kleene algebra. In: Relational and Algebraic Methods in Computer Science,
  RAMiCS 2015. LNCS, vol. 9348, pp. 51--66. Springer (2015)

\bibitem{KurtzEA07}
K{\"{u}}rtz, K., K{\"{u}}sters, R., Wilke, T.: Selecting theories and nonce
  generation for recursive protocols. In: Formal methods in security
  engineering, {FMSE} 2007. pp. 61--70. {ACM} (2007)

\bibitem{KurzEA12}
Kurz, A., Suzuki, T., Tuosto, E.: On nominal regular languages with binders.
  In: Foundations of Software Science and Computational Structures, {FOSSACS}
  2012. LNCS, vol. 7213, pp. 255--269. Springer (2012)

\bibitem{LibkinEA15}
Libkin, L., Tan, T., Vrgoc, D.: Regular expressions for data words. J.\
  Comput.\ Syst.\ Sci.  81,  1278--1297 (2015)

\bibitem{maclane98}
MacLane, S.: Categories for the working mathematician. Springer, 2nd edn.
  (1998)

\bibitem{ManuelEA16}
Manuel, A., Muscholl, A., Puppis, G.: Walking on data words. Theory Comput.
  Sys.  59,  180--208 (2016)

\bibitem{msw16}
Milius, S., Schr\"oder, L., Wi{\ss}mann, T.: Regular behaviours with names: On
  rational fixpoints of endofunctors on nominal sets. Appl.\ Cat.\ Struct.
  (2016), in print

\bibitem{NevenEA04}
Neven, F., Schwentick, T., Vianu, V.: Finite state machines for strings over
  infinite alphabets. {ACM} Trans.\ Comput.\ Log.  5,  403--435 (2004)

\bibitem{Pitts13}
Pitts, A.: Nominal Sets: Names and Symmetry in Computer Science. Cambridge
  University Press (2013)

\bibitem{Rutten00}
Rutten, J.: Universal coalgebra: a theory of systems. Theoret.\ Comput.\ Sci.
  249(1),  3--80 (2000)

\bibitem{Segoufin06}
Segoufin, L.: Automata and logics for words and trees over an infinite
  alphabet. In: Computer Science Logic, {CSL} 2006. LNCS, vol. 4207, pp.
  41--57. Springer (2006)

\bibitem{sbbr13}
Silva, A., Bonchi, F., Bonsangue, M.M., Rutten, J.J.M.M.: Generalizing
  determinization from automata to coalgebras. Log.~Methods Comput.~Sci  9(1:9)
  (2013)

\bibitem{StockhusenTantau13}
Stockhusen, C., Tantau, T.: Completeness results for parameterized space
  classes. In: Parameterized and Exact Computation, {IPEC} 2013. LNCS, vol.
  8246, pp. 335--347. Springer (2013)

\bibitem{TalMScThesis}
Tal, A.: Decidability of Inclusion for Unification Based Automata. Master's
  thesis, Technion (1999)

\bibitem{TurnerWinskel09}
Turner, D., Winskel, G.: Nominal domain theory for concurrency. In: Computer
  Science Logic, {CSL} 2009. pp. 546--560 (2009)

\bibitem{Tzevelekos11}
Tzevelekos, N.: Fresh-register automata. In: Principles of Programming
  Languages, {POPL} 2011. pp. 295--306. {ACM} (2011)

\bibitem{Wysocki13}
Wysocki, T.: Alternating register automata on finite words. Master's thesis,
  University of Warsaw (2013), (In Polish)

\end{thebibliography}
\newpage\appendix

\makeatletter
\edef\thetheorem{\expandafter\noexpand\thesection\@thmcountersep\@thmcounter{theorem}}
\makeatother

\section{Omitted Proofs}

In this appendix we assume that readers are familiar with basic
notions of category theory (see e.g.~\cite{maclane98}), with the
theory of algebras and coalgebras for a functor (see
e.g.~\cite{Rutten00}), and with basic properties of nominal sets (see
e.g.~\cite{Pitts13}).

\subsection{Abstraction in Nominal Sets}

We occasionally use, without express mention, the following
alternative description of equality in the abstraction $[\Names]X$,
which formalizes the usual intuitions about $\alpha$-equivalence:
\begin{lem}\label{lem:abstr}
  Let $a,b\in\Names$ and $x,y\in X$. Then $\eq{a}{x}=\eq{b}{y}$ in
  $[\Names]X$ iff either
  \begin{myrmenumerate}
  \item $(a,x)=(b,y)$, or
  \item $b\neq a$, $b\fresh x$, and $(ab)\cdot x=y$.
  \end{myrmenumerate}
\end{lem}
\begin{proof}
  \emph{`If':} the case where (i) holds is trivial, so assume
  (ii). Let $c$ be fresh; we have to show $(ca)\cdot x=(cb)\cdot y$.
  But $(cb)\cdot y=(cb)\cdot(ab)\cdot x=(acb)\cdot x = (ca)\cdot x$,
  where we use in the last step that $b,c$ are both fresh for $x$ so that 
  $(ca)^{-1}(acb)=(ca)(acb)=(bc)$ fixes $x$.

  \emph{`Only if':} We assume $(a,x)\neq (b,y)$ and prove (ii). We
  first show $a\neq b$: Assume the contrary. Let $c$ be fresh; by the
  definition of abstraction, we then have $(ca)\cdot x = (cb)\cdot y$,
  so $y=(cb)(ca)\cdot x=(ca)(ca)\cdot x=x$, contradiction.  We have
  $\supp(x)\subseteq\{a\}\cup\supp(\eq{a}{x})=\{a\}\cup\supp(\eq{b}{y})$,
  whence $b\fresh x$ since $a\neq b$ and $b\fresh\eq{b}{y}$. Finally,
  with $c$ as above
  $y=(cb)^{-1}(ca)\cdot x=(cb)(ca)\cdot x=(acb)\cdot x=(ab)\cdot x$,
  again because $(ab)^{-1}(abc)=(ab)(abc)=(bc)$ and $b,c$ are fresh
  for $x$.
\end{proof}

As an easy consequence we obtain:
\begin{cor}
  Let $X$ be a nominal set, $a \in \Names$ and $x \in X$. Then $\supp(\eq{a}{x}) = \supp(x) - \{a\}$. 
\end{cor}

\myparagraph{Proof of Lemma~\ref{lem:of}}
    Firstly, any finite set $S\subseteq X$ is ufs, because $\bigcup_{y\in S}
    \supp(y)$ is a finite union of finite sets. Secondly, for any ufs $S\subseteq
    X$, we have $\supp(S) = \bigcup_{y\in S} \supp(y)$, which is a finite union
    (because $X$ is orbit-finite) of again finite sets. \qed

\subsection{Proofs and Lemmas for Section~\ref{sec:barstrings}}

\begin{lem}\label{lem:global-injective}
  The operator $N$ is injective on closed bar languages.
\end{lem}
The proof of Lemma~\ref{lem:global-injective} relies on the following
simple fact:
\begin{lem}\label{lem:ub-injective}
  The operator $\unbar$ is injective on closed clean bar strings.
\end{lem}
\begin{proof}
  Let $w_1,w_2$ be closed clean bar strings, and let
  $\unbar(w_1)=\unbar(w_2)$. Assume for a contradiction that
  $w_1\neq w_2$. Picking the leftmost position where $w_1$ and $w_2$
  differ, we have w.l.o.g.\ $u,v_1,v_2\in\bar\Names^*$ and
  $a\in\Names$ such that
  \begin{equation*}
    w_1=uav_1\qquad \text{and}\qquad w_2=u\newletter av_2. 
  \end{equation*}
  Since $w_1$ is closed, $u$ must contain $\newletter a$, in
  contradiction to $w_2$ being clean. \qed
\end{proof}

\begin{proof}[Lemma~\ref{lem:global-injective}] Let $L_1,L_2$
  be closed bar languages such that $N(L_1)=N(L_2)$, and let
  $[w]_\alpha\in L_1$. We have to show $[w]_\alpha\in L_2$. We have
  that $w$ is closed, and w.l.o.g.\ $w$ is clean. Then
  $\unbar(w)\in N(L_1)$, and hence $\unbar(w)\in N(L_2)$, so there
  exists a clean and closed $w'$ such that $\unbar(w')=\unbar(w)$ and
  $[w']_\alpha\in L_2$. By Lemma~\ref{lem:ub-injective}, $w=w'$, so
  that $[w]_\alpha=[w']_\alpha\in L_2$ as required. \qed
\end{proof}

\subsection{Proofs and Lemmas for Section~\ref{sec:rnna}}

\begin{defn}
  Given a state $q$ in an RNNA $A$ we write $\litlang(q)$ and
  $\barlang(q)$ for the literal language and the bar language,
  respectively, accepted by the automaton obtained by making $q$ the
  initial state of $A$.
\end{defn}

\begin{lem}\label{lem:lang-equivar}
  In an RNNA, the map $q\mapsto \barlang(q)$ is equivariant.
\end{lem}

\begin{proof}
  Note first that the set of bar strings $\bar \Names^*$ is the
  initial algebra for the functor
  $SX = 1 + \Names \times X +\Names \times X$ on $\Nom$. And the set
  $\bar M$ of bar strings modulo $\alpha$-equivalence is the intial
  algebra for the functor
  $S_\alpha X = 1 + \Names \times X + [\Names]X$ on $\Nom$. The
  functor $S_\alpha$ is a quotient of the functor $S$ via the natural
  transformation $q: S \epito S_\alpha$ given by the canonical
  quotient maps $\Names \times X \epito [\Names]X$. The canonical
  quotient map $[-]_\alpha: \bar\Names^* \epito \bar M$ that maps
  every bar string to its $\alpha$-equivalence class is obtained
  inductively, i.e.~$[-]_\alpha$ is the unique equivariant map such
  that the following square commutes:
  \[
  \begin{tikzcd}
    S\bar\Names^*
    \arrow{d}[swap]{S[-]_\alpha}
    \arrow{rr}{\iota}
    &&
    \bar\Names^*
    \arrow{d}{[-]_\alpha}
    \\
    S\bar M 
    \arrow{r}{q_{\bar M}}
    & S_\alpha\bar M
    \arrow{r}{\iota_{\alpha}}
    & \bar M 
  \end{tikzcd}
  \]
  where $\iota: S\bar\Names^* \to \bar\Names^*$ and
  $\iota_\alpha: S_\alpha\bar M \to \bar M$ are the structures of the
  initial algebras, respectively. Since the map $[-]_\alpha$ is
  equivariant we thus have $\pi[w]_\alpha = [\pi w]_\alpha$ for every
  $w \in \bar\Names^*$.

  Now we prove the statement of the lemma. Since both free and bound
  transitions are equivariant, the literal language $\litlang(-)$ is
  equivariant. It follows that the bar language $\barlang(-)$ is
  equivariant: If $m\in\barlang(q)$ then there is $w\in\litlang(q)$
  such that $[w]_\alpha=m$. For $\pi\in G$, it follows that
  $\pi\cdot w\in\litlang(\pi q)$, and hence
  $[\pi\cdot w]_\alpha\in\barlang(\pi q)$. But
  $[\pi\cdot w]_\alpha=\pi[w]_\alpha$, so $\pi m\in\barlang(\pi q)$.
  \qed
\end{proof}

\myparagraph{Proof of Lemma~\ref{lem:trans}}
\begin{itemize}

\item[\ref{item:free-trans}.] Consider the ufs set $Z=\{(a,q')\mid q\trans{a}q'\}$. Then we have
$\supp(q')\cup\{a\}=\supp(a,q')\subseteq\supp(Z)\subseteq\supp(q)$
where the second inclusion holds because $Z$ is ufs, and the third because $Z$ depends equivariantly on $q$.

\item[\ref{item:bound-trans}.] Consider the ufs set $Z=\{[a]q'\mid q\trans{\scriptnew{a}}q'\}$. Then we have 
$\supp(q')\subseteq\supp([a]q')\cup\{a\}\subseteq\supp(Z)\cup\{a\}\subseteq\supp(q)\cup\{a\}$
where the second inclusion holds because $Z$ is ufs, and the third because $Z$ depends equivariantly on $q$. 
\qed
\end{itemize}

\begin{rem}\label{rem:paths}
  Given an RNNA $A$ with the state set $Q$ the paths in $A$ form the
  initial algebra for the functor $Q \times S(-)$, where $S$ is the
  functor in the proof of Lemma~\ref{lem:lang-equivar}. Paths in $A$
  modulo $\alpha$-equivalence then form the initial algebra for
  $Q \times S_\alpha(-)$ and the canonical quotient map $[-]_\alpha$
  mapping a path to its $\alpha$-equivalence class is obtained by
  initiality similarly as the canonical quotient map in
  Lemma~\ref{lem:lang-equivar}.
\end{rem}

\myparagraph{Proof of Lemma~\ref{lem:name-drop-alpha}}

\begin{lem}\label{lem:lang-restrict} 
  Let $A$ be a name-dropping RNNA, and let $q|_N$ restrict a state $q$
  in $A$ to $N\subseteq\supp(q)$. Then
  $\{w\in\litlang(q)\mid\FN(w)\subseteq N\}\subseteq\litlang(q|_N)$.
\end{lem}
\begin{proof}
  Induction on the length of $w\in\litlang(q)$ with
  $\FN(w)\subseteq N$, with the base case immediate from the finality
  condition in Definition~\ref{def:name-drop}. So let $w=\alpha v$
  with $\alpha\in\bar\Names$, accepted via a path
  $q\trans{\alpha}q'\trans{v}p$, and let $q'|_{N_v}$ restrict $q'$ to
  $N_v:=\FN(v)$.  By the induction hypothesis, $v\in\litlang(q'|_{N_v})$.  Moreover,
  $q\trans{\alpha}q'|_{N_v}$. We are done once we show that
  $q|_N\trans{\alpha}q'|_{N_v}$. If $\alpha$ is free, then we have to
  show $N_v\cup\{\alpha\}\subseteq N$, and if $\alpha=\newletter a$
  is bound, we have to show $N_v\subseteq N\cup\{a\}$. In both cases,
  the requisite inclusion is immediate from
  $\FN(\alpha v)\subseteq N$.
  \qed
\end{proof}

\begin{proof}[Lemma~\ref{lem:name-drop-alpha}]
  We use induction on the word length. It suffices to show that
  $\litlang(A)$ is closed under single $\alpha$-conversion steps. So
  let $v\newletter{a}w\in\bar\Names^*\in\litlang(A)$, via a path
  $s=q_0\trans{v} q_1\trans{\scriptnew{a}} q_2\trans{w}q_3$ (with
  $q_3$ final), let $b\neq a$ with $b\notin\FN(w)$, and let $w'$ be
  obtained from $w$ by replacing free occurrences of $a$ with $b$. We
  have to show that $v\newletter{b}w'\in\litlang(A)$; it suffices to
  show that $\newletter{b}w'\in\litlang(q_1)$. Put
  $N=\supp(q_2)-\{b\}$, and let $q_2|_N$ restrict $q_2$ to $N$. By
  Lemma~\ref{lem:lang-restrict}, $w\in\litlang(q_2|_N)$, and hence
  $[w]_\alpha\in\barlang(q_2|_N)$. By Lemma~\ref{lem:lang-equivar}, it
  follows that $[w']_\alpha\in\barlang((ab)(q_2|_N))$, so by the
  induction hypothesis, $w'\in\litlang((ab)(q_2|_N))$. We clearly have
  $q_1\trans{\scriptnew a}q_2|_N$, and by $\alpha$-invariance,
  $q_1\trans{\scriptnew b}(ab)(q_2|_N)$ because $b\fresh(q_2|_N)$.
Thus, $\newletter b w'\in\litlang(q_1)$ as required. \qed
\end{proof}

\myparagraph{Proof of Lemma~\ref{lem:wlog-name-drop}}

\begin{defn}
  Given the transition data of an RNNA $A$ (not necessarily assuming
  any finiteness and invariance conditions) and a state $q$ in $A$, we
  denote by $\freesuc(q)$ the set
  \begin{equation*}
    \freesuc(q)=\{(a,q')\mid q\trans{a}q'\}
  \end{equation*}
  of free transitions of $q$, and by $\boundsuc(q)$ the set
  \begin{equation*}
    \boundsuc(q)=\{\eq{a}{q'}\mid q\trans{\scriptnew a}q'\}
  \end{equation*}
  of bound transitions of $q$ modulo
  $\alpha$-equivalence. 
\end{defn}

Note that under $\alpha$-invariance of transitions we have $\eq a q'
\in \boundsuc(q)$ if and only if $q \trans{\scriptnew a} q'$. 

Before we proceed to the proof of the lemma we note the following general fact about nominal sets: 
for the value of $\pi\cdot x$, it matters only what $\pi$
does on the atoms in $\supp(x)$:
\begin{lem}\label{onlySuppMatters}
    For $x\in (X,\cdot)$ and any $\pi,\sigma \in G$ with $\pi(v) =
    \sigma(v)$ for all $v\in \supp(x)$, we have $\pi\cdot x = \sigma\cdot x$.
\end{lem}
\begin{proof}
  Under the given assumptions, $\pi^{-1}\sigma\in\Fix(\supp(x))\subseteq\fix(x)$.\qed
\end{proof}

\begin{proof}[Lemma~\ref{lem:wlog-name-drop}]
  Let $A$ be an RNNA with set $Q$ of states. 

  (1)~We construct an
  equivalent name-dropping RNNA $A'$ as follows. As states, we take
  pairs
  \begin{equation*}
    q|_N := \Fix(N) q 
  \end{equation*}
  where $q\in Q$, $N\subseteq\supp(q)$, and $\Fix(N) q$ denotes the
  orbit of $q$ under $\Fix(N)$.  We define an action of $G$ on states
  by $\pi\cdot (q|_N)=(\pi q)|_{\pi N}$.  To see well-definedness, let
  $\pi'\in \Fix(N)$ (i.e.\ $(\pi' q)|_N=q|_N$); we have to show
  $(\pi\pi' q)|_{\pi N}=(\pi q)|_{\pi N}$.  Since
  $(\pi\pi'\pi^{-1})\pi q=\pi\pi' q$, this follows from
  $\pi\pi'\pi^{-1}\in \Fix(\pi N)$. The map $(q,N)\mapsto q|_N$ is
  equivariant, which proves the bound on the number of orbits
  in~$A'$. A state $q|_N$ is final if $q$ is final in $A$; this clearly yields an equivariant subset of states of $A'$. 
  The initial state of $A'$ is $s|_{\supp(s)}$
  where $s$ is the initial state of $A$. We have
  \begin{equation}\label{eq:name-drop-supp}
    \supp(q|_N)= N;
  \end{equation}
  in particular, the states of $A'$ form a nominal set. To see
  `$\subseteq$' in~\eqref{eq:name-drop-supp}, it suffices to show that
  $N$ supports $q|_N$. So let $\pi\in\Fix(N)$. Then
  $\pi\cdot(q|_N)= (\pi q)|_{\pi N}=q|_N$, as required. For
  `$\supseteq$', let $a\in N$; we have to show that $N-\{a\}$ does not
  support $q|_N$. Assume the contrary. Pick $b\fresh q$. Then
  $(ab)\in\Fix(N-\{a\})$, so
  $(ab)\cdot(q|_N)=\Fix((ab)\cdot N)(ab)\cdot q =\Fix(N) q= q|_N$.  In
  particular, $q\in\Fix (ab)\cdot N)(ab)\cdot q$, i.e.~there is
  $\rho\in\Fix((ab)\cdot N)$ such that $\rho(ab)\cdot q=q$. By
  equivariance of $\supp$, it follows that
  $\rho(ab)\cdot\supp(q)=\supp(q)$. Now $b\in(ab)\cdot N$, so
  $\rho(b)=b$. Since $a\in\supp(q)$, it follows that
  $b\in\rho(ab)\cdot\supp(q)$; but $b\notin\supp(q)$, contradiction.

  As transitions of $A'$, we take
  \begin{itemize}
  \item $q|_N\trans{a}q'|_{N'}$ whenever
    $q\trans{a}q'$, $N'\subseteq N$, and $a\in N$, and
  \item $q|_N\trans{\scriptnew a} q'|_{N'}$ whenever
    $q\trans{\scriptnew b}q''$,
    $N''\subseteq\supp(q'')\cap(N\cup\{b\})$, and
    $\eq{a}{ (q'|_{N'})}=\eq{b}{(q''|_{N''})}$.
  \end{itemize}
  (We do not require the converse implications. E.g.\
  $ q|_N\trans{a} q'|_{N'}$ need not imply that $q\trans{a}q'$, only
  that $\pi q\trans{a}q'$ for some $\pi\in \Fix(N)$; see
  also~\eqref{eq:free-suc-char} below.) Transitions are clearly
  equivariant. Moreover, bound transitions are, by construction,
  $\alpha$-invariant.
  \begin{fact}\label{fact:name-drop-bound}
    By construction, every bound transition in $A'$ is
    $\alpha$-equivalent\footnote{Recall that a transition $q \trans{\scriptnew a} q'$ is $\alpha$-equivalent to a transition $r \trans{\scriptnew b} r'$ if $q = r$ and $\langle a \rangle q' = \langle b \rangle r'$.} to one of the form
    $q|_N\trans{\scriptnew a} q'|_{N'}$ where
    $q\trans{\scriptnew a}q'$ and
    $N'\subseteq\supp(q')\cap(N\cup\{a\})$.
  \end{fact}
  (2)~To see ufs branching, let $q|_N$ be a state in $A'$. For free
  transitions, we have to show that the set
  \begin{equation*}
    \freesuc(q|_N)
    =\{(a,q'|_{N'})\mid N'\subseteq N, a\in N, \pi q\trans{a} q'
    \text{ for some $\pi\in\Fix(N)$}\}
  \end{equation*}
  of free successors of $q|_N$ is ufs. But for
  $\pi\in\Fix(N)$, $N'\subseteq N$, and $a\in N$, we have
  $\pi q\trans{a} q'$ iff $q\trans{a}\pi^{-1}q'$, and then moreover
  $\pi^{-1}\in\Fix(N')$ so $\Fix(N')\pi^{-1}q'=\Fix(N')q'$, i.e.~$q'|_{N'} = (\pi^{-1}q')|_{N'}$. We thus
  have
  \begin{align}
    \freesuc(q|_N) & = \{(a, (\pi^{-^1}q')|_{N'} \mid N' \subseteq N, a \in N, q \trans{a} \pi^{-1} q'\} \nonumber \\
    & = \{(a,q'|_{N'})\mid N'\subseteq N, a\in N, q\trans{a} q'\}, \label{eq:free-suc-char}
  \end{align}
  which is ufs. 

  We proceed similarly for the bound transitions: We need to show that
  the set $\boundsuc(q|_N)$ of bound successors of $q|_N$ is ufs. By
  Fact~\ref{fact:name-drop-bound}, a bound transition
  $q|_N\trans{\scriptnew a}q'|_{N'}$ arises from $\pi\in\Fix(N)$ and
  $N'\subseteq \supp(q')\cap(N\cup\{a\})$ such that
  $\pi q\trans{\scriptnew a} q'$. Then
  $q\trans{\scriptnew{(\pi^{-1}a)}}\pi^{-1}q'$. Moreover, we claim
  that
  \begin{equation}\label{eq:name-drop-bound-suc}
    \eq{a}{(q'|_{N'})}=\eq{\pi^{-1}a}{(\pi^{-1}(q'|_{N'}))}.
  \end{equation}
  To see \eqref{eq:name-drop-bound-suc}, we distinguish two cases: If
  $\pi^{-1}(a)=a$ then the two sides are equal because $\pi^{-1}$
  fixes the support of $q'|_{N'}$. If $\pi^{-1}(a)\neq a$ then
  $\pi^{-1}a\notin N$ because $\pi^{-1}$ fixes $N$, so
  $\pi^{-1}a\notin N\cup\{a\}$ and therefore
  $\pi^{-1}a\notin N' = \supp(q'|_{N'})$. This means that we can
  $\alpha$-equivalently rename $a$ into $\pi^{-1}a$ in
  $(\newletter{a},q'|_{N'})$; since $\pi^{-1}$ fixes $N$, the result
  of this renaming equals $(\newletter\pi^{-1}a,\pi^{-1}q'|_{N'})$. Since
  $\pi^{-1}(q'|_{N'})=(\pi^{-1}q)|_{\pi^{-1}N'}$ and
  $\pi^{-1}N'\subseteq\supp(\pi^{-1}q')\cap(N\cup\{\pi^{-1}(a)\})$
  (recall $\pi\in\Fix(N)$),~\eqref{eq:name-drop-bound-suc} proves
  \begin{align}\label{eq:bound-suc-char}
     \boundsuc(q|_N)& = \{\eq{\pi^{-1}a}{((\pi^{-1}q')|_{\pi^{-1}N'})}\mid \nonumber\\
    & \hspace{4em}\pi\in\Fix(N), \pi^{-1}N'\subseteq\supp(\pi^{-1}q')\cap(N\cup\{\pi^{-1}a\}),\\
    & \hspace{4em}q\trans{\scriptnew \pi^{-1}a} \pi^{-1}q'\} \nonumber\\
    &=\{\eq{a}{(q'|_{N'})}\mid N'\subseteq\supp(q')\cap(N\cup\{a\}), 
    q\trans{\scriptnew a} q'\}.
  \end{align}
  By~\eqref{eq:bound-suc-char}, $\boundsuc(q|_N)$ is ufs; indeed, we have $\supp(\eq{a}{(q'_{N'})}) = N' - \{a\}$ so the support of every element of $\boundsuc(q|_N)$ is a subset of $N$. (Note
  that~\eqref{eq:bound-suc-char} is not the same as
  Fact~\ref{fact:name-drop-bound}, as in~\eqref{eq:bound-suc-char} we
  use a fixed representative $q$ of $\Fix(N)q$.)
 
  (3)~We show next that $A'$ is name-dropping. So let $q|_N$ be a
  state in $A'$, and let $N'\subseteq\supp(q|_N) = N$. We show that
  $q|_{N'}$ restricts $q|_N$ to~$N'$. We first establish that
  $q|_{N'}$ has at least the same incoming transitions as $q|_N$. For
  the free transitions, let $\pi\in\Fix(N)$, $q'\trans{a}\pi q$ and
  $a\in N''\supseteq N$, so that $q'|_{N''}\trans{a}(\pi
  q)|_N=q|_N$.
  Then also $\pi\in\Fix(N')$ and $N'\subseteq N''$, so
  $q'|_{N''}\trans{a}\pi q|_{N'}=q|_{N'}$ as required. For the bound
  transitions, let $\pi\in\Fix(N)$, let
  $\eq{a}{(q|_N)}=\eq{b}{((ab)\cdot(q|_N))}$, let
  $q'\trans{\scriptnew b}\pi(ab)\cdot q$ where $\pi\in\Fix((ab)N)$,
  and let $(ab)\cdot N\subseteq N''\cup\{b\}$, so that
  $q'|_{N''}\trans{\scriptnew a}q|_N$. We have to show that
  $q'|_{N''}\trans{\scriptnew a}q|_{N'}$. From
  $q'\trans{\scriptnew b}\pi(ab)\cdot q$ we have
  $q|_{N''}\trans{\scriptnew
    b}((\pi(ab)q)|_{(ab)N'})=(((ab)q)|_{(ab)N'})$,
  because $(ab)N'\subseteq(ab)N\subseteq N''\cup\{b\}$ and
  $\pi\in\Fix((ab)N)\subseteq\Fix((ab)N')$. If $b=a$, we are done. So
  assume $b\neq a$. Since $A'$ is $\alpha$-invariant, it remains only
  to show that
  \begin{equation*}
    \eq{b}{((ab)q|_{(ab)N'})}=\eq{a}{(q|_{N'})},
  \end{equation*}
  i.e.\ that $b\notin\supp(q|_{N'})$; but since $b\neq a$ and
  $\eq{a}{(q|_N)}=\eq{b}{((ab)\cdot (q|_N))}$, we even have
  $b\notin\supp(q|_N)\supseteq N'\supseteq\supp(q|_{N'})$.

  Next, we show that $q|_{N'}$ has the requisite outgoing
  transitions. For the free transitions, let $q|_N\trans{a}q'|_M$
  where $\supp(q'|_M)\cup\{a\}=M\cup\{a\}\subseteq N'$. We have to
  show $q|_{N'}\trans{a}q'|_M$. By~\eqref{eq:free-suc-char}, we have
  $q\trans{a}\pi q'$ for some $\pi\in\Fix M$. By construction of $A'$,
  $q|_{N'}\trans{a}(\pi q')|_M=q'|_M$, as required. For the bound
  transitions, we proceed as follows. By~\eqref{eq:bound-suc-char}, a
  given outgoing bound transition of $q|_N \trans{\scriptnew a} q'|_M$
  yields a state $r|_S$ of $A'$ and $b\in \Names$ such that
  $\eq a {q'|_M} = \eq b {r|_S}$,
  $S \subseteq \supp(r) \cap (N \cup \{a\})$ and
  $q \trans{\scriptnew b} r$.

  Now if $M \subseteq N' \cup\{a\}$ this
  yields a transition $q|_{N'}\trans{\scriptnew a}q'|_M$ by
  construction of $A'$; indeed, we already have
  $q \trans{\scriptnew b} r$, $S \subseteq \supp(r)$ and
  $\eq a {q'|_M} = \eq b {r|_S}$, so it remains to show that
  $S \subseteq N' \cup \{b\}$. By Lemma~\ref{lem:abstr},
  $\eq a {q'|_M} = \eq b {r|_S}$ iff either $a = b$ and $q'|_M = r|_S$
  (and the latter yields $M = S$, thus we are done), or $a \neq b$,
  $a \fresh r|_S$ and $(b a) (r|_S) = q'|_M$. It follows that
  $a \not\in S$ and $M = (b a)S$. Since
  $(b a) S \subseteq N' \cup \{a\}$ we have equivalently
  $S \subseteq (ba) N' \cup \{b\}$. This implies
  $S \subseteq N' \cup \{b\}$ using that $a \not\in S$.


  (4)~It remains to show that $\barlang(A')=\barlang(A)$. To show
  `$\subseteq$', we show that $[w]_\alpha\in\barlang(q)$ for every
  state $q|_N$ in $A'$ and every $w\in\litlang(q|_N)$, by induction on
  $w$: for the empty word, the claim follows from the definition of
  final states in $A'$. For $w=\alpha v$, let
  $q|_N\trans{\alpha}q'|_{N'}\trans{v}t$ be an accepting path
  in~$A'$. Then we have $[v]_\alpha \in\barlang(q')$ by induction hypothesis and 
  $\FN(v)\subseteq N'$ by Corollary~\ref{cor:lang-ufs}. By~\eqref{eq:free-suc-char}
  and~\eqref{eq:bound-suc-char}, we have $q\trans{\alpha}\pi q'$ for
  some $\pi\in\Fix N'$. It follows that $\pi \cdot v = v$ and therefore 
  $[v]_\alpha = \pi[v]_\alpha\in\barlang(\pi q')$ by the equivariance of $L_\alpha$ (see Lemma~\ref{lem:lang-equivar}). Hence
  $[\alpha v]_\alpha\in\barlang(q)$.

  To see $\barlang(A')\supseteq\barlang(A)$, it suffices to note that
  $A$ is included as a subautomaton in $A'$ via the map that takes
  $q$ to $q|_{\supp(q)}$, i.e.\ $q\trans{\alpha}q'$ in $A$ implies
  $q|_{\supp(q)}\trans{\alpha}q'|_{\supp(q')}$ in $A'$.
\end{proof}

\begin{lem}\label{lem:nfba-ufs}
  Let $q$ be a state in a bar NFA; then $\barlang(q)$ is ufs.
\end{lem}
\begin{proof}
  The finitely many transitions of $A$ only mention letters from a
  finite subset of $\bar\Names$, and
  $\bigcup_{w \in \barlang(q)} \supp(w)$ is contained in that finite
  subset.
\end{proof}

\myparagraph{Details for Theorem~\ref{thm:freernna}}

As indicated in the text, we split the construction into two parts,
and first construct a plain RNNA~$\tilde A$.  The states of $\tilde A$
are pairs
\begin{equation*}
  (q,\pi H_q)\quad\text{where $H_q=\Fix(\supp(\barlang(q)))$}
\end{equation*}
consisting of a state $q$ in $A$ and a left coset $\pi H_q$, where the
action of $G$ is as on $\bar A$: 
\[
\pi_1 \cdot (q, \pi_2H_q) = (q, \pi_1\pi_2H_q).
\]
We continue to write
$N_q=\supp(\barlang(q))$ (note $H_q=F_{N_q}$ in the notation
used in the construction of $\bar A$). The initial state of $\tilde A$
is $(s,H_s)$ where $s$ is the initial state of $A$; a state
$(q,\pi H_q)$ is final in $\tilde A$ iff $q$ is final in $A$. Free
transitions in $\tilde A$ are of the form
\begin{equation*}
  (q,\pi H_q)\trans{\pi(a)}(q',\pi H_{{q'}})\quad\text{where } q\trans{a}q'\text{ and }a\in N_q,
\end{equation*}
(where the condition $a\in N_q$ is automatic unless
$\barlang(q')=\emptyset$) and bound transitions are of the form
\begin{equation*}
  (q,\pi H_q)\trans{\scriptnew{a}}(q',\pi'H_{{q'}})\;\;\text{where}
  \;\; q\trans{\scriptnew{b}}q'\;\;\text{and} \;\; \eq{a}{\pi' H_{{q'}}} = \eq{\pi(b)}{\pi H_{{q'}}}.
\end{equation*}
\begin{rem}\label{rem:coset}
  \begin{enumerate}
  \item Note that by Lemma~\ref{lem:supp},
    $N_q=\supp(\barlang(q)) = \bigcup_{w \in \barlang(q)} \supp(w)$, i.e.\
    $N_q$ is the set of names that appear free in some word
    $w \in \barlang(q)$.
  \item\label{item:name-inj} Observe that $\pi H_q = \pi' H_q$ iff
    $\pi'(v) = \pi(v)$ for all $ v \in N_q$: $\pi H_q = \pi' H_q$ iff
    $\pi^{-1}\pi'\in H_q$ iff $\pi^{-1}\pi'(v)=v$ for all $v\in N_q$ iff
    $\pi'(v) = \pi(v)$ for all $v \in N_q$. 
  \item For a coset $\pi H_q$, we have 
    \begin{equation*}
      \supp(q, \pi H_q) = \supp(\pi H_q)=\pi N_q
    \end{equation*}
    so the set $\tilde Q$ of states of $\tilde A$ is a nominal set. This
    is by Item~(\ref{item:name-inj}): for $\pi'\in G$, we have
    $\pi'\pi H_q=\pi H_q$ iff $\pi'\pi(a)=\pi(a)$ for all $a\in N_q$ iff
    $\pi'\in\Fix(\pi N_q)$.
  \item Note that $\eq a {(\pi H_q)} = \eq b {(\pi' H_q)}$ implies $\eq a {(q,\pi H_q)} = \eq b {(q,\pi' H_q)}$ since the action of $G$ on states of $\tilde A$ is trivial in the first component.
\end{enumerate}
\end{rem}

\begin{rem}\label{rem:ext}
  Left cosets for $H_q$ are in one-to-one correspondence with injections $N_q \to \Names$. Indeed, in the light of Remark~\ref{rem:coset}(\ref{item:name-inj}) it suffices to prove that every injection $i: N_q \to \Names$ can be extended to a finite permutation. Define $\pi$ by
    \[
      \pi(a) = 
      \begin{cases}
         i(a) & a\in N_q \\
           i^{-n}(a) & \text{else, for $n\geq 0$ minimal s.t.~$i^{-n}(a) \not\in i[N_q]$}  
      \end{cases}
    \]
    For the proof that $\pi$ is a indeed a finite permutation see~\cite[Corollary 2.4]{msw16}.
\end{rem}

Transitions from a given state $(q,\pi H_q)$ can be
characterized as follows.
\begin{lem}\label{lem:freernna-trans}
  Let $(q,\pi H_q)$ be a state in $\tilde A$. Then
  \begin{equation}\label{eq:freernna-free-char}
    \freesuc(q,\pi H_q)  = \{(\pi(a),(q',\pi H_{q'}))\mid
    q\trans{a}q', a\in N_q\}
  \end{equation}
  and
  \begin{equation}\label{eq:freernna-bound-char}
    \boundsuc(q,\pi H_q) = \{\eq{\pi(a)}{(q',\pi H_{q'})}\mid
    q\trans{\scriptnew a}q'\}.
  \end{equation}
\end{lem}
\begin{proof}
  For the free transitions, we have by definition
  \begin{equation*}
    \freesuc(q,\pi H_q) = \{(\pi'(a),(q',\pi' H_{q'}))\mid
    \pi' H_q=\pi H_q, q\trans{a}q',a\in N_q\}.
  \end{equation*}
  Now if $\pi' H_q=\pi H_q$ and $q\trans{a}q'$, then $\pi$ and $\pi'$
  agree on $N_q$ and hence on $N_{q'}\cup\{a\}$ (as
  $N_{q'}\cup\{a\}\subseteq N_q$), so
  $(\pi'(a),(q',\pi' H_{q'}))=(\pi(a),(q',\pi H_{q'}))$. This
  shows~\eqref{eq:freernna-free-char}.

  For the bound transitions, we have by definition and using Remark~\ref{rem:coset}(4)
  \begin{equation*}
    \boundsuc(q,\pi H_q) = \{\eq{\pi'(a)}{(q',\pi' H_{q'})}\mid
    \pi' H_q=\pi H_q, q\trans{\scriptnew a}q'\}.
  \end{equation*}
  So let $\pi' H_q=\pi H_q$ and $q\trans{\scriptnew a}q'$. The
  claim~\eqref{eq:freernna-bound-char} follows from
  \begin{equation}\label{eq:freernna-bound-trans}
    \eq{\pi'(a)}{(q',\pi' H_{q'})}=\eq{\pi(a)}{(q',\pi H_{q'})},
  \end{equation}
  which we now prove. By Remark~\ref{rem:coset}(\ref{item:name-inj}) we know that $\pi$ and $\pi'$ agree on $N_q$.  
  In order to prove~\eqref{eq:freernna-bound-trans}, we distinguish two cases:
  if $\pi(a)=\pi'(a)$ then $\pi$ and $\pi'$ agree on $N_{q'} \subseteq N_q \cup\{a\}$, i.e.\
  $\pi'H_{q'}=\pi H_{q'}$, so the two sides
  of~\eqref{eq:freernna-bound-trans} are literally equal. Otherwise,
  $a\notin N_q$, and $\pi'$, $\pi$ differ on $N_{q'}$ only w.r.t.\ their
  value on $a$. It follows that $(\pi(a)\,\pi'(a))\pi$ and $\pi'$ agree on $N_{q'} = \supp(H_{q'})$. Therefore $(\pi(a)\,\pi'(a))\pi H_{q'} = \pi' H_{q'}$ by Lemma~\ref{onlySuppMatters}. So, by Lemma~\ref{lem:abstr}, to show~\eqref{eq:freernna-bound-trans} it suffices to show that $\pi'(a)\notin\supp(q',\pi H_{q'})$. But
  $\supp(q',\pi H_{q'})=\pi N_{q'}\subseteq \pi N_q\cup\pi(a)=\pi'
  N_q\cup\pi(a)$,
  and $\pi'(a)\notin\pi' N_q\cup\pi(a)$ because $a\notin N_q$ and
  $\pi'(a)\neq\pi(a)$.
\end{proof}
The key ingredient in the proof that $\tilde A$ accepts the same bar
language as $A$ will be a normalization result on paths that uses an
obvious notion of $\alpha$-equivalence on paths in an RNNA (see
Remark~\ref{rem:paths}); explicitly:
\begin{defn} $\alpha$-equivalence of paths in an RNNA is defined
  inductively by
  \[
    q_0\trans{a}q_1\trans{\alpha_2}q_2\trans{\alpha_3}\cdots \trans{\alpha_n} q_n
    \
    \text{is $\alpha$-equivalent to}
    \
    q_0\trans{a}q'_1\trans{\alpha'_2}q'_2\trans{\alpha_3'}\cdots\trans{\alpha_n'} q'_n
  \]
  if $q_1\trans{\alpha_2}q_2\trans{\alpha_3}\cdots \trans{\alpha_n}q_n$ is $\alpha$-equivalent to
  $q'_1\trans{\alpha'_2}q'_2\trans{\alpha_3}\cdots \trans{\alpha_n'} q'_n$, and
  \[
    q_0\trans{\scriptnew{a}}q_1\trans{\alpha_2}q_2\trans{\alpha_3}\cdots \trans{\alpha_n} q_n
    \ 
    \text{is $\alpha$-equivalent to}
    \ 
    q_0\trans{\scriptnew{b}}q'_1\trans{\alpha'_2}q'_2\trans{\alpha_3'}\cdots\trans{\alpha_n'} q'_n
  \] if $\eq a {[q_1\trans{\alpha_2}q_2\trans{\alpha_3} \cdots\trans{\alpha_n} q_n]_\alpha} = \eq b {[q'_1\trans{\alpha'_2}q'_2\trans{\alpha_3'} \cdots\trans{\alpha_4'} q'_n]_\alpha}$, 
  where we use $[-]_\alpha$ to denote $\alpha$-equivalence classes of paths. 
\end{defn}

\begin{lem}\label{lem:pathsalpha}
  The set of paths of an RNNA is closed under $\alpha$-equivalence.
\end{lem}
\begin{proof}
  Observe that by equivariance, $G$ acts pointwise on paths.  It
  suffices to show closure under single $\alpha$-conversion steps. So
  let
  $q_0\trans{\scriptnew a}q_1\trans{\alpha_2}\dots\trans{\alpha_n}q_n$
  be path in an RNNA $A$, denote the path from $q_1$ onwards by $P$,
  and let $\eq{a}{P}=\eq{b}{P'}$, so $P'=(ab)\cdot P$. Then by
  $\alpha$-invariance of $\to$, we have
  $q_0\trans{\scriptnew b}(ab)q_1$, and by equivariance, $(ab)\cdot P$
  is a path from $(ab)q_1$.
\end{proof}
\begin{lem}\label{lem:path-alpha}
  Let
  $P=q_0\trans{\scriptnew{a}}q_1\trans{\alpha_2}\dots\trans{\alpha_n}q_n$
  be a path in an RNNA $A$, and let
  $\langle a\rangle q_1=\langle{b}\rangle q_1'$. Then there exists a
  path in $A$ of the form
  $q_0\trans{\scriptnew{b}}q_1'\trans{\alpha_2'}\dots\trans{\alpha_n'}q_n'$
  that is $\alpha$-equivalent to $P$.
\end{lem}

\begin{proof}
  Since $A$ is an RNNA, the support of the $\alpha$-equivalence class
  of $q_1\trans{\alpha_2}\dots\trans{\alpha_n}q_n$ is $\supp(q_1)$
  (Remark~\ref{rem:eqpath}), so we obtain an $\alpha$-equivalent path
  $q_0\trans{\scriptnew{b}}q_1'\trans{\alpha_2'}\dots\trans{\alpha_n'}q_n'$
  by renaming $a$ into $b$ in
  $q_1\trans{\alpha_2}\dots\trans{\alpha_n}q_n$.
\end{proof}
\begin{rem}\label{rem:eqpath}
  Note that the support of the $\alpha$-equivalence class of a path in
  an RNNA is the support of its starting state. Indeed, let
  $[P]_\alpha$ be such an equivalence class and let $q$ be the
  starting state of $P$. The inclusion
  $\supp(q) \subseteq \supp([P]_\alpha)$ holds because we have a
  well-defined equivariant projection from paths to their initial
  states. The converse inclusion is shown by induction, using
  Lemma~\ref{lem:trans}.
\end{rem}

In the proof that $\tilde A$ accepts $\barlang(A)$, the following
normalization result for paths is crucial.
\begin{defn}
  A path in $\tilde A$ is \emph{$\pi$-literal} for $\pi\in G$ if all
  transitions in it are of the form
  $(q,\pi H_q)\trans{\pi\alpha}(q',\pi H_{q'})$ where
  $\alpha\in\bar\Names$ and $q\trans{\alpha}q'$.
\end{defn}
Intuitively, a $\pi$-literal path is one that uses the same pattern of
name reusage for free and bound names as the underlying path in $A$,
up to a joint renaming $\pi$ of the free and bound names.
\begin{lem}\label{lem:literal-path}
  Let $P$ be a path in $\tilde A$ beginning at
  $(q_0,\pi_0H_{{q_0}})$. Then $P$ is $\alpha$-equivalent to a
  $\pi_0$-literal path.
\end{lem}
\begin{proof}
  We prove the statement by induction over the path length. The base
  case is trivial. For the inductive step, let
  $P=(q_0,\pi_0
  H_{{q_0}})\trans{\alpha_1}(q_1,\pi_1H_{{q_1}})\trans{\alpha_2}\cdots\trans{\alpha_n}
  (q_n, \pi_n H_{{q_n}})$
  be a path of length $n>0$. If $\alpha_1$ is a free name then
  $\pi_1 H_{{q_1}}=\pi_0 H_{{q_1}}$ by~\eqref{eq:freernna-free-char};
  by induction, we can assume that the length-$(n-1)$ path from
  $(q_1,\pi_0H_{{q_1}})$ onward is $\pi_0$-literal, and hence the
  whole path is $\pi_0$-literal. If $\alpha_1=\newletter{a}$ then
  by~\eqref{eq:freernna-bound-char} we have
  $\eq{a}{(q_1,\pi_1 H_{{q_1}})}=\eq{\pi_0(b)}{(q_1,\pi_0H_{{q_1}})}$
  for some transition $q_0\trans{\scriptnew{b}}q_1$ in $A$. By
  Lemma~\ref{lem:path-alpha}, this induces an $\alpha$-equivalence of
  $P$ with a path
  $(q_0,\pi_0H_{{q_0}})\trans{\scriptnew{\pi_0(b)}}(q_1,\pi_0H_{{q_1}})\trans{}\dots$;
  by the induction hypothesis, we can transform the length-$(n-1)$
  path from $(q_1,\pi_0H_{{q_1}})$ onward into a $\pi_0$-literal one,
  so that the whole path becomes $\pi_0$-literal as desired.
\end{proof}
\begin{lem}\label{lem:livereg}
  $\tilde A$ is an RNNA, with as many orbits as $A$ has states, and
  accepts the bar language $\barlang(A)$.
\end{lem}
\begin{proof}
  The free and bound transitions of $\tilde A$ are equivariant, and
  the bound transitions are $\alpha$-invariant by construction of the
  transition relation on $\tilde A$ (note that all states in the orbit
  of $(q,\pi H_q)$ have the form $(q,\pi'H_q)$). Every orbit of
  $\tilde A$ contains a state of the form $(q,\id H_q)$. This proves
  the claim on the number of orbits, which implies that $\tilde A$ is
  orbit-finite.  Finite branching is immediate from
  Lemma~\ref{lem:freernna-trans}. Thus, $\tilde A$ is an RNNA.
  
  It remains to show that $\barlang(\tilde A)=\barlang(A)$. The
  inclusion `$\supseteq$' is clear because $A$ is a
  subautomaton in $\tilde A$ via the inclusion map $f$ taking a state $q$ to
  $(q,\id H_q)$; i.e.\ $q\trans{\alpha}q'$ in $A$ implies
  $(q,\id H_q)\trans{\alpha}(q',\id H_{q'})$ in $\tilde A$. For the
  reverse inclusion, note that by Lemma~\ref{lem:literal-path}, every
  accepting path of $\tilde A$ is $\alpha$-equivalent to an
  $\id$-literal accepting path starting at the initial state
  $(s,\id H_s)$ of~$\tilde A$; such a path comes from an accepting
  path in~$A$ for the same bar string via the map $f$.
\end{proof}

\noindent
\textit{Proof of Theorem~\ref{thm:freernna}.}
By combining the above construction of $\tilde A$ with that of Lemma~\ref{lem:wlog-name-drop};
i.e.\ we show that $(\tilde A)'$ is isomorphic, and in fact equal, to $\bar A$:
A state in $(\tilde A)'$ has the form
\begin{equation*}
  (q,\pi H_q)|_N = (\Fix N)\cdot(q,\pi H_q)=(q,(\Fix N)\pi H_q)
\end{equation*}
for $N\subseteq \supp(q,\pi H_q)=\pi N_q$ (hence
$\pi^{-1}N\subseteq N_q$).  We claim that
\begin{equation}\label{eq:name-drop-state}
  (\Fix N)\pi H_q = \pi\Fix(\pi^{-1}N)(=\pi F_{\pi^{-1}N}).
\end{equation}
To see `$\subseteq$', let $\rho\in\Fix N$ and $\sigma\in H_q$. Then
$\pi^{-1}\rho\pi\in\Fix(\pi^{-1}N)$ and, since
$\pi^{-1}N\subseteq N_q$, $\sigma\in\Fix(\pi^{-1}N)$, so
$\pi^{-1}\rho\pi\sigma\in\Fix(\pi^{-1}N)$ and therefore
$\rho\pi\sigma=\pi\pi^{-1}\rho\pi\sigma\in\pi\Fix(\pi^{-1}N)$. 

For `$\supseteq$', let $\rho\in\Fix(\pi^{-1}N)$. Then
$\pi\rho\pi^{-1}\in\Fix(N)$, so to show $\pi\rho\in\Fix(N)\pi H_q$ it
suffices to show $(\pi\rho\pi^{-1})^{-1}\pi\rho\in\pi H_q$. But
$(\pi\rho\pi^{-1})^{-1}\pi\rho=\pi\in\pi H_q$.

This proves equality of the state sets. It remains to show that the
transitions in $(\tilde A)'$ and $\bar A$ are the same. The free
transitions in $(\tilde A)'$ are of the form
$(q,\pi H_q)|_{N}\trans{\pi(a)} (q',\pi H_{q'})|_{N'}$ where
$q\trans{a}q'$, $N'\subseteq\pi N_{q'}$, and
$N'\cup\{a\}\subseteq N\subseteq\pi N_q$;
by~\eqref{eq:name-drop-state}, they thus have, up to
$\alpha$-equivalence, the form
$(q,\pi F_{\pi^{-1}N})\trans{\pi(a)} (q',\pi F_{\pi^{-1}N'})$ where
$\pi^{-1}N'\subseteq N_{q'}$
$\pi^{-1}N'\cup\{a\}\subseteq \pi^{-1}N\subseteq N_q$, and hence are
the same as in $\bar A$.

The bound transitions in $(\tilde A)'$ are, up to
$\alpha$-equivalence, those of the form
$(q,\pi H_q)|_{N}\trans{\scriptnew\pi(a)} (q',\pi H_{q'})|_{N'}$ where
$q\trans{\scriptnew a}q'$, $N'\subseteq \{a\}\cup N$,
$N'\subseteq\pi N_{q'}$; and $N\subseteq\pi N_q$;
by~\eqref{eq:name-drop-state}, they thus have the form
$(q,\pi F_{\pi^{-1}N})\trans{\scriptnew\pi(a)} (q',\pi
F_{\pi^{-1}N'})$
where $\pi^{-1}N'\subseteq\{a\}\cup\pi^{-1}N$,
$\pi^{-1}N'\subseteq N_{q'}$, and $\pi^{-1}N\subseteq N_q$, and hence
again are the same as in $\bar A$. \qed

\myparagraph{Proof of Theorem~\ref{thm:nom-bar}}

We have to show that every accepting path in $A$ is
$\alpha$-equivalent to an accepting path in $A_0$. Note that $Q_0$
is closed under free transitions in $A$, so by
Lemma~\ref{lem:path-alpha}, it suffices to show that for every bound
transition $q\trans{\scriptnew{b}}q'$ in $A$ with $q\in Q_0$ we find
an $\alpha$-equivalent transition $q\trans{\scriptnew a}q''$ in
$A_0$. We distinguish the following cases.
\begin{itemize}
\item If already $b\in\Names_0$ then
  $\supp(q')\subseteq\supp(q)\cup\{b\}\subseteq\Names_0$, so
  $q'\in Q_0$ and we are done.
\item If $b\notin\Names_0$ and $b\notin\supp(q')$ then
  $\supp(q')\subseteq\supp(q)\subseteq\Names_0$. In particular, $q'$
  is already in $Q_0$ and $*$ is fresh for $q'$, so we can rename
  $b$ into $*$ and obtain an $\alpha$-equivalent transition
  $q\trans{|*}q'$ in $A_0$.
\item If $b\notin\Names_0$ and $b\in\supp(q')$ then
  $|\supp(q')\cap\Names_0|<k$, so that we can pick a name
  $a\in\Names_0$ that is fresh for $q'$. We put $q''=(ab)q'$; then
  $\langle b\rangle q'=\langle a\rangle q''$, and $q''\in Q_0$
  because
  $\supp(q'')=\{a\}\cup(\supp(q')-\{b\})\subseteq\{a\}\cup\supp(q)\subseteq\Names_0$;
  thus, $q\trans{\scriptnew{a}}q''$ is a transition in $A_0$.
\end{itemize}









\subsection{Proofs and Lemmas for Section~\ref{sec:name-dropping}}

\myparagraph{NOFAs as coalgebras} We show that the standard
description of NOFAs as repeated at the beginning of
Section~\ref{sec:name-dropping} is equivalent to the one as
$F$-coalgebras for $FX=2\times\powfs (\Names\times X)$. For the
direction from the standard description to $F$-coalgbras, recall that
the transition relation is assumed to be equivariant; therefore, the
map taking a state $q$ to $\{(a,q')\mid q \xrightarrow{a} q'\}$ is
equivariant, hence preserves supports and therefore ends up in $FQ$
where $Q$ is the set of states. Conversely, let $\xi:Q\to FQ$ be an
$F$-coalgebra with components $f:Q\to 2$,
$g:Q\to\powfs(\Names\times Q)$. Define the transition relation on $Q$
by $q\trans{a}q'$ iff $(a,q')\in g(Q)$, and make $q$ final iff
$f(q)=\top$. Then finality is equivariant by equivariance of $f$. To
see that the transition relation is equivariant let $q\trans{a}q'$ and
$\pi\in G$. Then $(\pi a,\pi q')\in\pi(g(q))=g(\pi(q))$ by
equivariance of $g$, i.e.\ $\pi q\trans{\pi a}\pi q'$. \qed

\begin{lem}\label{lem:nofa-restrict-equivar} 
  If $\pi\in G$ and $q|_N$ restricts a state $q$ in a NOFA to
  $N\subseteq\supp(q)$, then $\pi(q|_N)$ restricts $\pi q$ to $\pi N$.
\end{lem}
\begin{proof}
  We have $\pi N\subseteq\pi\supp(q)=\supp(\pi q)$, so the claim is
  well-formed.
  
  For the support of $\pi(q|_N)$, we have
  $\supp(\pi(q|_N))=\pi\supp(q|_N) = \pi N$ as required.

  By the equivariance of final states we have that $\pi(q|_N)$ is final if
  $q$ is final. 
  
  For incoming transitions, let $p\trans{a}\pi q$. Then
  $\pi^{-1}p\trans{\pi^{-1}a}q$ by equivariance, hence
  $\pi^{-1}p\trans{\pi^{-1}a}q|_N$ so that $p\trans{a}\pi(q|_N)$.
  
  For outgoing transitions, let $\pi q\trans{a}q'$ where
  $\supp(q')\subseteq \pi N\cup\{a\}$. Then
  $q\trans{\pi^{-1}a}\pi^{-1}q'$ by equivariance, and
  $\supp(\pi^{-1}q')\subseteq N\cup\pi^{-1}a$, so
  $q|_N\trans{\pi^{-1}a}\pi^{-1}q'$ and hence, by equivariance,
  $\pi(q|_N)\trans{a}q'$. 
  \qed
\end{proof}

\begin{proof}[Proposition~\ref{prop:nonspont}]
  In the first claim, `only if' is immediate by
  Lemma~\ref{lem:trans}. To see `if', let $A$ be a non-spontaneous and
  $\alpha$-invariant NOFA. We construct an RNNA $B$ with the same
  states as $A$, as follows.
  \begin{itemize}
  \item $q\trans{a}q'$ in $B$ iff $q\trans{a}q'$ in $A$ and
    $a\in\supp(q)$.
  \item $q\trans{\scriptnew{a}}q'$ in $B$ iff $q\trans{b}q''$ in $A$
    for some $b,q''$ such that $b\fresh q$ and 
    $\eq b q''=\eq a q'$.
  \end{itemize}
  The transition relation thus defined is clearly equivariant and
  $\alpha$-invariant. That for every $q$ the sets $\{(a,q') \mid q \trans{a} q'\}$
  and $\{\eq a q' \mid q \trans{\scriptnew a} q'\}$ are ufs (whence
  finite) easily follows from non-spontaneity.

  It remains to verify that $D(B)=A$, i.e.\ that
  \begin{equation*}
    q\trans{a}q'\text{ in }A\quad\text{iff}\quad (q\trans{a}q'\text{ or }q\trans{\scriptnew{a}}q'\text{ in }B).
  \end{equation*}
  To see the `only if' direction, let $q\trans{a}q'$ in $A$. If
  $a\in\supp(q)$ then $q\trans{a}q'$ in $B$. Otherwise, $a\fresh q$
  and hence $q\trans{\scriptnew a}q'$. For the `if' direction, we have
  two cases; the case where $q\trans{a}q'$ in $B$ is immediate by
  construction of $B$.  So let $q\trans{\scriptnew a}q'$ in $B$, that
  is, we have $q\trans{b}q''$ in $A$ for some $b,q''$ such that
  $\langle b\rangle q''=\langle a\rangle q'$ and $b\fresh q$. Then by
  $\alpha$-invariance of $A$, $q\trans{a}q'$.
  
  We proceed to prove the second claim, beginning with `only if'. So
  let $C$ be a name-dropping RNNA, let $q$ be a state, let
  $N\subseteq\supp(q)$, and let $q|_N$ restrict $q$ to $N$ in $C$. We
  show that $q|_N$ restricts $q$ to $N$ in $D(C)$. The condition
  $\supp(q|_N)\subseteq N$ is clear, as the nominal set of states is
  not changed by $D$. Since $q|_N$ has at least the same incoming
  transitions as $q$ in $C$, the same holds in $D(C)$. For the
  outgoing transitions, first let $q\trans{a}q'$ in $D(C)$ where
  $a\in\supp(q)$ and $\supp(q')\cup\{a\}\subseteq N$. Then either
  $q\trans{a}q'$ or $q\trans{\scriptnew a}q'$ in $C$. In the first
  case, $q|_N\trans{a}q'$ in $C$ and hence also in $D(C)$. In the
  second case, we have $\supp(q')\subseteq N\subseteq \{a\}\cup N$ and
  therefore $q|_N\trans{\scriptnew a}q'$ in $C$, so $q|_N\trans{a}q'$
  in $D(C)$. Second, let $q\trans{a}q'$ in $D(C)$ where $a\fresh q$
  and $\supp(q')\subseteq N\cup\{a\}$. By
  Lemma~\ref{lem:trans}.\ref{item:free-trans} we know that
  $q\trans{\scriptnew a}q'$ in $C$, so $q|_N\trans{\scriptnew a}q'$ in
  $C$ and hence $q|_N\trans{a}q'$ in $D(C)$.

  For the `if' direction of the second claim, let $A$ be a
  non-spontaneous, name-dropping, and $\alpha$-invariant NOFA. We
  construct $B$ such that $D(B)=A$ as for the first claim, and show
  additionally that $B$ is name-dropping. Let $q$ be a state, let
  $N\subseteq\supp(q)$, and let $q|_N$ restrict $q$ to $N$ in $A$. We
  claim that $q|_N$ also restricts $q$ to $N$ in $B$. We first show
  that $q|_N$ has at least the same incoming transitions as $q$ in
  $B$. For the free transitions, let $p\trans{a}q$ in $B$. Then by
  construction of $B$, $p\trans{a}q$ in $A$ and $a\in\supp(p)$, so
  since $A$ is name-dropping,
  $p\trans{a}q|_N$ in $A$ and hence $p\trans{a}q|_N$ in $B$. For the
  bound transitions, let $p\trans{\scriptnew a}q$ in $B$, i.e.\ we
  have $p\trans{b}q'$ in $A$ with $b\fresh p$ and
  $\eq{a}{q}=\eq{b}{q'}$, in particular $q'=(ab)q$. If $a=b$ then $p
  \trans{a} q$ in $A$, and since $A$ is name-dropping $p \trans{a}
  q|_N$ in $A$ whence in $B$. Otherwise, $b\fresh q$. By
  Lemma~\ref{lem:nofa-restrict-equivar}, $(ab)(q|_N)$ restricts $q'$
  to $(ab)N$ in $A$, so $p\trans{b}(ab)(q|_N)$ in $A$. Since
  $\supp(q|_N)\subseteq\supp(q)$, we have $b\fresh q|_N$, so
  $\eq{b}{((ab)(q|_N))}=\eq{a}{(q|_N)}$. By construction of $B$ we
  have $q \trans{\scriptnew a} q|_N$ in $B$, as required.

  For the outgoing transitions, first let $q\trans{a}q'$ in $B$ where
  $\supp(q')\cup\{a\}\subseteq N$. Then $q\trans{a}q'$ in $A=D(B)$, so
  $q|_N\trans{a}q'$ in $A$; since $a\in N=\supp(q|_N)$, it follows by
  construction of $B$ that $q|_N\trans{a}q'$ in $B$. Second, let
  $q\trans{\scriptnew a}q'$ in $B$ where
  $\supp(q')\subseteq N\cup\{a\}$. Pick $b\fresh (q,a)$ (so
  $b\notin N$); then $\eq{b}{((ab)q')}=\eq{a}{q'}$ and therefore
  $q\trans{\scriptnew b}(ab)q'$ in $B$ by $\alpha$-invariance. Thus,
  $q\trans{b}(ab)q'$ in $A=D(B)$, and
  $\supp((ab)q')=(ab)\supp(q')\subseteq (ab)(N\cup\{a\})= N - \{a\}
  \cup \{b\} \subseteq N\cup\{b\}$, where the last but one equation
  holds since $b \not\in N$. Therefore, $q|_N\trans{b}(ab)q'$ in
  $A$ since $A$ is name-dropping. By construction of $B$, it follows that
  $q|_N\trans{\scriptnew a}q'$ in $B$.
  \qed
\end{proof}

\begin{proof}[Proposition~\ref{prop:alphainv}]
  Let $A$ be a non-spontaneous and name-dropping NOFA. We construct a
  NOFA $\bar A$ by closing $A$ under $\alpha$-equivalence of
  transitions; that is, $\bar A$ has the same states as $A$ (in
  particular is orbit-finite), and its transitions are given by
  \begin{align*}
    q\trans{a}q'\text{ in $\bar A$ iff } 
    &  q\trans{a}q'\text{ in }A\text{ or}\\
    & \text{there exist } b,q''\text{ such
      that }q\trans{b}q''\text{ in $A$}, b\fresh q,\text{ and }\eq{a}{q'}=\eq{b}{q''}. 
  \end{align*}
  We say that a transition $q\trans{a}q'$ in $\bar A$ is \emph{new} if
  it is not in $A$.
  \begin{fact}\label{fact:alphainv-new-trans}
    If $q\trans{a}q'$ is new then $a\in\supp(q)$ and there exist
    $b,q''$ such that $q\trans{b}q''$ in $A$, $b\fresh q$ (so
    $a\neq b$), and $\eq{a}{q'}=\eq{b}{q''}$.
  \end{fact}
  We check that $\bar A$ has the requisite properties. First, the
  transition relation is clearly equivariant. Moreover, $\bar A$ is
  $\alpha$-invariant by construction.

  \emph{$\bar A$ is non-spontaneous:} It suffices to check new
  transitions $q\trans{a}q'$. By Fact~\ref{fact:alphainv-new-trans},
  we have $a\in\supp(q)$ and $b,q''$ such that $q\trans{b}q''$ in $A$,
  $b\fresh q$, and $\eq{a}{q'}=\eq{b}{q''}$. Since $A$ is
  non-spontaneous, $\supp(q'')\subseteq\supp(q)\cup\{b\}$. Let
  $c\in\supp(q')$ and $c\neq a$; we have to show $c\in\supp(q)$. Now
  $q'=(ab)\cdot q''$, so
  $\supp(q')=(ab)\cdot\supp(q'')\subseteq(ab)\cdot(\supp(q)\cup\{b\})$. Since
  $\eq{a}{q'}=\eq{b}{q''}$, we have $b\fresh q'$, so $c\notin\{a,b\}$;
  thus, $c\in(ab)\cdot(\supp(q)\cup\{b\})$
  implies $c\in\supp(q)\cup\{a\}$, hence $c\in\supp(q)$.

  \emph{$\bar A$ is name-dropping:} Let $N\subseteq\supp(q)$ for a
  state $q$, and let $q|_N$ restrict $q$ to $N$ in $A$; we show that
  $q|_N$ also restricts $q$ to $N$ in $\bar A$. The support of $q|_N$
  stays unchanged in $\bar A$, so we only have to check that $q|_N$
  retains the requisite transitions. Throughout, it suffices to check
  new transitions.

  For incoming transitions, let $p\trans{a}q$ in $\bar A$ be new,
  i.e.\ by Fact~\ref{fact:alphainv-new-trans} we have $a\in\supp(p)$,
  $p\trans{b}q'$ in $A$, $b\fresh p$ (hence $a\neq b$), and
  $\eq{b}{q'}=\eq{a}{q}$. Then $(ab)\cdot q=q'$. Therefore,
  $(ab)\cdot (q|_N)$ restricts $q'$ to $(ab)\cdot N$ in $A$ by
  Lemma~\ref{lem:nofa-restrict-equivar}. It follows that
  $p\trans{b}(ab)\cdot(q|_N)$ in $A$. Since $a\neq b$ and
  $\eq{b}{q'}=\eq{a}{q}$, we have $a\fresh q'$ and therefore
  $a\fresh((ab)\cdot(q|_N))$, so
  $\eq{b}{((ab)\cdot(q|_N))}=\eq{a}{(q|_N)}$ and therefore
  $p\trans{a}q|_N$ by construction of $\bar A$.

  For outgoing transitions, let $q\trans{a}q'$ be new in $\bar A$;
  i.e.\ by Fact~\ref{fact:alphainv-new-trans} we have $a\in\supp(q)$,
  $q\trans{b}q''$ in $A$, $b\fresh q$ (hence $a\neq b$) and
  $\eq{b}{q''}=\eq{a}{q'}$. Since $a\in\supp(q)$, we have to show that
  $q|_N\trans{a}q'$ in $\bar A$, assuming
  $\supp(q')\cup\{a\}\subseteq N$. From $b\fresh q$ we have
  $b\fresh q|_N$, so by construction of $\bar A$, it suffices to show
  $q|_N\trans{b}q''$ in $A$, which will follow once we show
  $\supp(q'')\subseteq N \cup \{b\}$. So let $b\neq c\in\supp(q'')$; we
  have to show $b\in N$. Now $a\neq b$ and $\eq{b}{q''}=\eq{a}{q'}$
  imply $a\fresh q''$, so $c\neq a$ and hence
  $c\notin\{a,b\}$. Therefore
  $c\in(ab)\cdot\supp(q'')=\supp((ab)\cdot q'')=\supp(q')\subseteq N$,
  as required.

  \emph{$\bar A$ is equivalent to $A$:} $L(A)\subseteq L(\bar A)$ is
  immediate as $A\subseteq\bar A$ by construction. For the reverse
  inclusion, we show that\footnote{For greater clarity we write $L(A, q)$ for $L(q)$ where
    $q$ is a state in $A$.}
  \begin{quote}
    $(*)$\quad whenever $w\in L(\bar A,q)$ then there exists
    $N\subseteq\supp(q)$ such that if $q|_N$ restricts $q$ to $N$ in
    $A$ then $w\in L(A,q|_N)$
  \end{quote}
  (in fact, $N$ will be such that $|\supp(q)-N|\le 1$). Since
  $\supp(s)=\emptyset$ for the initial state~$s$, this implies that
  $L(\bar A)\subseteq L(A)$.

  We prove $(*)$ by induction on $w$, with trivial induction base. So
  let $w=av$ and $q\trans{a}q'$ in $\bar A$ where $v\in L(\bar
  A,q')$.
  By induction, there is $N\subseteq\supp(q')$ such that
  $v\in L(A,q'|_N)$ whenever $q'|_N$ restricts $q'$ to $N$ in $A$. If
  $q\trans{a}q'$ in $A$ then $q\trans{a}q'|_N$ in $A$, so that
  $av\in L(A,q)$. The remaining case is that $q\trans{a}q'$ is new. By
  Fact~\ref{fact:alphainv-new-trans}, we have $a\in\supp(q)$ and
  $b,q''$ such that $q\trans{b}q''$ in $A$, $b\fresh q$ (so
  $a\neq b$), and $\eq{a}{q'}=\eq{b}{q''}$. We claim that whenever
  $q|_{N_a}$ restricts $q$ to $N_a:=\supp(q)-\{a\}$ in $A$ then
  $av\in L(A,q|_{N_a})$. It suffices to show
  \begin{equation}\label{eq:key-trans-alpha}
    q|_{N_a}\trans{a}q'|_N
    \text{ in $A$}.
  \end{equation}
  Since $a\neq b$ and $\eq{a}{q'}=\eq{b}{q''}$, we have $a\fresh q''$
  so from $q\trans{b}q''$ in $A$ we obtain
  $\supp(q'')\subseteq\{b\}\cup N_a$ by non-spontaneity of $A$. By the
  definition of restriction, it follows that $q|_{N_a}\trans{b}q''$
  in~$A$ (recall that $b\fresh q$). Since $a\notin\supp(q|_{N_a}) = N_a$, we
  obtain by equivariance of transitions that $q|_{N_a}\trans{a}q'$,
  which implies~\eqref{eq:key-trans-alpha} by the definition of
  restriction: we have $q' = (ab)\cdot q''$ which implies
  \[
    \supp(q') = (ab) \cdot \supp(q'') \subseteq \{a\} \cup (ab) \cdot
    N_a = \{a\}  \cup N_a,
  \]
  where the last step holds since $a, b \not\in N_a$. 
  \qed
\end{proof}

\begin{proof}[Additional proof details for
  Corollary~\ref{cor:intersections}]
  It is straightforward to verify that non-spontaneous name-dropping
  NOFAs are closed under the standard product construction;
  specifically, given a state $(q_1,q_2)$ in a product automaton and
  $N\subseteq\supp(q_1,q_2)=\supp(q_1)\cup\supp(q_2)$, one checks
  readily that if $q_i|_{N_i}$ restricts $q_i$ to
  $N_i:=N\cap\supp(q_i)$ for $i=1,2$, then $(q_1|_{N_1},q_2|_{N_2})$
  restricts $(q_1,q_2)$ to $N$.
  \qed
\end{proof}

\myparagraph{Details for Remark~\ref{rem:local-expressivity}} We
  show that the data language
  \begin{equation*} 
    L=\{wava \mid w,v\in\Names^*,a\in\Names\} 
  \end{equation*}
  is not accepted by any DOFA. Assume for a contradiction that $A$ is
  a DOFA that accepts $L$. Let $n$ be the maximal size of a support of
  a state in $A$. Let $w=a_1\dots a_{n+1}$ for distinct $a_i$, and let
  $q$ be the state reached by $A$ after consuming $w$. Then there is
  $i\in\{1,\dots,n+1\}$ such that $a_i\notin\supp(q)$. Pick a fresh
  name $b$. Then $\delta(a_i,q)$ is final and $\delta(b,q)$ is not;
  but since $\delta(a_i,q) =(a_ib)\cdot\delta(b,q)$, this is in
  contradiction to equivariance of the set of final states. \qed

  \myparagraph{Details on Name-Dropping Register Automata.}  By
  definition, the full class of name-dropping RAs is characterized by
  the existence, for every configuration $(c,w)$ and subset $N$ of the
  names appearing in $w$, of a configuration $(c,w)|_N$ that restricts
  $(c,w)$ to $N$. A priori, nothing excludes the possibility that
  $(c,w)|_N$ uses a location other than $c$. We do not anticipate that
  any expressivity can be gained from that (we refrain from proving
  this formally as we currently wish to give lower rather than upper
  estimates for the expressivity of name-dropping RAs) and therefore
  concentrate on name-dropping RAs where restrictions use the same
  location; these are the \emph{forgetful} RAs mentioned in
  Section~\ref{sec:name-dropping}. Since there is only one transition
  constraint between any two locations, forgetful RAs are given by a
  condition concerning the individual transition constraints
  $c\xrightarrow{\phi}c'$; recall that $\phi$ is an equivariant subset
  of $R\times\Names\times R$ where $R=(\Names\cup\{\bot\})^n$ is the
  set of register assignments. For a register assignment
  $w=(w_1,\dots,w_n)$ write $|w|=\{w_1,\dots,w_n\}\cap\Names$, and for
  $N\subseteq\Names$ let $w|_N$ be defined as in
  Section~\ref{sec:name-dropping}, i.e.\ by $(w|_N)_i=w_i$ for $w_i\in N$ and
  $(w|_N)_i=\bot$ otherwise. Explicitly, an RA is forgetful iff every
  transition constraint $c\xrightarrow{\phi}c'$ in it satisfies
\begin{myenumerate}
\item \emph{Non-spontaneity:} $(w,a,v)\in\phi$ implies
  $|v|\subseteq\{a\}\cup|w|$.
\item\label{item:contraint-nd} \emph{Name-dropping:} For all
  $(w,a,v)\in\phi$,
  \begin{enumerate}
  \item\label{item:contraint-nd-in} $(w,a,v|_N)\in\phi$ for all
    $N\subseteq|v|$;
  \item\label{item:contraint-nd-out-free} if
    $\{a\}\cup|v|\subseteq N\subseteq |w|$ then $(w|_N,a,v)\in\phi$;
  \item\label{item:contraint-nd-out-bound} if $a\notin |w|\supseteq N$
    and $|v|\subseteq\{a\}\cup N$ then $(w|_N,a,v)\in\phi$.
  \end{enumerate}
\end{myenumerate}
Constraints satisfying
Conditions~\ref{item:contraint-nd-in}--\ref{item:contraint-nd-out-bound}
are clearly closed under unions and intersections. To see that
positive Boolean combinations of the four types of basic constraints
$\mathsf{cmp}_i$, $\mathsf{store}_i$, $\mathsf{fresh}_i$,
$\mathsf{keep}_{ji}$ given in Section~\ref{sec:name-dropping} satisfy these
conditions, it thus suffices to show that the basic constraints
satisfy them. This is immediate for
Condition~\ref{item:contraint-nd-in} as $\mathsf{store}_i$ and
$\mathsf{keep}_{ji}$ both allow
$v_i=\bot$. Conditions~\ref{item:contraint-nd-out-free}
and~\ref{item:contraint-nd-out-bound} are immediate for
$\mathsf{store}_i$ and $\mathsf{keep}_{ji}$ (in the case of
$\mathsf{keep}_{ji}$, again because it allows $v_i=\bot$). To see that
they hold for $(w,a,v)\in\mathsf{fresh}_i$, just note that
$(w|_N)_i\in\{w_i,\bot\}$ and
$a\neq\bot$. Condition~\ref{item:contraint-nd-out-bound} does not
apply to $(w,a,v)\in\mathsf{cmp}_i$. Finally,
Condition~\ref{item:contraint-nd-out-free} holds for
$(w,a,v)\in\mathsf{cmp}_i$ because $a\in N$ and $a=w_i$ imply that
$(w|_N)_i=w_i=a$.

\myparagraph{Translation of FSUBAs into RNNAs.}
We now compare RNNAs to \emph{finite-state unification-based automata
  (FSUBAs)}~\cite{KaminskiTan06,TalMScThesis}. A particular feature
of FSUBAs is that they distinguish a finite subset $\Theta$ of the
alphabet that is \emph{read-only}, i.e.\ cannot be written into the
registers. We have no corresponding feature, therefore restrict to
$\Theta=\emptyset$ in the following discussion. An FSUBA then consists of finite
sets $Q$ and $r$ of states and registers, respectively, a transition
relation $\mu \subseteq Q×r×\powf(r)×Q$, an initial state $q_0 \in Q$,
a set $F \subseteq Q$ of final states, and an initial register
assignment $u$. Register assignments are partial maps
$v:r\partialto \A$, which means a register $k \in r$ can be empty
($v(k) = \bot$) or hold a name from $\A$. An \emph{FSUBA
  configuration} is a pair $(q, v)$, where $q \in Q$ and $v$ is a
register assignment. The initial configuration is $(q_0, u)$. A
transition $(q,k,S,p)\in \mu$ applies to a configuration with state
$q$ for an input symbol $a\in \A$ if register $k$ is empty or holds
$a$; the resulting configuration has state $p$, with the input $a$
first written into register $k$ and the register contents from $S$
cleared afterwards. A word is accepted if there is a sequence of
transitions from $(q_0,u)$ to a configuration with a final state.

As the name \emph{unification-based} suggests, FSUBAs can check
equality of input symbols, but not inequality (except with respect to the
read-only letters); in other words, they have no notion of
freshness. Thus the above-mentioned language $\{aba\mid a\neq b\}$
cannot be accepted by an FSUBA~\cite{KaminskiTan06}.


We proceed to show that for every FSUBA $A$ with empty read-only
alphabet its configurations form an RNNA that accepts the same
$\Names$-language as $A$ under local freshness semantics; that is,
RNNAs are strictly more expressive than FSUBAs with empty read-only
alphabet.

Let $A$ be an FSUBA with set $Q$ of state, set $r$ of registers,
initial state $q_0$, set $F$ of final states, transition relation
$\mu \subseteq Q×r×\powf(r)×Q$, and initial register assignment $u$;
we restrict the read-only alphabet $\Theta$ to be empty. We denote the
$\Names$-language accepted by $A$ by $L(A)$.  We construct an
equivalent RNNA $R(A)$ as follows. The states of $R(A)$ are the
configurations of $A$, which form a nominal set $C$ under the group
action $\pi\cdot(q,v) = (q,\pi\cdot v)$. The transitions of $R(A)$ are
given by
\begin{align}
    \op{fsuc}(q,v) &=
        \{ (v(k),(p,\op{erase}_S(v))) \mid
            (q,k,S,p)\in \mu
        \}
        \label{eq:fsubaRead}
        \\&\cup
        \ %
        \{
        \label{eq:fsubaWriteUsed}
            (a,(p,\op{erase}_S(v[k\mapsto a]))) \mid
            (q,k,S,p)\in \mu,
            a\in \supp(v),
            v(k) = \bot
        \}
        \\
    \op{bsuc}(q,v) &=
        \{
        \label{eq:fsubaWriteFresh}
            \eq{a}{(p,\op{erase}_S(v[k\mapsto a]))}\mid
            (q,k,S,p)\in \mu,
            a\fresh v, v(k) =\bot
        \}
\end{align}
where $\op{erase}_S$ clears the contents of the registers in $S$.

This RNNA $R(A)$ behaves, under local freshness semantics, like the
FSUBA $A$:
\begin{lem}\label{fsubaToRnna}
  The transitions between configurations of $A$ are precisely given by
  $(q,v)\trans{\unbar(\alpha)} (p,w)$, where
  $(q,v)\trans{\alpha} (p,w)$, $\alpha \in \bar \A$, is a transition
  in $R(A)$.
\end{lem}
\begin{proof}
  Let $(q,v) \trans{\alpha} (p,w)$ be a transition in the RNNA
  $R(A)$. We distinguish cases:
    \begin{itemize}
    \item For \eqref{eq:fsubaRead}, we have an FSUBA transition
      $(q,k,S, p) \in \mu$ with $\alpha = v(k)\in \A$, and $w =
    \op{erase}_S(v)$. Hence we have a transition $(q,v)
    \trans{\op{ub}(\alpha)} (p,w)$ between FSUBA configurations. 

  \item For \eqref{eq:fsubaWriteUsed}, we have an FSUBA transition
    $(q,k,S,p) \in \mu$ and $v(k) =\bot$, $\alpha = v(i) \in \A$ for
    some $i\in r$, and $w = \op{erase}_S(v[k \mapsto v(i)])$.  Hence,
    from the FSUBA configuration $(q,v)$ the input $v(i)$ is read into
    register $k$ and then the registers in $S$ are cleared,
    i.e.~$(q,v) \trans{\alpha} (p,w)$ is a transition of FSUBA
    configurations.

    \item For \eqref{eq:fsubaWriteFresh}, i.e.~for $\alpha =
      \newletter{a}$, we have an FSUBA transition
      $(q,k,S,p') \in \mu$ and $v(k) =\bot$ and some $b \fresh v$ with
      $\eq a {(p,w)} = \eq b {(p',\op{erase}_S(v[k \mapsto b]))}$. It
      follows that $(ab) (p,w) = (p',\op{erase}_S(v[k \mapsto b]))$,
      and equivalently, $p = p'$ and $(ab) w = \op{erase}_S(v[k
      \mapsto b])$. The latter implies that $w = \op{erase}_S(v[k
      \mapsto a])$. Thus, we obtain a transition of FSUBA
      configurations $(q,v) \trans{a} (p,w)$ as desired. 
    \end{itemize}

    \noindent Conversely, consider a transition
    $(q,v) \trans{a} (p,w)$ of FSUBA configurations admitted by
    $(q,k,S,p) \in \mu$.
    \begin{itemize}
    \item If $v(k)\neq \bot$, then $v(k) = a$. Hence $(q,v) \xrightarrow{a}
           (p,w)$ is a transition in $R(A)$ by \eqref{eq:fsubaRead}.
    \item If $v(k) = \bot$ and $a\in \supp(v)$, then $(q,v) \xrightarrow{a}
           (p,w)$ is a transition in $R(A)$ by \eqref{eq:fsubaWriteUsed}.
         \item If $v(k) = \bot$ and $a\fresh v$, then
           $w = \op{erase}_S(v[k \mapsto a])$ and
           $\eq a {(p,w)} \in \op{bsuc}(q,v)$. By $\alpha$-invariance,
           this implies $(q,v) \trans{\scriptnew a} (p,w)$ in $R(A)$.
           \qedhere
    \end{itemize}
\end{proof}

\noindent Using Lemma~\ref{fsubaToRnna}, one shows by induction on $w$
that $L(A)=\{\unbar(w)\mid w\in L_0(R(A))\}$.  The RNNA $R(A)$ in
general fails to be name-dropping, but for any
$[\newletter{a}w]_\alpha \in \barlang(q,v)$, $w\in \bar \A^*$, we have
\begin{equation}\label{eq:fsuba-trans-alternative}
    (q,v) \xrightarrow{\scriptnew{a}} (p,v'),
    w\in \barlang(p,v')
    \quad\text{or}\quad
    (q,v) \xrightarrow{a} (p,v'),
    w\in \barlang(p,v'):
  \end{equation}
  Since $[\newletter{a}w]_\alpha \in \barlang(q,v)$, we have some
  transition $(q,v) \xrightarrow{\scriptnew{b}} (p',v'')$ in $R(A)$
  such that $\eq{a}{w}=\eq{b}{w'}$ for some $w'\in\barlang(p',v'')$;
  if we cannot $\alpha$-equivalently rename the
  $\newletter b$-transition into an $\newletter a$-transition to
  obtain the left alternative in \eqref{eq:fsuba-trans-alternative},
  then $b\neq a\in\supp(v'')$ and hence $a\in\supp(v)$, so by
  construction of $R(A)$ we obtain the right alternative in
  \eqref{eq:fsuba-trans-alternative}. By induction on $w$, it follows
  that $\{\unbar(w)\mid w\in L_0(q_0,u)\} = D(\barlang(q_0,u))$, so
  that $L(A)=D(\barlang(R(A))$, as claimed. \qed

  \myparagraph{Relationship to Unambigous Register Automata}
  \noindent We have seen that the languages accepted by RNNA under the
  local freshness semantics RNNA are a proper subclass of languages
  accepted by RAs. An important subclass of RAs are the unambigous
  ones. An RA $A$ is called \emph{unambigous} if every input word has
  at most one accepting run in $A$. The class of languages accepted by
  unambigous RA lies strictly between those accepted by deterministic
  and by arbitrary (non-deterministic) RAs. Unambigous RA are closed
  under complement and have decidable universality, language
  containment, and equivalence problems. Moreover, a language is
  accepted by an unambigous RA iff both itself and its complement are
  accepted by an RA (these results were announced in~\cite[Theorems~12
  and~13]{Colcombet15}).

  It follows that unambigous RA are incomparable to RNNA under local
  freshness semantics. Indeed, as we mentioned in the introduction,
  the language
\begin{equation}\label{eq:ab}
  L = \{\,ab \mid a,b\in \Names, a \neq b\,\}
\end{equation}
is not accepted by any RNNA but can clearly be accepted by an
unambigous (even deterministic) RA with only one
register. 

Now consider the language $L_2$ of all words in which some letter occurs at
least twice, which is defined by the bar expression
\[
(\newletter a)^*\newletter a (\newletter b)^* a \newletter a
\]
and is therefore accepted by an RNNA (hence also by an RA).
The complement of $L_2$,
\begin{equation}\label{eq:comp}
  \{ a_1\cdots a_n \mid n \in \Nat, \text{$a_i$ pairwise distinct}\,\},
\end{equation}
is not accepted by any RA~\cite{BojanczykEA14}. If the conjectured closure under
complement holds (by~\cite[Theorem~13]{Colcombet15}), then $L_2$ can not be
accepted by an unambiguous RA. \footnote{Since no proof of the claimed
  closure has been provided or referenced in op.cit.~as of January 2021, we prove 
  the independence explicitly in Proposition~\ref{unambL2}.}

In the following, we prove explicitly that $L_2$ is not accepted by any
unambiguous RA. As a lemma, we use the following:

\begin{lem} \label{lemDegreeClique}
Given a number $r$ and a directed graph $G = (V, \to)$ such that $|V| > 2\cdot r
+ 1$ and every vertex in $V$ has out-degree at most $r$, there exist
two distinct vertices that are not connected by an edge (i.e. no edge
in either direction).
\end{lem}
\begin{proof}
  Let $n := |V|$. In order to connect every (distinct) pair of
  vertices (in at least one direction), we need at least
  \[
      \frac{n \cdot (n-1)}{2}
  \]
  edges. By the assumption that $n > 2\cdot r + 1$, we have:

  \[
      \frac{n \cdot (n-1)}{2}
      > \frac{n \cdot ((2\cdot r + 1) - 1)}{2}
      = \frac{n \cdot 2\cdot r}{2} = n \cdot r
  \]
  Due to the bound on the out-degree, we know that the graph has at most
  $n \cdot r$ edges, which is strictly less than then minimum number of edges
  required to connect every pair of morphisms. Hence, there is some pair
  of distinct vertices with no direct edge between them.
  \qed
\end{proof}
Now we are ready finish the independence proof:
\begin{proposition} \label{unambL2}
  There is no unambiguous RA that accepts the language
  \[
  L_2 = \{w\in \A^*\mid \text{some letter occurs at least twice in $w$}\}.
  \]
\end{proposition}
\begin{proof}[by contraposition]
  Let $A$ be a register automaton with $r$ registers that accepts $L_2$. Put $n
  := 2\cdot r + 2$ and consider the set $S$ of configurations that are reached
  by the prefix
  \[
      p := a_1\, a_2 \ldots a_n          \tag*{for $n$ fresh letters.}
  \]
  Assume further fresh letters $f_1, f_2, \ldots, f_n$, and by $w_i$ we denote the
  word
  \[
     w_i := f_1 \ldots f_{i-1} \,a_i\, f_{i+1} \ldots f_n 
     \tag*{for every $i \in \{1,\ldots ,n\}$.}
  \]
  Put $V := \{1,\ldots ,n\}$. For every $i \in V$, we have that the word
  \[
     p\,w_i
  \]
  is contained in the language $L_2$. Hence, for every $i \in V$, there is a
  configuration (i.e.~pair of a state and a register assignment) $c_i\in S$ such
  that $c_i$ accepts the word $w_i$. Fixing the witnesses $c_i$ for every $i\in
  V$, we define the following directed graph structure on $V$:
  \[
    i \to j     \qquad\text{iff}\qquad  a_j\text{ is in one of the registers of }c_i.
  \]
  Since $A$ has $r$ registers, this graph has an out-degree of at most
  $r$, and by the definition of $n$, we have that the graph $(V,\to)$ has $n >
  2\cdot r + 1$ vertices. By the above Lemma~\ref{lemDegreeClique}, there exist distinct $i$ and $j$
  such that $i\not\to j$ and $j\not\to i$. Define the word $w_{i,j}$ as $w_i$ with $f_j$
  replaced by $a_j$:
  \[
    w_{i,j}  := \begin{cases}
      f_1 \ldots f_{i-1} \,a_i\, f_{i+1} \ldots f_{j-1} \,a_j\, f_{j+1}\ldots f_n
      &\text{if $i < j$}
      \\
      f_1 \ldots f_{j-1} \,a_j\, f_{j+1} \ldots f_{i-1} \,a_i\, f_{i+1}\ldots f_n
      &\text{if $j < i$}
      \end{cases}
  \]
  Since $i \not\to j$, the letter $a_j$ is not in any register of the
  configuration $c_i$, and $f_j$ neither is by the above assumption of
  freshness. Hence, $c_i$ accepts $w_{i,j}$. Symmetrically, $j \not\to i$
  implies that $c_j$ accepts $w_{i,j}$ as well. Moreover, $c_i$ must have $a_i$ in
  one of its registers, because otherwise $c_i$ would accept the word
  \[
    f_1 \ldots f_n,
  \]
  contradicting the correctness of $A$ (the word $p \,f_1 \ldots f_n$ is
  clearly not in the language $L_2$). Since $c_i$ has $a_i$ in one of its
  registers and $c_j$ does not, we can conclude that $c_i$ and $c_j$ are
  distinct configurations. Both of them accept $w_{i,j}$, and so we have two
  runs of $A$ accepting the word
  \[
    p \, w_{i,j},
  \]
  so $A$ is not unambiguous.     
  \qed
\end{proof}

\myparagraph{Relationship to Alternating 1-Register Automata}
\noindent Another model with decidable language inclusion are
alternating 1-register automata (see e.g.~\cite{NevenEA04}). We shall
prove that they are incomparable to RNNA under local freshness
semantics. For the failure of the inclusion of alternating 1-register
automata into RNNAs, just recall that the language $L$ defined in
Equation~\eqref{eq:ab}, which is not accepted by any RNNA, can be
accepted by an alternating, even deterministic, 1-register automaton.

Alternating 1-register automata are fairly expressive, and in fact can
express the language~\eqref{eq:comp} as well as, maybe unexpectedly,
the language `the first and the second letter both appear again, not
necessarily in the same order'. The latter is due to the fact that the
place where the automaton needs to look for the second occurrence of
the second letter does not depend on where it finds the second
occurrence of the first letter, so it can run independent searches for
the second occurrences of the respective letters in two conjunctive
branches, each using only one register.

However, we shall show that the language that is given under local freshness
semantics by the regular bar expression
\[
\newletter a \newletter b (\newletter c)^*ab (\newletter c)^*,
\]
and can therefore be accepted by an RNNA, cannot be accepted by any
alternating 1-register automaton.

We first recall the definition of the latter as presented in
\cite{DemriLazic09}. We restrict to the case where the finite
component of the input alphabet (just called the alphabet
in~\cite{DemriLazic09}) is a singleton, matching our example.
\begin{defn}
    An \emph{alternating oneway 1-register automaton (A1-RA)} consists of
    \begin{enumerate}
    \item a finite set $Q$ of \emph{locations} with an initial location $q_0 \in Q$;
    \item an one-step transition function $\op{step}: Q\to \Delta(Q)$
      where \footnote{Branching on end-of-word as additionally
        foreseen in \cite{DemriLazic09} can be encoded using the
        Boolean connectives and $\X$.}
    \begin{align*}
        \Delta(Q) = \{\ &
             p_1 \branchleft \uparrow \branchright p_2,
             p_1 \vee p_2,
             p_1 \wedge p_2,
             \top,\bot,\downarrow p,\X p, \bar \X p
         \mid
            p,p_1,p_2 \in Q
        \};
    \end{align*}
  \item a height function $\gamma: Q \to \Nat$ such that whenever
    $\op{step}(q) \in \{p_1 \branchleft \uparrow \branchright p_2, p_1 \vee
    p_2, p_1 \wedge p_2, \downarrow p\}$ we have
    $\gamma(p), \gamma(p_1), \gamma(p_2) < \gamma(q)$.
  \end{enumerate}
\end{defn}

\begin{notn}
  For any set $A\subseteq \A$ and any A1-RA we write 
  \[
    Q_A = Q \times (A + \{\star\})
  \]
  for the set of \emph{configurations} $(q,a)$ where $c$ is a location
  and $a$ indicates content of the register, which is either empty or
  contains a letter from $A$. For the special case $A=\Names$,
  $Q_\Names$ is a nominal set, with the nominal structure determined
  by taking $Q$ to be discrete. We write
  $i_A: Q_A \hookrightarrow Q_\A$ for the inclusion map, and define
  $u_A: Q_\A \twoheadrightarrow Q_A$ by
  \[
    u_A(q,a) = \begin{cases}
      (q,a) & a \in A, \\
      (q, \star) & \text{otherwise.}
    \end{cases}
  \]
  (Of course, $i_A\circ u_A:Q_\Names\to Q_\Names$ in general fails to
  be equivariant.)
\end{notn}

\begin{rem}
  For $A \subseteq B \subseteq \A$ we have
  \begin{equation}\label{eq:AsubB}
    i_A = i_B\cdot u_B\cdot i_A 
    \quad\text{and}\quad 
    u_A = u_A\cdot i_B\cdot u_B.
  \end{equation}
\end{rem}
\noindent 
Input words from $\Names^*$ are provided on a read-only tape whose
head may move to the right or stay in its position in every
computation step. The automaton starts to run in the initial
configuration $(q_0,\star)$.  For an input symbol $i\in \A$, the
behaviour of a configuration $(q,r) \in Q_\A$ is determined by
$\op{step}(q)$ as follows:
\begin{itemize}
\item For $p_1\vee p_2$ (resp.~$p_1\wedge p_2$),
  branch disjunctively (resp.~conjunctively) into $p_1$ and $p_2$
  without moving the head and changing the register content.
\item For $p_1 \branchleft \uparrow \branchright p_2$,
  transfer to location $p_1$ if of $i = r$, and to $p_2$ otherwise. 
\item For $\top$ (resp.~$\bot$), accept (resp.~reject) instantly.
\item For $\downarrow p$, write $i$ into the register and
  transfer to $p$.
\item For $ \X p$ (resp.~$\bar\X p$), transfer to $p$ and move the head
  one step to the right to the next input symbol if there exists one;
  otherwise, reject (resp.~accept).
\end{itemize}
\noindent It should be clear how the above determines acceptance of
words by an A1-RA; see~\cite{DemriLazic09} for a more formal
definition.  We shall now explain how an A1-RA can be translated into
an infinite deterministic automaton. To this end we consider the free
Boolean algebra monad $\BoolCombis$ on $\Set$. Concretely,
$\BoolCombis$ assigns to a set $X$ the set of Boolean formulas
built from elements of $X$ using the binary operations $\vee$ and
$\wedge$, the unary operation $\neg$, and the constants
$\op{\top,\bot}$, modulo the axioms of Boolean algebras. We denote the
unit of this monad by $\eta: \Id \to \BoolCombis$. Each of its
components $\eta_X: X \to \BoolCombis X$ is the universal map of the
free Boolean algebra on $X$, which embeds generators $x \in X$ into
$\BoolCombis X$. We note that $\BoolCombis Q_\Names$ is a nominal set,
with the action of $G$ given by $\pi\cdot t=\BoolCombis\pi(t)$ for
$t\in\BoolCombis Q_\Names$ where we abuse $\pi$ to denote the
associated bijection $Q_\Names\to Q_\Names$.

\begin{definition}
  Given an A1-RA as above, we define an function
  $\delta: \A × Q_\A \to \BoolCombis Q_\A$ recursively by
\[
\delta\big(c,(q,r)\big) = \begin{cases}
    \delta(c,(p_1,r)) \diamond \delta (c,(p_2,r)
        & \text{if }\op{step}(q) = p_1 \diamond p_2, \diamond\in\{\vee,\wedge\}
\\
    \delta(c,(p_1,r))
        & \text{if }\op{step}(q) = p_1 \branchleft \uparrow \branchright p_2\text{ and }c = r
\\
    \delta(c,(p_2,r))
        & \text{if }\op{step}(q) = p_1 \branchleft \uparrow \branchright p_2\text{ and }c \neq r 
\\
    \op{step}(q) & \text{if }\op{step}(q) \in \{\top,\bot\}
\\
    \delta(c,(p, c)) & \text{if }\op{step}(q) = \ \downarrow p
\\
    (p,r) & \text{if }\op{step}(q) = \X p
\\
    \neg (p,r) & \text{if }\op{step}(q) = \bar \X p
\end{cases}
\]
This recursion terminates because the height of states on the
right-hand side is strictly smaller than $\gamma(q)$.

The \emph{restriction} of $\delta$ to $A \subseteq \A$ is defined by 
\[
  \delta_A = (\A \times Q_A 
  \xrightarrow{\A \times i_A} 
  \A \times Q_\A 
  \xrightarrow{\delta}
  \BoolCombis Q_\A
  \xrightarrow{\BoolCombis u_A}
  \BoolCombis Q_A).
\]
Note that $\delta_\A = \delta$.
\end{definition}
It is easy to see that $\delta$ is an equivariant function.

We now consider the curried version $t_A: Q_A \to (\BoolCombis Q_\A)^\A$
of $\delta_A$ and pair it with the constant map $o_A = \bot!: Q_A \to 2$,
where $2 = \{\bot, \top\}$. Since $2 \times (\BoolCombis Q_A)^\A$
carries the obvious componentwise structure of a Boolean algebra, we
can uniquely extend $\langle o_A , t_A\rangle: Q_A \to 2 \times
(\BoolCombis Q_A)^\A$ to a Boolean algebra morphism
\begin{equation}\label{eq:dfa}
  \BoolCombis Q_A \to 2 \times (\BoolCombis Q_A)^\A. 
\end{equation}
We write $\op{acc}_A: \BoolCombis Q_A \to 2$ for the left-hand
component\lsnote{Since Boolean combinations can contain negation, I
  don't think that $\op{acc}_A$ is constantly $\bot$} and
$\delta^\dagger_A: \A \times \BoolCombis Q_A \to \BoolCombis Q_A$ for
the uncurrying of the right-hand one.  As announced, $\op{acc}_A$ and
$\delta_A$ determine the final states and the next state function of a
deterministic automaton on $\BoolCombis Q_A$.

\begin{rem}
  The formation of~\eqref{eq:dfa} from $\langle o_A, t_A\rangle$ is an
  instance of the \emph{generalized powerset
    construction}~\cite{sbbr13}. Indeed,
  $(\BoolCombis, \langle o_A, t_A \rangle)$ is a coalgebra for the
  functor $FT$ on $\Set$, where $F = 2 \times (-)^\A$ is the type
  functor of deterministic automata considered as coalgebras and
  $T = \BoolCombis$ is the free boolean algebra monad. The functor $F$
  clearly lifts to the Eilenberg-Moore category of $\BoolCombis$
  (i.e.~the category of boolean algebras). Therefore, any coalgebra
  $X \to FTX$ uniquely extends to the coalgebra $TX \to FTX$ for the
  lifting of $F$ to boolean algebras, viz.~a deterministic automaton
  on the state set $TX$.
\end{rem}  

\takeout{ 
Note that $2 = \{\bot,\top\}$ has a canonical $\BoolCombis$-algebra structure $\op{val}:
\BoolCombis 2 \to 2$ given by evaluation.
The intuitive meaning of a leaf $\ell \in Q_\A = Q×(\A+1)$ in a Boolean expression $t\in
\BoolCombis Q_\A$ is ``go to the next input symbol and transfer to the
specified configuration''. According to the definition of $\X$, this only
succeeds if there is an input symbol, and we define acceptance of a
complex configuration (i.e.~something in $\BoolCombis Q_\A$) by:
\begin{equation}
    \op{acc}_A := \op{val}\cdot \BoolCombis(\bot!): \BoolCombis Q_A\to 2
    \label{eq:accDefinition}
\end{equation}
Define the restriction of $\delta$ to $A\subseteq \A$ by $\delta_A(c) = \BoolCombis
u_A\cdot \delta(c)\cdot i_A$.
We have
\[
\begin{tikzcd}
    \A×\BoolCombis  Q_A
    \rar{s_{\A, Q_A}}
    \arrow[shiftarr={yshift=6mm}]{rrr}{\delta^*_A :=}
    &
    \BoolCombis (\A× Q_A)
    \rar{\BoolCombis \delta_A}
    &
    \BoolCombis \BoolCombis( Q_A)
    \rar{\mu}
    &
    \BoolCombis  Q_A.
\end{tikzcd}
\]}

\noindent As usual we extend any $\delta^\dagger_A$ to words by
induction, which yields
\[
  \delta^*_A: \A^*×\BoolCombis  Q_A \to \BoolCombis Q_A.
\] 
From now on we shall abuse notation further and denote by $\delta_A$,
$\delta_A^\dagger$ and $\delta_A^*$ also their curried versions with
codomain $\BoolCombis Q_A^{\BoolCombis Q_A}$. 

It is now straightforward to work out the following lemma (note that
A1-RA as defined in~\cite{DemriLazic09} do not handle the empty word
at all):

\begin{lem}\label{lem:acc}
  A given A1-RA accepts a non-empty word $w \in \A^+$ iff 
  \[
    \begin{tikzcd}
      \op{acc}_\A\cdot \delta^*_\A(w)\cdot \eta(q_0,\star) = \top.
    \end{tikzcd}
  \]
\end{lem}

\begin{lem}\label{lem:freshIsFresh}
    For any $c,d\not \in A$, $\delta_A(c) = \delta_A(d)$.
\end{lem}
\begin{proof}
  Since $c,d\not\in A$, we have $u_A=u_A\circ \pi$ where
  $\pi:Q_\Names\to Q_\Names$ denotes the bijection associated to the
  transposition $(c\,d)$. Recall that the action of $G$ on
  $\BoolCombis Q_\Names$ yields $(c\,d)\cdot
  t=\BoolCombis\pi(t)$. Thus, for any $(q,r)\in Q_A$, we obtain
  \begin{align}
    & \delta_A(c,(q,r)) \tag*{}\\
    & = \BoolCombis u_A(\delta(c,(q,r)))\tag*{}\\
    & = \BoolCombis u_A\BoolCombis\pi(\delta(c,(q,r)))\tag*{}\\
    & = \BoolCombis u_A(\delta((c\,d)\cdot (c,(q,r)))
    && \tag{$\delta$ equivariant}\\
    & = \BoolCombis u_A(\delta(d,(q,r))) \tag*{}\\
    & = \delta_A(d,(q,r)).
    \tag*{\qed}
  \end{align}
%
%
%
%
\end{proof}
\begin{corollary}\label{cor:freshness}
  Let $w,v \in \Names^n$ be words that differ only in letters that
  appear only once in them, i.e.~for all $k < n$ if $w_k$
  (resp.~$v_k$) appears again in $w$ (resp.~$v$) then $w_k = v_k$. Let
  $A\subseteq \A$ be the letters that occur more than once in $w$
  and $v$. Then we have $\delta^*_A(w) = \delta^*_A(v)$.
\end{corollary}
\begin{proof}
    This follows from Lemma~\ref{lem:freshIsFresh} by induction on $n$.
\qed
\end{proof}
This means one can directly discard all those atoms of a word that appear only
once and so we only need to keep those atoms in the register that appear at
least twice. \smnote[inline]{It's unclear what this is supposed to
  mean precisely. I think you mean the following corollary; @Thorsten can you insert its proof?}

Before coming to the main result, we show that we preserve the acceptance when
restricting the register contents to those letters appearing twice in a word.

\begin{lem} \label{lem:dropFreshAtoms} For all
  $A\subseteq B\subseteq \A$, $c\in A\cup(\A\setminus B)$, the diagram
    \[
    \begin{tikzcd}
        \BoolCombis Q_B
        \arrow{d}[left]{\BoolCombis(u_A\cdot i_B)}
        \arrow{r}{\delta^\dagger_B(c)}
        & \BoolCombis Q_B
        \arrow{d}[right]{\BoolCombis(u_A\cdot i_B)}
        \\
        \BoolCombis Q_A
        \arrow{r}{\delta^\dagger_A(c)}
        & \BoolCombis Q_A
    \end{tikzcd}
    \]
    commutes.
\end{lem}
\begin{proof}
  By universality of $\eta_X: X\to \BoolCombis X$, it suffices to show
  that
\[
    \BoolCombis(u_A\cdot i_B)\cdot \delta_B(c,(q,r))
    = \delta_A(c) \cdot u_A\cdot i_B (q,r)
\]
for all $q\in Q$, $r\in B+\{\star\}$. By the definition of $\delta_A$,
$\delta_B$ and Equation~\eqref{eq:AsubB}, this reduces to the equation
\begin{equation*}
  \BoolCombis(u_A)(\delta(c,(q,r)))=\BoolCombis(u_A)(\delta(c,u_A(q,r))).
\end{equation*}
This is trivial in case $r\in A+\{\star\}$, so assume from now on that
$r\in B\setminus A$; it then remains to show that
\begin{equation}\label{eq:no-branch-on-empty}
  \BoolCombis(u_A)(\delta(c,(q,r)))=\BoolCombis(u_A)(\delta(c,(q,\star))).
\end{equation}
Intuitively, this means that the automaton model does not foresee
branching on emptyness of the register. We
prove~\eqref{eq:no-branch-on-empty} by induction on the height
$\gamma(q)$, distinguishing cases on the form of $\op{step}(q)$. The
Boolean cases are trivial. The remaining cases are as follows.
\begin{itemize}
\item For $\op{step}(q) = p_1 \branchleft \uparrow \branchright p_2$,
  note that our assumptions $c\in A\cup(\Names\setminus B)$ and
  $r\in B\setminus A$ imply that $c\neq r$, so we have
  \begin{align*}
     \BoolCombis u_A(\delta_B(c,(q,r))) 
    & = \BoolCombis u_A(\delta_B(c,(p_2,r)))\\
    & = \BoolCombis u_A(\delta_B(c,(p_2,\star))) 
      && \by{IH}\\
    & = \BoolCombis u_A(\delta_B(c,(q,\star))).
\end{align*}
\item For $\op{step}(q)=\mathord{\downarrow} p$, we have
\begin{align*}
  \BoolCombis u_A(\delta(c,(q,r)) & = \BoolCombis u_A(\delta(c,(p,c))) \\
  & = \BoolCombis u_A(\delta(c,(p,\star)))
    && \by{IH}\\
  & = \BoolCombis u_A(\delta(c,(q,\star)))  
\end{align*}
\item For $\op{step}(q) = \X p$, we have (using $r\notin A$)
  \begin{equation*}
    \BoolCombis u_A(\delta(c,(q,r))  = \BoolCombis u_A(p,r)
    = (p,\star)= \BoolCombis u_A(p,\star) = \BoolCombis u_A(\delta(c,(q,\star)).
  \end{equation*}
The case for $\op{step}(q) = \bar \X p$ is similar. \qed
\end{itemize}
\end{proof}
\begin{lem} \label{lem:onlyUsedAtoms}
    For $w\in \A^*$, $B\supseteq \supp(w)$, the diagram
    \[
    \begin{tikzcd}
        \BoolCombis Q_B
        \arrow[hook]{d}[left]{\BoolCombis i_B}
        \arrow{r}{\delta^*_B(w)}
        & \BoolCombis Q_B
        \arrow[hook]{d}[right]{\BoolCombis i_B}
        \\
        \BoolCombis Q_\A
        \arrow{r}{\delta^*_\A(w)}
        & \BoolCombis Q_\A
    \end{tikzcd}
    \]
    commutes.
\end{lem}
\begin{proof}
  Let $c \in B$, $q\in Q$, and $r\in B + \{\star\}$. Since
  $\delta_\Names$ is equivariant, we have
  $\supp(\delta_\Names(c,(q,r)))\subseteq B$, i.e.\ as a Boolean
  algebra term over variables from $Q_\Names$,
  $\delta_\Names(c,(q,r))$ depends only on the variables from
  $Q_B$. Therefore,
  \begin{equation}\label{eq:toshow}
    \delta_\A(c, (q,r)) = \BoolCombis i_B\cdot \BoolCombis u_B \cdot \delta_\A(c, (q,r)). 
  \end{equation}
  This implies that 
  \begin{align*}
    \delta_\A^\dagger(c) \cdot \BoolCombis i_B \cdot \eta_{Q_B} 
    & = 
    \delta_\A(c) \cdot i_B 
    \\
    & \overset{\mathclap{\eqref{eq:toshow}}}{=} 
    \BoolCombis i_B\cdot \BoolCombis u_B \cdot \delta_\A(c) \cdot i_B
    \\
    & = 
    \BoolCombis i_B\cdot \delta_B(c) 
    \\
    & = 
    \BoolCombis i_B \cdot \delta_B^\dagger(c) \cdot \eta_{Q_B}.
  \end{align*}
  Using the universal property of $\eta_{Q_B}$, we conclude that
  $\delta_\A^\dagger(c) \cdot \BoolCombis i_B = \BoolCombis i_B \cdot
  \delta_B^\dagger(c)$, from which the main claim follows by induction
  on $w$. \qed
\end{proof}
In the following, we see that for certain words $w$, $\delta^*_A(w)$
behaves similarly to $\delta^*_A(c)$ for atoms $c\in \A$.
\begin{lem} \label{lem:finiteA1RA}
    If $A\subseteq \A$ and $w\in \A^*$  contains only letters in
    $A$ more than once,
    \[
    \begin{tikzcd}[row sep =5mm, column sep = 8mm]
    \BoolCombis  Q_\A
    \arrow{r}{\delta_\A^*(w)}
    & \BoolCombis  Q_\A
    \arrow{d}{\BoolCombis u_{A}}
    \\
    \BoolCombis  Q_{A}
    \arrow[hook]{u}{\BoolCombis i_A}
    \arrow{r}[swap]{\delta_{A}^*(w)}
    & \BoolCombis  Q_{A}
    \end{tikzcd}
    \]
    commutes.
\end{lem}
\begin{proof}Induction over the word length.  For $w =\varepsilon$ the
  claim is clear because $\delta^*_\A(\varepsilon)$ is just
  the identity and because $u_A\cdot i_A = \id_{ Q_A}$.  For a word
  $wc$, $c\in \A$, $w\in \A^*$, $W=\supp(w)$, $c\not\in W\setminus A$
  we have the commutativity of the diagram below:
    \begin{equation}
    \begin{tikzcd}[row sep =5mm, column sep = 13mm,baseline=(bot.base)]
    \BoolCombis  Q_\A
    \arrow{r}{\delta_\A^*(w)}
      \descto[yshift=1mm]{dr}{\text{Lemma~\ref{lem:onlyUsedAtoms}}}
    & \BoolCombis  Q_\A
      \arrow{r}{\delta^\dagger_\A(c)}
      \descto[xshift=4mm]{dr}{\text{Def $\delta^\dagger_{A\cup W}$}}
    &
      \BoolCombis  Q_\A
      \arrow{d}{\BoolCombis u_{A\cup W}}
      \arrow[shiftarr={xshift=14mm}]{ddd}{\BoolCombis u_{A}}
                [pos=0.33,xshift=-1mm,left]{\text{\eqref{eq:AsubB}}}
    \\
    \BoolCombis  Q_{A\cup W}
      \arrow[hook]{u}{\BoolCombis i_{A\cup W}}
      \arrow{r}{\delta^*_{A\cup W}(w)}
      \descto[yshift=1mm]{dr}{\text{Lemma~\ref{lem:onlyUsedAtoms}}}
      \arrow[hook]{d}[sloped,above,near end,xshift=1.3mm,yshift=2mm]{\BoolCombis i_{A\cup W}}
    & \BoolCombis  Q_{A\cup W}
      \arrow[hook]{u}[swap]{\BoolCombis i_{A\cup W}}
      \descto{ddr}{\text{Lemma~\ref{lem:dropFreshAtoms}}}
      \arrow{r}{\delta^\dagger_{A\cup W}(c)}
      \arrow[hook]{d}{\BoolCombis i_{A\cup W}}
    & \BoolCombis  Q_{A\cup W}
      \arrow[hook]{d}{\BoolCombis i_{A\cup W}}
    \\
    |[xshift=5mm]|
    \BoolCombis Q_\A
      \arrow{r}{\delta_\A^*(w)}
    & \BoolCombis  Q_\A
      \arrow{d}{\BoolCombis u_{A}}
    & \BoolCombis  Q_\A
      \arrow{d}{\BoolCombis u_{A}}
    \\
    |[alias=bot]|
    \BoolCombis  Q_{A}
      \arrow{r}{\delta_{A}^*(w)}
      \arrow[hook,xshift=-2mm,bend left=20]{uu}[anchor=base,rotate=90,yshift=1.5mm]{
            \BoolCombis (u_{A\cup W}
            \cdot i_A)
      }[right,xshift=1mm]{\text{\eqref{eq:AsubB}}}
      \arrow[hook]{u}[yshift=-1mm,right]{\BoolCombis i_A}
      \arrow[shiftarr={xshift=-14mm},hook']{uuu}
                [left]{\BoolCombis i_{A}}
                [pos=0.67,xshift=1mm,right]{\text{\eqref{eq:AsubB}}}
      \descto[yshift=1mm]{ur}{\text{(IH)}}
    & \BoolCombis  Q_{A}
      \arrow{r}{\delta_{A}^*(c)}
    & \BoolCombis  Q_{A}
    \end{tikzcd}
    \tag*{\qed}
    \end{equation}
\end{proof}
\begin{lem} \label{lem:finiteA1RAInclusion}
    If $A\subseteq B\subseteq \A$ and $w\in \A^*$ contains only letters in
    $A$ more than once,
    then
    \[
    \begin{tikzcd}[row sep =5mm, column sep = 8mm]
    \BoolCombis  Q_B
    \arrow{r}{\delta_\A^*(w)}
    & \BoolCombis  Q_B
    \arrow{d}{\BoolCombis (u_A\cdot i_B)}
    \\
    \BoolCombis  Q_{A}
    \arrow[hook]{u}{\BoolCombis (u_B\cdot i_A)}
    \arrow{r}[swap]{\delta_{A}^*(w)}
    & \BoolCombis  Q_{A}
    \end{tikzcd}
    \]
    commutes.
\end{lem}
\begin{proof}
  We have the following commutative diagram:
    \begin{equation}
    \begin{tikzcd}[baseline=(bottomNode.base)]
        \BoolCombis Q_B
            \arrow{rrr}{\delta^*_B(w)}
            \arrow{dr}{\BoolCombis i_B}
        & \descto[yshift=1mm]{dr}{\text{Lemma~\ref{lem:finiteA1RA}}}
        &&
        \BoolCombis Q_B
            \arrow{dd}{\BoolCombis (u_A\cdot i_B)}
                [left,xshift=-5mm]{\text{\eqref{eq:AsubB}}}
    \\
    & \BoolCombis Q_\A
            \arrow{r}{\delta^*_\A(w)}
             \descto[yshift=1mm]{dr}{\text{Lemma~\ref{lem:finiteA1RA}}}
        & \BoolCombis Q_\A
            \arrow{ur}{\BoolCombis u_B}
            \arrow{dr}{\BoolCombis u_A}
    \\
        \BoolCombis Q_A
            \arrow{rrr}{\delta^*_A(w)}
            \arrow{ur}{\BoolCombis i_A}
            \arrow[hook]{uu}{\BoolCombis (u_B\cdot i_A)}
                [right,xshift=5mm]{\text{\eqref{eq:AsubB}}}
        &&{}&
        |[alias=bottomNode]|
        \BoolCombis Q_A.
    \end{tikzcd}
    \tag*{\qed}
    \end{equation}
\end{proof}

In the following we write parts of words as \nfresh{n}, denoting any
$w\in \A^n$ that is fresh for the register contents and all the other
characters of the word and consists of distinct
letters\lsnote{inserted that last bit}. In particular, $\delta_A(\nfresh{n})$
is well-defined, because $\delta_A(w_1) = \delta_A(w_2)$ for any those candidates
$w_1,w_2 \in \A^n$ by Corollary~\ref{cor:freshness}.
\begin{proposition}\label{prop:A1-RA}
  There is no A1-RA recognising the language
  $\newletter{a}\newletter{b}(\newletter{c})^*ab(\newletter{d})^*$
  under local freshness semantics, that is, the language of all words
  where the first two letters appear again later, in the same order
  and adjacent.
\end{proposition}
\begin{proof} Assume for a contradiction that an A1-RA
  $(Q,q_0,\op{step})$ recognizes the language in question, and
  construct $\delta$ as described above. Pick any $a\neq b\in \A$. Let
  $n$ be greater than the (finite) cardinality of
  $\BoolCombis(Q_{\{a,b\}})$, and put $\ell=n!$. By assumption,
  the automaton accepts
  \[
      ab\mathbf{x}\quad\text{with}\quad\mathbf{x} := \nfresh{n}ab\nfresh{n},
  \]
  keeping the concrete choice of $\mathbf{x}$ fixed for the rest of the prove.
  We show that the automaton also accepts
  \[
  ab\mathbf{y}\quad\text{with}\quad\mathbf{y} := \nfresh{n+1}b\nfresh{\ell-2}a\nfresh{n+1},
  \]
  which is clearly not an element of the language, again keeping $\mathbf{y}$
  fixed for the rest of the proof. For any $r\in\A$ we have that
  $\delta^*_{\{r\}}(\nfresh{n}) = \delta^*_{\{r\}}(\nfresh{n+\ell})$: since
  $n$ is greater than $|\BoolCombis(Q_{\{r\}})|$, any run on
  $\nfresh{n}$ in $\BoolCombis(Q_{\{r\}})$ goes through a loop, and by
  Corollary~\ref{cor:freshness}, the claim follows by iterating that
  loop, whose length divides $\ell$ by the choice of~$\ell$.  Since
  $a\neq b$, this implies that
  \begin{equation} \label{eq:XiffY}
    \hspace{0mm}
    \begin{aligned}
      \delta^*_{\{a\}}({\bf x}) =\ &\delta^*_{\{a\}}(\nfresh{n}ab\nfresh{n})
      \\ =\ &
      \delta^*_{\{a\}}(\nfresh{n+\ell}ab\nfresh{n})
      \\ =\ &
      \delta^*_{\{a\}}(\nfresh{n+\ell}a\nfresh{n+1})
      \\ =\ &
      \delta^*_{\{a\}}(\nfresh{n+1}b\nfresh{\ell-2}a\nfresh{n+1})
      \\ =\ &
      \delta^*_{\{a\}}({\bf y})
    \end{aligned}
     \hspace{3mm}
     \begin{aligned}
       \mathllap{\delta^*_{\{b\}}({\bf x}) =}\ & \delta^*_{\{b\}}(\nfresh{n}ab\nfresh{n})
       \\ =\ &
       \delta^*_{\{b\}}(\nfresh{n}ab\nfresh{n+\ell})
       \\ =\ &
       \delta^*_{\{b\}}(\nfresh{n+1}b\nfresh{n+\ell})
       \\ =\ &
       \mathrlap{
         \delta^*_{\{b\}}(\nfresh{n+1}b\nfresh{\ell-2}a\nfresh{n+1})
       }
       \\ =\ &
       \delta^*_{\{b\}}({\bf y}).
     \end{aligned}
  \end{equation}
  \twnote{I re-enabled the second equation because it's in total shorter but
  more readable. Furthermore, in the following diagram, (\ref{eq:XiffY}) does
  not only refer to $\delta^*_{\{a\}}(\mathbf{x})= \delta^*_{\{a\}}(\mathbf{x})$ but
  also to $\delta^*_{\{b\}}(\mathbf{x})= \delta^*_{\{b\}}(\mathbf{x})$.}
  For $r\in\{a,b\}$, the respective equality proves commutation of
  \[
  \begin{tikzcd}[column sep = 15mm]
    Q_{\{a,b\}}
    \arrow{r}{\eta_{Q_{\{a,b\}}}}
    & \BoolCombis Q_{\{a,b\}}
    \arrow{r}{\delta^*_{\{a,b\}}(\mathbf{x})}
    \descto[yshift=3mm]{dr}{\text{Lemma~\ref{lem:finiteA1RAInclusion}}}
    & \BoolCombis Q_{\{a,b\}}
    \arrow{d}[right]{\BoolCombis(u_{\{r\}}\cdot i_{\{a,b\}})}
    \arrow[bend left=30]{dr}[sloped,above]{\op{acc}_{\{a,b\}}}
    \\
    Q_{\{r\}}
    \arrow{r}{\eta_{Q_{\{r\}}}}
    \arrow[hook]{u}[left]{u_{\{a,b\}}\cdot i_{\{r\}}}
    \arrow[hook]{d}[left]{u_{\{a,b\}}\cdot i_{\{r\}}}
    & \BoolCombis Q_{\{r\}}
    \arrow[hook]{u}[left]{\BoolCombis(u_{\{a,b\}}\cdot i_{\{r\}})}
    \arrow[hook]{d}[left]{\BoolCombis(u_{\{a,b\}}\cdot i_{\{r\}})}
    \arrow[bend left =7,yshift=1mm]{r}{\delta^*_{\{r\}}(\mathbf{x})}
    \arrow[bend right=7,yshift=-1mm]{r}[below]{\delta^*_{\{r\}}(\mathbf{y})}
    \descto[fill=none]{r}{\text{\eqref{eq:XiffY}}}
    & \BoolCombis Q_{\{r\}}
    \arrow{r}{\op{acc}_{\{r\}}}
    & 2
    \\
    Q_{\{a,b\}}
    \arrow{r}[below]{\eta_{Q_{\{a,b\}}}}
    & \BoolCombis Q_{\{a,b\}}
    \arrow{r}[below]{\delta^*_{\{a,b\}}(\mathbf{y})}
    \descto[yshift=-3mm,fill=none]{ur}{\text{Lemma~\ref{lem:finiteA1RAInclusion}}}
    & \BoolCombis Q_{\{a,b\}}
        \arrow{u}[right]{\BoolCombis(u_{\{r\}}\cdot i_{\{a,b\}})}
        \arrow[bend right=30]{ur}[sloped,below]{\op{acc}_{\{a,b\}}}
    \end{tikzcd}
    \]
    Since $u_{\{a,b\}}\cdot i_{\{r\}}: Q_{\{r\}} \hookrightarrow
    Q_{\{a,b\}}$, $r = a, b$, are the inclusion maps, and since these
    are jointly surjective, we have
    \[
        \op{acc}_{\{a,b\}}\cdot \delta^*_{\{a,b\}}(\mathbf{x})\cdot \eta_{Q_{\{a,b\}}} =
        \op{acc}_{\{a,b\}}\cdot \delta^*_{\{a,b\}}(\mathbf{y})\cdot \eta_{Q_{\{a,b\}}}.
    \]
    \twnote{I re-included this again, because this is needed in my opinion. Lemma~\ref{lem:finiteA1RAInclusion} is not applicable in the above diagram if $\mathbf{x} = ab\nfresh{n}ab\nfresh{n}$.}
    and thus by the universality of $\eta_{Q_{\{a,b\}}}$ also
    \begin{equation}
        \op{acc}_{\{a,b\}}\cdot \delta^*_{\{a,b\}}(\mathbf{x}) =
        \op{acc}_{\{a,b\}}\cdot \delta^*_{\{a,b\}}(\mathbf{y})
        \label{eq:abAcceptsY}
    \end{equation}
    Hence, $ab\mathbf{y}$ is accepted, because $ab\mathbf{x}$ is:
    \begin{equation*}
    \begin{tikzcd}[column sep=12mm, row sep=0mm, baseline=(bot.base)]
      & \BoolCombis Q_\A
      \arrow{r}{\delta^*_\A(ab)}
      & \BoolCombis Q_\A
      \arrow{r}{\delta^*_\A(\mathbf{x})}
      \descto{dr}{\text{Lemma~\ref{lem:finiteA1RA}}}
      & \BoolCombis Q_\A
      \arrow{d}{\BoolCombis u_{\{a,b\}}}
      \arrow[bend left=30]{ddr}[sloped,above]{\op{acc}_\A}
      \\[4mm]
      &&&
      \BoolCombis Q_{\{a,b\}}
      \arrow{dr}[sloped,above]{\op{acc}_{\{a,b\}}}
      \descto{dd}{\text{\eqref{eq:abAcceptsY}}}
      \\
    1
        \arrow[bend left=30]{uur}[sloped,above]{\eta(q_0,\star)}
        \arrow[bend right=30]{ddr}[sloped,below]{\eta(q_0,\star)}
        \arrow{r}[above]{\eta(q_0,\star)}
    & \BoolCombis Q_{\{a,b\}}
        \arrow[hook]{uu}[sloped,above]{\BoolCombis i_{\{a,b\}}}
        \arrow[hook]{dd}[sloped,above,rotate=180]{\BoolCombis i_{\{a,b\}}}
        \arrow{r}{\delta^*_{\{a,b\}}}
        \descto[yshift=1mm]{uur}{\text{Lemma~\ref{lem:onlyUsedAtoms}}}
        \descto[yshift=1mm]{ddr}{\text{Lemma~\ref{lem:onlyUsedAtoms}}}
    & \BoolCombis Q_{\{a,b\}}
        \arrow[hook]{uu}[sloped,above]{\BoolCombis i_{\{a,b\}}}
        \arrow[hook]{dd}[sloped,above,rotate=180]{\BoolCombis i_{\{a,b\}}}
        \arrow{ur}[sloped,above]{\delta^*_{\{a,b\}}(\mathbf{x})}
        \arrow{dr}[sloped,below]{\delta^*_{\{a,b\}}(\mathbf{y})}
    && 2
    \\
    &&& \BoolCombis Q_{\{a,b\}}
        \arrow{ur}[sloped,below]{\op{acc}_{\{a,b\}}}
    \\[4mm]
    & \BoolCombis Q_\A
        \arrow{r}[below]{\delta^*_\A(ab)}
    & \BoolCombis Q_\A
        \arrow{r}[below]{\delta^*_\A(\mathbf{y})}
        \descto{ur}{\text{Lemma~\ref{lem:finiteA1RA}}}
    & |[alias=bot]|
    \BoolCombis Q_\A
        \arrow{u}[right]{\BoolCombis u_{\{a,b\}}}
        \arrow[bend right=30]{uur}[sloped,below]{\op{acc}_\A}
    \end{tikzcd}
    \tag*{\qed}
    \end{equation*}
\end{proof}

\myparagraph{Relationship to 2-Register Automata} Finally, we consider
(nondeterministic) RAs with at most $2$ registers, another class of
automata with decidable language inclusion (see Kaminsky and
Francez~\cite{KaminskiFrancez94}).  This class is also incomparable to
RNNA. Indeed, the language~\eqref{eq:ab} can be accepted even by a
one-register RA but not by an RNNA. To see that the reverse inclusion
also fails one considers the language `the frist three letters appear
again'. Clearly, this can be accepted by an RNNA, but not by any RA
with at most 2 registers. Informally, such an RA would have to store
the first three letters to compare each of them to their subsequent
letters in the given input word; this is impossible with only 2
registers. A formal argument is similar to (but simpler than) the one
given in Proposition~\ref{prop:A1-RA}; we leave the details to the
reader.

Note that essentially the same argument also shows RNNA to be
incomparable to RA with at most $k$ registers for any fixed $k$; to
see this consider the language `the first $k+1$ letters appear again'.

%

\subsection{Proofs and Lemmas for Section~\ref{sec:inclusion}}

\begin{proof}[Additional details for the proof of Theorem~\ref{thm:barnfa-expspace}]
  We have omitted the space analysis of the initialization step.  To
  initialize $\Xi$ we need to compute $N_2 = \supp(\barlang(s_2))$.
  This can be done in nondeterministic logspace: for every free
  transition $q\trans{a} q'$ in $A_2$, in order to decide whether or
  not $a \in N_{s_2}$, remove from the transition graph of $A_2$ all
  transitions with label $\newletter a$ and then check whether there
  exists a path from $s_2$ to a final state passing through the given
  transition.
\end{proof}

\myparagraph{Details for Remark~\ref{rem:nka-complexity}} The spines
of an NKA expression $r$ arise by $\alpha$-renaming and subsequent
deletion of some binders from expressions that consist of
subexpressions of $r$, prefixed by at most as many binders as occur
already in $r$; therefore, the degree of the RNNA formed by the
spines, and hence, by Theorem~\ref{thm:nom-bar} (and the fact that
the translation from bar NFA to regular bar expressions is polynomial
and preserves the degree), that of the arising regular bar expression,
is linear in the degree of~$r$ (specifically, at most twice as large).
\qed\medskip

\noindent We shortly write
$D(w)=D(\barlang(w))=\{\unbar(w')\mid w'\alphaeq w\}$ for
$w\in\bar\Names^*$.
\begin{lem}\label{lem:dataord}
  If $w\dataord w'$ then $D(w)\subseteq D(w')$.
\end{lem}
\begin{proof}
  Induction over $w$, with trivial base case. The only non-trivial
  case in the induction step is that $w= av$ and $w'=\newletter a v'$
  where $v\dataord v'$.  All bar strings that are $\alpha$-equivalent
  to $w$ have the form $au$ where $v\alphaeq u$; we have to show
  $\unbar(au)\in D(\newletter av')$. We have $\unbar(u)\in D(v)$, so $\unbar(u)\in D(v')$ by
  induction; that is, there exists $\bar v'\alphaeq v'$ such that
  $\unbar(\bar v')=\unbar(u)$. Then $\unbar(\newletter a\bar v')=\unbar(au)$ and
  $\newletter a\bar v'\alphaeq \newletter a v'$, so $au\in
  D(\newletter av')$.
\end{proof}
Lemma~\ref{lem:datalang} is immediate from the following:
\begin{lem}\label{lem:weak-inclusion}
  Let $L$ be a regular bar language, and let $w\in\bar\Names^*$. Then
  $\datalang(w)\subseteq \datalang(L)$ iff there exists $w'\datasup w$
  such that $[w']_\alpha\in L$.
\end{lem}

\begin{proof}
  \emph{`If':} If $[w']_\alpha\in L$ then
  $\datalang(w')\subseteq D(L)$, so $\datalang(w)\subseteq D(L)$ by
  Lemma~\ref{lem:dataord}.

  \emph{`Only if':} We generalize the claim to state that whenever
  \begin{equation*}
    D(w)\subseteq\bigcup_{i\in I} D(\barlang(q_i))
  \end{equation*}
  for states $q_i$ in a name-dropping RNNA $A$ and a finite index set
  $I$, then there exist $i$ and $w'\datasup w$ such that
  $[w']_\alpha\in \barlang(q_i)$.  

  We prove the generalized claim by induction over $w$. The base case
  is trivial.

  Induction step for words $aw$: Let
  $D(aw)\subseteq \bigcup_{i=1}^n D(\barlang(q_i))$. We prove below that
  \begin{equation}\label{eq:free-step-dlang}
    D(w)\subseteq \bigcup_{i\in I, q_i\trans{\alpha}q',\alpha\in\{a,\scriptnew{a}\}}D(\barlang(q')).
  \end{equation}
  Indeed, let $u \in D(w)$, i.e.~there exists $v \alphaeq w$ with
  $\unbar(v) = u$. Then $\unbar(av) = au$ and $av \alphaeq aw$ imply $au \in
  D(aw)$, so by assumption there exists $i \in \{1, \ldots, n\}$ such
  that $au \in D(\barlang(q_i))$, i.e.~$au = \unbar(\alpha \bar u)$
  for $\alpha \in \{a, \newletter a\}$ and $[\alpha\bar u]_\alpha
  \in \barlang(q_i)$. By Lemma~\ref{lem:name-drop-alpha}, $\alpha \bar
  u \in L_0(q_i)$. Therefore there exists a transition $q
  \trans{\alpha} q'$ and $\bar u \in L_0(q')$. We conclude that $u =
  \unbar(\bar u) \in D(\barlang(q'))$ as desired. 

  Now, by induction hypothesis, it follows from~\eqref{eq:free-step-dlang} that we have $i\in I$,
  $\alpha\in\{a,\newletter{a}\}$, $q_i\trans{\alpha}q'$, and
  $w'\datasup w$ such that $[w']_\alpha\in \barlang(q')$. Then
  $\alpha w'\datasup aw$ and $[\alpha w']_\alpha \in \barlang(q_i)$, as
  required.

  Induction step for words $\newletter{a}w$: Let
  $D(\newletter{a}w)\subseteq \bigcup_{i=1}^n D(\barlang(q_i))$. Notice
  that 
  \begin{equation*}
    D(\newletter{a}w)=\bigcup_{b=a\lor b\fresh [w]_\alpha}bD(\pi_{ab}\cdot w)
  \end{equation*}
  (where $\cdot$ denotes the permutation group action and
  $\pi_{ab}=(a\ b)$ the transposition of $a$ and $b$; also note that $b \fresh [w]_\alpha$ iff $b \not\in \FN(w)$). Now pick
  $b\in\Names$ such that $b\fresh[w]_\alpha$ and none of the $q_i$ has
  a $b$-transition (such a $b$ exists because the set of free
  transitions of each $q_i$ is finite, as $A$ is an RNNA). We prove below that 
  \begin{equation*}
    bD(\pi_{ab}\cdot w)\subseteq \bigcup_{i\in I,q_i\trans{\subscriptnew{b}}q'}bD(\barlang(q')),
  \end{equation*}
  and hence we have
  \begin{equation}\label{eq:bound-step-dlang}
    D(\pi_{ab}\cdot w)
    \subseteq \bigcup_{i\in I,q_i\trans{\subscriptnew{b}}q'}D(\barlang(q')),
  \end{equation}
  again a finite union. In order to see that the above inclusion
  holds, let $bu \in bD(\pi_{ab} \cdot w)$, i.e., we have
  $v \alphaeq \pi_{ab}\cdot w$ with $\unbar(v) = u$. Then
  $\newletter b v \alphaeq \newletter b (\pi_{ab}\cdot w) \alphaeq
  \newletter a w$ and $\unbar(\newletter bv) = bu$, which implies that
  $bu \in D(\newletter a w)$. By our assumption
  $D(\newletter{a}w)\subseteq \bigcup_{i=1}^n D(\barlang(q_i))$ we
  obtain $i \in \{1, \ldots, n\}$ such that $bu \in D(\barlang(q_i))$,
  i.e. $bu = \unbar(\beta \bar u)$ for $\beta \in \{b, \newletter b\}$
  and $[\beta \bar u]_\alpha \in \barlang(q_i)$. By
  Lemma~\ref{lem:name-drop-alpha}, we have $\beta\bar u \in L_0(q_i)$,
  and since $q_i$ has no $b$-transitions, we therefore know that
  $\beta = \newletter b$. Hence we have a transition $q_i
  \trans{\newletter b} q'$ and $\bar u \in L_0(q')$. It follows that
  $u = \unbar(\bar u) \in D(\barlang(q'))$, whence $bu \in
  dD(\barlang(q'))$ as desired. 

  Now, by induction hypothesis, we obtain
  from~\eqref{eq:bound-step-dlang} $i\in I$,
  $q_i\trans{\scriptnew{b}}q'$, and $w'\datasup \pi_{ab}\cdot w$ such
  that $[w']_\alpha\in \barlang(q')$. It follows that
  \begin{equation*}
    \newletter{b}w'\datasup  \newletter{b}(\pi_{ab}\cdot w)\quad\text{and}\quad [\newletter{b}w']_\alpha\in \barlang(q_i).
  \end{equation*}
  Now we have $a\fresh[\pi_{ab}\cdot w]_\alpha$ (because
  $b\fresh [w]_\alpha$)), and therefore $a\fresh[w']_\alpha$ because
  $\pi_{ab}\cdot w\dataord w'$; it follows that
  $\newletter a(\pi_{ab}\cdot w') \alphaeq \newletter b w'$. As
  $\dataord$ is clearly equivariant, we have
  $\pi_{ab}\cdot w'\datasup w$, so
  \begin{equation*}
    \newletter{a}(\pi_{ab}\cdot w')\datasup\newletter{a}w\quad\text{and}\quad [\/\newletter{a}(\pi_{ab}\cdot w')]_\alpha=[\newletter b w']_\alpha\in \barlang(q_i),
  \end{equation*}
  which proves the inductive claim.
\end{proof}

\end{document}